\documentclass[a4paper,english,11pt]{scrartcl}
\usepackage{authblk}
\pdfoutput=1
\usepackage{tikz}
\usepackage[utf8]{inputenc}
\usepackage[T1]{fontenc}

\newif\ifspringer\springerfalse

\newif\ifarxiv\arxivtrue

\ifarxiv

\usepackage{mathtools,adjustbox}				% bessere Variante von amsmath (das von diesem Paket geladen wird). Zahlreiche Mathesachen
\usepackage{amssymb}				% zahlreiche mathematische Symbole
\usepackage{mathrsfs, stmaryrd}		% Zusätzliche Mathesymbole
\usepackage{dsfont}

\usepackage{amsthm}			% Theorem- und Beweisumgebungen
\usetikzlibrary{chains,shapes.multipart}
\usetikzlibrary{shapes,calc}
\usetikzlibrary{automata,positioning}
\usetikzlibrary{arrows, decorations.markings}
\definecolor{myred}{RGB}{220,43,25}
\definecolor{mygreen}{RGB}{0,146,64}
\definecolor{myblue}{RGB}{0,143,224}

\usetikzlibrary{shapes}

\usepackage{pgfplots}

\tikzstyle{normalNodeS}=[circle, color=black!75!white, fill, draw, inner sep = 0.1em, minimum size = 1.5 em, scale=0.5]
\tikzstyle{labeledNodeS}=[circle, color=black!75!white, draw, inner sep = 0.1em, minimum size = 1.5em, scale=1.25]

\tikzstyle{normalEdgeF}=[line width=1.6pt, color=black!75!white, >=stealth]
\tikzstyle{holdoverEdge}=[normalEdgeF, deeporange, thick]

\tikzstyle{demandNode}=[inner sep=0em]
\tikzstyle{demandEdge}=[>=stealth, double, thick, transparent]

\tikzstyle{normalNode}=[circle, fill=black, draw, inner sep = 0.1em, minimum size = 1em, scale=0.75]
\tikzstyle{labeledNode}=[circle, draw, inner sep = 0.1em, minimum size = 1.5em, scale=0.75, fill=white]
\tikzstyle{normalEdge}=[very thick, >=stealth]
\tikzstyle{holdoverEdge}=[normalEdge, red]

\tikzstyle{demandNode}=[inner sep=0em]
\tikzstyle{demandEdge}=[>=stealth, double, thick]

\definecolor{thirdcolor}{RGB}{0,76,153}
\definecolor{graphcolor}{RGB}{70,70,70}

\tikzstyle{aElement}=[circle, draw, fill=white, inner sep=0.1em, minimum size=1.2em, draw]
\tikzstyle{aElementT}=[aElement, minimum size=2em]
\tikzstyle{clNode}=[circle, draw=graphcolor, fill=graphcolor, scale=0.6]
\tikzstyle{ifNode}=[rectangle, draw=graphcolor, fill=white, scale=0.9]
\tikzstyle{afNode}=[rectangle, draw=graphcolor, fill=maincolor, scale=0.9]
\tikzstyle{boxnode}=[rectangle split, draw, minimum width=.5cm]

\tikzstyle{namedVertex} = [shape=circle,draw, color=black,fill=white]
\tikzstyle{vertex} = [shape=circle,draw=black]
\tikzstyle{edge} = [draw,->,thick]
\tikzstyle{blueEdge} = [edge,color=blue]
\tikzstyle{redEdge} = [edge,color=red]

\usepackage[linesnumbered,ruled,nokwfunc,noend]{algorithm2e}
\DontPrintSemicolon

\usepackage{listings}
\usepackage{color}
\definecolor{purplekeywords}{rgb}{0.5,0,0.33}
\definecolor{greencomments}{rgb}{0,0.5,0}
\definecolor{bluestrings}{rgb}{0,0,1}
\definecolor{types}{rgb}{0.17,0.57,0.68}
\usepackage[bookmarks=false,colorlinks=true, linkcolor=blue, urlcolor=blue, citecolor=blue, breaklinks=true]{hyperref}
\usepackage{cleveref}			% Referenzen mit Name

\crefname{cons}{constraint}{constraints}
\Crefname{cons}{Constraint}{Constraints}
\creflabelformat{cons}{(#2#1#3)}

\theoremstyle{definition}
\newtheorem{defn}{Definition}[section]
\newtheorem{example}[defn]{Example}

\theoremstyle{plain}

\newtheorem{lemma}[defn]{Lemma}
\newtheorem{cor}[defn]{Corollary}
\newtheorem{theorem}[defn]{Theorem}

\newtheorem{claim}{Claim}
\newenvironment{proofClaim}[1][]{\ifthenelse{\equal{#1}{}}{\begin{proof}}{\begin{proof}[#1]}}{\end{proof}}

\theoremstyle{remark}
\newtheorem{remark}[defn]{Remark}
\newtheorem{obs}[defn]{Observation}

\usepackage{array}					% Verbesserte Implementierung der tabular und array-Umgebungen

\usepackage{braket}					% \Set{ ... | ... } und \set{ ... | ... }
\usepackage{bm}						% fette Symbole in Math-Umgebung (\bm)
\usepackage{nicefrac} 				% TODO: Schönere Brüche?

\allowdisplaybreaks					% erlaubt Seitenumbrüche in align*-Umgebung (aber nicht align, gather, ...)

\newcommand{\IN}{\mathbb{N}}	% Natürliche Zahlen
\newcommand{\IR}{\mathbb{R}}
\newcommand{\R}{\mathbb{R}}	% Reelle Zahlen
\newcommand{\BigO}{\mathcal{O}}

\newcommand{\abs}[1]{\left|#1\right|}

\newcommand{\noqed}{\let\qed\relax}

\usepackage{comment}			% Umgebungen für Textabschnitte die nur in bestimmten Versionen erscheinen
\usepackage{textcomp}			% Textsymbole

\usepackage{rotating}

\usepackage{setspace}

\usepackage[font={small,it}]{caption} % Kleinere Captions

\usepackage{floatrow}

\usepackage{transparent}
\usepackage{color}
\usepackage{graphicx}

\usepackage[all]{xy}
\usepackage{lmodern}
\usepackage{tabto}

\usepackage[babel]{csquotes}

\usepackage{ifthen}

\usepackage[ngerman, textsize=small, textwidth=4.5cm, color=darkgreen, final]{todonotes}
\include{xcolor}
\definecolor{darkgreen}{rgb}{0.2,0.8,0.55} % Darf erst nach dem Einbinden von todonotes erstellt werden (alternativ - tikz einbinden...)

\newcommand{\keywords}[1]{\relax}
\newcommand{\subclass}[1]{\relax}

\title{Dynamic Flows with Adaptive Route Choice}
\author{Lukas Graf}
\author{Tobias Harks\thanks{The research of the authors was funded by the Deutsche Forschungsgemeinschaft (DFG, German Research Foundation) - HA 8041/1-1 and HA 8041/4-1.}}
\affil{\small Augsburg University, Institute of Mathematics, 86135 Augsburg\\
	\href{mailto:lukas.graf@math.uni-augsburg.de}{\{\texttt{lukas.graf,tobias.harks\}@math.uni-augsburg.de}}}
\author{Leon Sering}
\affil{\small Technische Universität Berlin, Institute of Mathematics, 10587 Berlin\\ \href{mailto:sering@math.tu-berlin.de}{sering@math.tu-berlin.de}}

	\newcommand{\revised}[1]{#1}
\fi
\ifspringer
	\input{header-Springer}
	\newcommand{\revised}[1]{{\color{revColor}#1}}
\fi

\usepackage{enumitem}

\newcommand*\diff{\,\mathop{}\!\mathrm{d}} % the upright d for integrals
\DeclarePairedDelimiterX{\scalar}[2]{\langle}{\rangle}{#1, #2} % Scalar product
\newcommand{\VI}{\textrm{VI}} % Variational inequality
\newcommand{\A}{\mathcal{A}} % Mapping for Variational Inequality
\usepackage{dsfont} % fuer charakteristische 1
\renewcommand{\l}{\ell}

\DeclareMathOperator{\supp}{supp}

\newcommand{\tauMinDiff}{\tau_\Delta}

\newcommand{\beforeN}{<}
\newcommand{\beforeNq}{\leq}

\newcommand{\IRnn}{\revised{\IR_{>0}}}

\newcommand{\edgeLoad}[1][]{%
	\ifthenelse{\equal{#1}{}}
	{F^{\Delta}}
	{F^{\Delta}_{#1}}
}

\definecolor{revColor}{rgb}{0.75,0,0}

\begin{document}
\maketitle
\begin{abstract}
We study dynamic network flows and introduce a notion of \emph{instantaneous dynamic
equilibrium (IDE)} requiring that for any positive inflow 
into an edge, this edge must lie on a currently shortest path
towards the respective sink. We measure current shortest path length
by current waiting times in queues plus physical travel times.
As our main results, we show:
\begin{enumerate}
 \item existence and constructive computation of IDE flows for multi-source single-sink networks assuming constant network inflow rates,
\item finite termination of IDE flows  for multi-source single-sink networks assuming bounded and finitely lasting inflow rates,
\item the existence of IDE flows for multi-source multi-sink instances assuming general measurable network inflow rates,
\item the existence of a complex single-source multi-sink instance in which any IDE flow is caught in cycles and flow remains forever in the network.
\end{enumerate}
% 4-6
\subclass{05C21 \and 90C35 \and 65K10}% https://mathscinet.ams.org/mathscinet/msc/msc2010.html
\end{abstract}

%!TEX root = ../article.tex
\section{Introduction}

Dynamic network flows have been studied for decades
in the optimization and transportation literature, see the classical book of Ford and Fulkerson~\cite{Ford62} or the more recent surveys of Skutella~\cite{Skutella08}
and Peeta~\cite{Peeta01}.
A fundamental model describing the dynamic flow propagation process
is the so-called \emph{deterministic queue model}, see Vickrey~\cite{Vickrey69}.
Here, a directed graph $G=(V,E)$ is given, where edges
$e\in E$ are associated with a queue with positive rate capacity $\nu_e\in\IRnn$ and a physical transit time $\tau_e\in \IRnn$.
%There is a finite set of commodities $I=\{1,\dots,n\}$
%each associated with a source node $s_i\in V$ and an inflow rate function $u_i:[r_i,R_i)\rightarrow\R_{\geq 0}$  generating the ``volume'' of agents originating in $s_i$. We assume for the moment that all agents of all commodities
%want to travel to a common destination $t\in V$.  
%In the fluid queuing model,
If the total inflow into an edge  $e=vw\in E$ exceeds the rate capacity
$\nu_e$, a queue builds up and agents  need to 
wait in the queue before they are forwarded along the edge. The total travel time
along $e$  is thus composed of the waiting time spent in the queue plus the physical transit time $\tau_e$.
A schematic illustration of the inflow and outflow mechanics of an edge $e$ is given in \Cref{fig:queue}.

\begin{figure}[h!]
	\begin{center}
		\begin{adjustbox}{max width=0.7\textwidth}
			\begin{tikzpicture}[scale=1.2]

		\begin{scope}
		
		\node (s1) at (0, 0) {};
		\draw (s1) +(0, 0) node[labeledNode] (s1) {$v$};
		\node[rectangle, color=gray!50,draw, minimum width=3cm, minimum height=0.1cm,fill={gray!50}] at (1.59,0.2) {};
		\node[rectangle, color=gray!50,draw, minimum width=0.5cm, minimum height=1cm,fill={gray!50}] at (0.55,0.55) {};

		\draw (s1) +(4, 0) node[labeledNode] (s3) {$w$}
		edge[normalEdgeF, <-] node[below] {$\tau_e$}
		node[above] {} (s1)

		;
		;
		\node at (-0.75,0) [] (inflow) { inflow};
		\node at (1.25,1.25) [] (queue) { queue $q_e(\theta)$};
		\node at (0.5,1.0) [] (qu) { };
		\draw (s1) edge[thick,>=stealth,dashed,bend left=90,->] node [left] {$f_e^+(\theta)$} (qu);
		\node at (6,1) [] (outflow) { outflow};
		\node at (4,0.2) [] (out) { };
		\draw (out) edge[thick,>=stealth,dashed,bend left=60,->] node [left] {$f_e^-(\theta)$} (outflow);
		% \draw [->, line width=1pt] (-1,0) -- (-0.5,0);
		\draw [<->, line width=0.5pt] (2.95,0.1) -- (2.95,0.3);
		\node at (3.2,0.2) [] (nu) { $\nu_e$};
			\end{scope}
\end{tikzpicture}			
		\end{adjustbox}
	\end{center}
	\vspace{-0.5cm}
	\caption[format=hang]{An edge $e=vw$ with a nonempty queue at time $\theta$.}
	\label{fig:queue}
	
\end{figure}
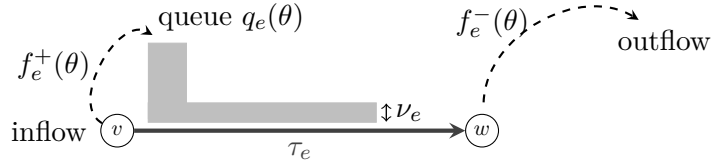

The fluid queue model has been mostly studied 
from a game-theoretic perspective, where
it is assumed that agents act selfishly and travel along shortest routes under prevailing  conditions.
This behavioral model is known as \emph{dynamic equilibrium} and has been analyzed in the transportation science literature for decades, see Friesz et al.~\cite{Friesz93}, Meunier and Wagner~\cite{MeunierW10} and Zhu and Marcotte~\cite{ZhuM00}.  In the past years, however, several new exciting developments have emerged: Koch and Skutella~\cite{Koch11}  elegantly characterized
dynamic equilibria by their derivatives, which gives a template
for their computation. Subsequently, Cominetti, Correa and Larr\'{e}~\cite{CominettiCL15}
derived alternative characterizations and
proved existence and uniqueness in terms of experienced travel times of equilibria  even for multi-commodity networks. Very recently, Cominetti, Correa and Olver~\cite{CominettiCO17} shed light on the  behavior of steady state queues
assuming single commodity networks and constant inflow rates. Sering and Vargas-Koch~\cite{Sering2019}
analyzed the impact of spillbacks in the fluid queuing model
and Bhaskar et al.~\cite{BhaskarFA15} devised Stackelberg strategies
in order to improve the efficiency of dynamic equilibria.

The concept `dynamic equilibrium'
assumes
\emph{complete knowledge and simultaneous route choice} by all
travelers.  Complete knowledge requires that
a traveler is able to exactly forecast future travel times
along the chosen path effectively anticipating the
whole evolution of the flow propagation process across the network.
This assumption has been justified by letting travelers
learn good routes over several trips and
a dynamic equilibrium then corresponds to an attractor of the underlying learning dynamic.
While certainly relevant, this concept may not  accurately reflect the behavioral
changes caused by the
wide-spread use of navigation devices.  As also discussed in Marcotte et al.~\cite{Marcotte04}, Hamdouch et al.~\cite{Hamdouch2004} and Unnikrishnan and Waller~\cite{UnnikrishnanW09}, drivers may not always learn
good routes over several trips but are now informed in
real-time about the current traffic situations and, if beneficial,
reroute instantaneously no matter how good or bad that route
was in hindsight. Also,  the information available to a navigation device 
is usually not complete, that is, congestion information is available
only as an aggregate (estimated waiting times for road traversal) but the individual
routes and/or source and destinations of travelers
are  unknown -- for good reason.\footnote{Some navigation systems have a large
	population of customers from which they infer even personalized
	information, but complete knowledge over all travelers seems unrealistic.}

In this paper, we consider
an adaptive route choice model, where at every node (intersection), travelers may
alter their route  depending on the current network conditions,
that is, based on current travel times and queuing delays.
%are continuously communicated.
%As a first basic variant, we assume that  travelers continuously
%receives information about the network-wide current queuing and travel times.
The needed information is anonymous and indeed available by navigation devices.
We assume that, if a traveler arrives at the end of an edge, she may change
the current route and opt for a currently shorter one. This type of reasoning does
neither rely on personalized information nor on the capability of unraveling the future flow propagation process.
We term a dynamic flow an \emph{instantaneous dynamic equilibrium (IDE)},
if for every point in time and every edge with positive inflow (of some commodity), this
edge lies on a currently shortest path towards the respective sink.
In the following, we illustrate IDE in comparison to classical dynamic equilibrium
with an example.

\subsection{An Example}\label{subsec:Example}

Consider the network in \Cref{fig:ide} (left). There are two source nodes $s_1$ and $s_2$ with constant inflow rates $u_1(\theta) \equiv 3$ for times $\theta \in [0,1)$ and $u_2(\theta)\equiv 4$ for $\theta\in [1,2)$. Commodity $1$ (red) has two simple paths connecting $s_1$ with the sink $t$. Since both have equal length ($\sum_e \tau_e =3$), in an IDE both can be used by commodity $1$. In \Cref{fig:ide},  the flow takes the direct edge to $t$ with a rate of one, while edge $s_1v$ is used at a rate of two. This is actually the only split possible in an IDE, since any other split (different in more than just a subset of measure zero of $[0,1)$) would result in a queue forming on one of the two edges, which would make the respective path longer than the other one.
At time $\theta=1$, the inflow at $s_1$ stops and a new inflow of commodity 2 (blue) at $s_2$ starts. This new flow again has two possible paths to $t$, however, here the direct path ($\sum_e \tau_e =1$) is shorter than the alternative ($\sum_e \tau_e =4$). So all flow enters edge $s_2t$ and starts to form a queue. 
At time $\theta= 2$, the first flow particles of commodity $1$ arrive at $s_2$ with a rate of $2$. Since the flow of commodity $2$ has built up a queue of length $3$ on edge $s_2t$ by this time, the estimated travel times $\sum_e (\tau_e + q_e(\theta))$ are the same on both simple $s_2$-$t$ paths. Thus, the red flow is split evenly between both possible paths. This results in the queue-length on edge $s_2t$ remaining constant and therefore this split gives us an IDE flow for the interval~$[2,3)$.
At time $\theta=3$, red particles will arrive at $s_1$ again, thus, completing a full cycle (namely $s_1,v,s_2,s_1$). This example shows that IDE flows may involve a flow decomposition along cycles.  In contrast, the (classical) dynamic equilibrium flow will just send more of the red flow along the direct path $s_1,t$ since the future queue growth at edge $s_2t$ of the alternative path is already anticipated.

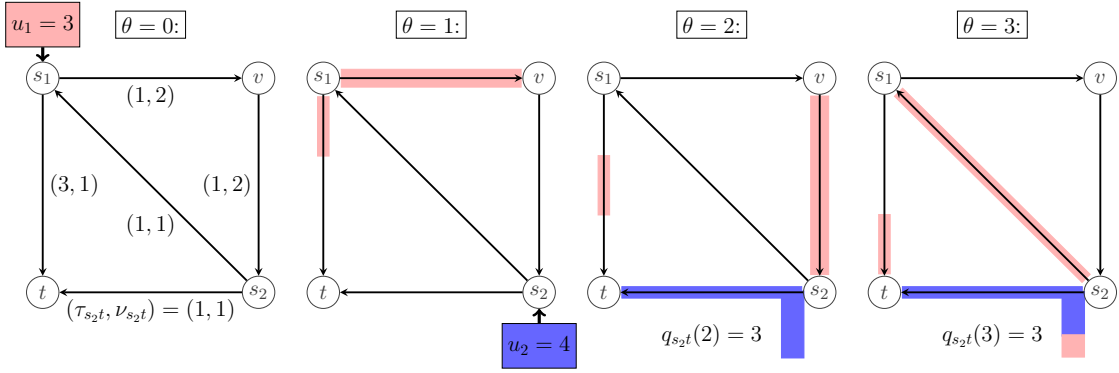
\begin{figure}[t!]
	\begin{center}
		\begin{adjustbox}{max width=1\textwidth}
			\begin{tikzpicture}[scale=1.2]
			
	\newcommand{\exampleGraph}[2]{		
		\node[labeledNodeS] (s1) at (0, 4) {$s_1$};
		\draw (s1) ++(4, 0) node[labeledNodeS] (v) {$v$};
		\draw (s1) ++(0, -4) node[labeledNodeS] (t) {$t$};
		\draw (v) ++(0, -4) node[labeledNodeS] (s2) {$s_2$};
		
		\draw[normalEdge,->] (s1) -- node[below] {\ifthenelse{\equal{#2}{labels}}{\Large$(1,2)$}{}} (v);
		\draw[normalEdge,->] (s1) -- node[right] {\ifthenelse{\equal{#2}{labels}}{\Large$(3,1)$}{}} (t);
		\draw[normalEdge,->] (v)  -- node[left] {\ifthenelse{\equal{#2}{labels}}{\Large$(1,2)$}{}} (s2);
		\draw[normalEdge,->] (s2) -- node[below] {\ifthenelse{\equal{#2}{labels}}{\Large$(\tau_{s_2t},\nu_{s_2t})=(1,1)$}{}} (t);
		\draw[normalEdge,->] (s2) -- node[below=.5cm] {\ifthenelse{\equal{#2}{labels}}{\Large$(1,1)$}{}} (s1);
		
		\node at (2,5) [rectangle,draw] (theta0) {\Large$\theta=#1$:};
	}
	
	\newcommand{\colComRed}{red!30}
	\newcommand{\colComBlue}{blue!60}
	
	% theta = 0 #########################################################################
	\begin{scope}
	
		\exampleGraph{0}{labels}
		
		\draw (s1) +(0,1) node[rectangle, draw, minimum width=1.5cm, minimum height=1cm,fill=\colComRed] (s1u) {\Large $u_1=3$}; 
		\draw [->, line width=2pt] (s1u) -- (s1);
		
		\end{scope}

		% theta = 1 #########################################################################
		\begin{scope}[xshift=5.2cm]
		
		\node[rectangle, color=\colComRed,draw, minimum width=1.33cm,minimum height=.2cm, fill=\colComRed,rotate around={90:(0,0)}] at (0,3.1) {};
		\node[rectangle, color=\colComRed,draw, minimum width=4cm,minimum height=.4cm, fill=\colComRed] at (2,4) {};
		
		\exampleGraph{1}{}	
		
		\draw (s2) ++(0,-1) node[rectangle, draw, minimum width=1.5cm, minimum height=1cm,fill=\colComBlue] (s2u) {\Large $u_2=4$}; 
		\draw [->, line width=2pt] (s2u) -- (s2);
	
	\end{scope}

	% theta = 2 #########################################################################
	\begin{scope}[xshift=10.4cm]
	
		\node[rectangle, color=\colComRed,  draw, minimum width=1.33cm,minimum height=.2cm, fill=\colComRed,rotate around={90:(0,0)}] at (0,2) {};
		\node[rectangle, color=\colComRed,  draw, minimum width=4cm,minimum height=.4cm, fill=\colComRed,rotate around={90:(0,0)}] at (4,2) {};
		\node[rectangle, color=\colComBlue, draw, minimum width=4cm,minimum height=.2cm, fill=\colComBlue] at (2,0) {};
		\node[rectangle, color=\colComBlue, draw, minimum width=.5cm,minimum height=1.5cm, fill=\colComBlue] at (3.5,-.6) {};
		\node at (2,-.8) [] () {\Large $q_{s_2t}(2)=3$};
		
		\exampleGraph{2}{}
	
	\end{scope}

	% theta = 3 #########################################################################
	\begin{scope}[xshift=15.6cm]
	
		\node[rectangle, color=\colComRed,draw, minimum width=1.33cm,minimum height=.2cm, fill=\colComRed,rotate around={90:(0,0)}] at (0,.9) {};
		\node[rectangle, color=\colComRed,draw, minimum width=5.9cm,minimum height=.2cm, fill=\colComRed,rotate around={-45:(0,0)}] at (2,2) {};
		\node[rectangle, color=\colComBlue,draw, minimum width=4cm,minimum height=.2cm, fill=\colComBlue] at (2,0) {};
		\node[rectangle, color=\colComBlue,draw, minimum width=.5cm,minimum height=1cm, fill=\colComBlue] at (3.5,-.4) {};
		\node[rectangle, color=\colComRed,draw, minimum width=.5cm,minimum height=.5cm, fill=\colComRed] at (3.5,-1) {};
		\node at (2,-.8) [] () {\Large $q_{s_2t}(3)=3$};
		
		\exampleGraph{3}{}
	
	\end{scope}		
\end{tikzpicture}
		\end{adjustbox}
	\end{center}
	\vspace{-0.5cm}
	\caption[format=hang]{The evolution of an IDE flow over the time horizon $[0,3]$.}
	\label{fig:ide}
\end{figure}

Note, that cycles can appear even in the case of only a single commodity (and therefore a single source and sink) and a constant inflow rate over a single interval - see \Cref{ex:IDEComputation} for such an instance. This shows that the differences between the two equilibrium concepts are quite fundamental and occur even in very simple examples.

\subsection{Related Work}

In the transportation science literature,  the idea of an instantaneous user or dynamic equilibrium has already been proposed since the late 80's, see 
Ran and Boyce~\cite[\S~VII-IX]{Ran96}, Boyce, Ran and LeBlanc~\cite{BoyceRL95,RanBL93}, Friesz et al.~\cite{FrieszLTW89}. These works develop an optimal control-theoretic formulation and
characterize instantaneous user equilibria by Pontryagin's optimality conditions.
However, not much is known regarding IDE existence and their structural properties.
In fact, the underlying equilibrium concept of Boyce, Ran and LeBlanc~\cite{BoyceRL95,RanBL93} and Friesz et al.~\cite{FrieszLTW89}
is different from ours. While the verbally written concept
of an IDE is similar to the one
we use here, the mathematical definition of an IDE in~\cite{BoyceRL95,FrieszLTW89,RanBL93} requires
that instantaneous travel times are minimal only for \emph{used paths} towards
the sink. A path is used, if every arc of the path has positive flow.
As, for instance, the authors in Boyce, Ran and LeBlanc~\cite[p.130]{BoyceRL95}
admit: ``Specifically with our definition of a used route, it is
possible that no route is ever `used' because vehicles stop entering the route before vehicles arrive at
the last link on the route. Thus, for some networks
every flow can be in equilibrium.'' Ran and Boyce~\cite[\S~VII, pp.148 ]{Ran96}
present a link-based definition of IDE. They define node labels at nodes $v\in V$
indicating the current shortest travel time from the source node
to some intermediate node $v$ and  require that whenever edge $vw$ has positive
flow, edge $vw$ must be contained in a shortest $s$-$w$ path, where $s$ is the flow's source. This is different
from our definition of an IDE, because we require that whenever there is positive
inflow into an edge $vw$, it must be contained in a currently shortest $v$-$t$ path,
where $t$ is the sink of the considered inflow.

Another important difference to our model is  
the assumed time horizon. The previous works~\cite{BoyceRL95,FrieszLTW89,Ran96,RanBL93} all assume a \emph{finite
	time horizon} on which the control problems are defined,
thus, only describing the flow trajectories over the given time
horizon. All numerical studies and simulation results
appearing in these works further implicitly assume that
for given finitely lasting bounded inflow rates, there exists a finite time horizon $[0,T]$
with $T$ large enough so that eventually all travelers reach their destination.
Our results reveal that this is in fact not true:
there are multi-commodity instances with finitely lasting bounded inflows that admit
IDE flows cycling forever. For the discrete version of this model using the natural discrete version of our equilibrium concept, such a behavior was already discovered in Ismaili \cite[Theorem 8]{ismaili2017}, though his instance makes critical use of edges with zero transit time (which we do not allow) and the instantaneous travel time always increases when there is positive inflow into an edge, even when the edge inflow rate is smaller than the rate capacity. In particular, players may observe increased instantaneous travel time, although no player is in fact delayed. This is not possible in our model. Based on this flexibility, the generated instance is considerably simpler than ours but one can show that IDE flows in our sense do terminate in that instance.

\subsection{Our Results}

In this paper we introduce the concept of an \emph{instantaneous dynamic equilibrium} (\emph{IDE} for short).
We call a feasible flow over time an IDE, if at any point in time, every edge with positive inflow (of some commodity) lies on a currently shortest path towards the respective sink.

Our first main result (Theorem~\ref{thm:existence}) shows that  IDE 
exist for  multi-source single-sink networks with piecewise constant inflow rates (generating
the volume of agents originating at the sources). The existence proof relies
on a constructive method extending any IDE flow up to time $\theta$
to an IDE flow on a strictly larger interval $\theta+\epsilon$ for some $\epsilon>0$.
The key insight for the extension procedure relies on solving a sequence of nonlinear programs,
each associated with finding the right outflow split for given node inflows. We also show that such solutions can be found by a simple water filling procedure.
With the  extension property we can apply a limit argument on the real numbers
implying the existence of IDE on the whole $\R_{\geq 0}$.

Given that, unlike the classical dynamic equilibrium,
IDE flows may involve cycling behavior (see the example in \Cref{fig:ide}), we turn
to the  issue of whether it is possible that positive flow volume remains forever in the network
(assuming finitely lasting bounded inflows). 
Our second main result (Theorem~\ref{thm:Termination_SingleSink}) shows that for multi-source single-sink networks, this
is impossible: Even for arbitrary bounded and finitely lasting inflow rate functions,
there always exists a finite time $T>0$
at which the network is cleared, that is, all flow particles have reached their destination. 

We then turn to general multi-commodity networks. For given piece-wise constant network inflow rates, we also can
extend an IDE flow to a strictly larger time interval by determining the edge inflow rates and the derivatives of the current shortest
path distances. This extension relies on a solution of a  system of equations, which is guaranteed to exist by Kakutani's fixed point theorem.
The equations are inspired by those defining a ``thin flows with resetting'' for Nash flows over time that were introduced by Koch and Skutella \cite{Koch11} and refined by Cominetti et al.\ \cite{cominetti2011existence,CominettiCL15}.
Again, a limit argument proves the existence of IDE flows over the whole $\R_{\geq 0}$
(\Cref{theorem:existence_multiThinFlow}). Furthermore, we give a brief idea, how these extensions can be
obtained by a mixed integer formulation. We then extend this result to arbitrary locally-integrable network inflow
rate functions by using a theorem about the existence of a solution to a variational inequality in
infinite-dimensional function spaces (\Cref{theorem:existence_arbitrary_inrates}).

 Finally, we show that for bounded 
and finitely lasting inflow rates, termination in finite time is \emph{not} guaranteed anymore as soon as the network contains more than one sink (\Cref{thm:NonTermination}). 
We construct a quite complex instance  where all IDE flows are caught in cycles and travel forever. 
This instance reveals that the assumption of a finite time horizon $[0,T]$
made previously in the transportation literature cannot be made 
without loss of generality. We also show that the instance can be modified in such a way that only a single-source is needed.

%!TEX root = ../article.tex
\section{The Flow Model}

In the following, we describe a  fluid queuing model as used before by Koch and Skutella~\cite{Koch11} and Cominetti, Correa and Larr\'{e}~\cite{CominettiCL15} and introduce the notation we will use throughout this paper.

We consider a finite directed graph\footnote{Note that all results of this paper also hold for multigraphs.} $G=(V,E)$ with positive rate capacities $\nu_e\in \IRnn$
and positive transit times $\tau_e\in \IRnn$ for every $e \in E$.\footnote{\revised{We exclude edges of zero travel time since -- intuitively -- our existence proofs all require that there is always some non-zero time between two decisions of a particle, see Remark~\ref{rem:zero} for a more detailed discussion.}} There is a finite set of commodities~$I=\{1,\dots,n\}$, each with a commodity-specific source node $s_i \in V$ and a commodity-specific sink node $t_i \in V$.
We will always assume that there is at least one $s_i$-$t_i$ path for each $i \in I$.
The (infinitesimally small) agents of every commodity $i\in I$ enter the network according to a integrable network inflow rate function $u_i:\IR_{\geq 0}\rightarrow\IR_{\geq 0}$. 

A \emph{flow over time} is a tuple $f = (f^+,f^-)$, where
$f^+, f^-: \R_{\geq 0} \times E \times I \to \R_{\geq 0}$ 
are integrable functions modeling the edge inflow rate $f^+_{i,e}(\theta)$ and edge outflow rate $f^-_{i,e}(\theta)$ 
of commodity $i$ of an edge $e\in E$ at time $\theta\geq 0$.

The \emph{queue length} of edge $e$ at time $\theta$ is given by
	\begin{align}	
		q_e(\theta) \coloneqq \sum_{i \in I}F^+_{i,e}(\theta) - \sum_{i \in I}F_{i,e}^-(\theta+\tau_e) & \quad \text{ for all } \theta \in \IR_{\geq 0} \label[cons]{eq:Cont-FlowDefProperties-QueueLengthWithF},
	\end{align}
where
	\[F_{i,e}^+(\theta):=\int_{0}^{\theta} f_{i,e}^+(z)dz \quad \text{ and } \quad F_{i,e}^-(\theta):=\int_{0}^{\theta}  f_{i,e}^-(z)dz\]
denote the \emph{cumulative (edge) inflow} and \emph{cumulative (edge) outflow}. We implicitly assume $f_{i,e}^-(\theta)=0$ for all $\theta \in [0,\tau_e)$, which will ensure together with \Cref{eq:Cont-FlowDefProperties-OpAtCap} (see below) that the queue lengths are always non-negative. Furthermore, we define the cumulative network inflow rate by $U_i(\theta):=\int_{0}^{\theta}u_i(z)dz$ and,
for the sake of simplicity, we denote the aggregated flow over all commodities by $f^+_e \coloneqq \sum_{i \in I}f^+_{i,e}$ and $f^-_e \coloneqq \sum_{i \in I}f^-_{i,e}$, as well as, $F^+_e \coloneqq \sum_{i \in I}F^+_{i,e}$ and $F^-_e \coloneqq \sum_{i \in I}F^-_{i,e}$.

A \emph{feasible} flow over time satisfies the following conditions \eqref{eq:Cont-FlowDefProperties-FlowCons}, \eqref{eq:Cont-FlowDefProperties-FlowConsSink}, \eqref{eq:Cont-FlowDefProperties-OpAtCap}, and \revised{\eqref{eq:FIFO}}.
The \emph{flow conservation constraints} are modeled for a commodity $i \in I$ and all nodes $v \neq t_i$ as
	\begin{align}\label[cons]{eq:Cont-FlowDefProperties-FlowCons} 
		\sum_{e \in \delta^+_v}f^+_{i,e}(\theta) - \sum_{e \in \delta^-_v}f^-_{i,e}(\theta) = \begin{cases}
			u_i(\theta), &\text{ if }v = s_i \\
			0,			 &\text{ if }v \neq s_i
		\end{cases} & \quad \text{ for all } \theta \in \IR_{\geq 0},
	\end{align}
where $\delta^+_v := \set{vu \in E}$ and $\delta^-_v := \set{uv \in E}$ are the sets of outgoing edges from $v$ and incoming edges into $v$, respectively. For the sink node $t_i$ of commodity $i$ we require
	\begin{align}\label[cons]{eq:Cont-FlowDefProperties-FlowConsSink} 
		\sum_{e \in \delta^+_{t_i}}f^+_{i,e}(\theta) - \sum_{e \in \delta^-_{t_i}}f^-_{i,e}(\theta) \leq 0 & \quad \text{ for all } \theta \in \IR_{\geq 0}.
	\end{align}
We assume that the queue operates at capacity which can be modeled by
	\begin{align}	\label[cons]{eq:Cont-FlowDefProperties-OpAtCap} 
		f_e^-(\theta + \tau_e) = \begin{cases}
		\nu_e, & \text{ if } q_e(\theta) > 0 \\
		\min\set {f^+_e(\theta), \nu_e}, & \text{ else }
		\end{cases}  & \quad \text{ for all }  e \in E, \theta \in \IR_{\geq 0}. 
	\end{align}
Since $q'_e(\theta)=\sum_{i \in I}f^+_{i,e}(\theta)-\sum_{i \in I}f^-_{i,e}(\theta+\tau_e)$, this condition is equivalent to the following equation describing  the queue length dynamics (cf. \cite[Section~2.2]{CominettiCL15}):
	\begin{equation}\label{eq:queue-dynamic}
		q'_e(\theta)=\begin{cases} f^+_e(\theta)-\nu_e,& \text{ if } q_e(\theta)>0\\
		\max\{0, f^+_e(\theta)-\nu_e\},& \text{ else } \end{cases}\quad \text{ for all }  e \in E, \theta \in \IR_{\geq 0}.
	\end{equation}

Finally we want the flow to follow a strict FIFO principle on the queues, which can be formalized by the following condition (see \cite{Sering2018}): 
	\begin{equation} \label[cons]{eq:FIFO}
		f^-_{i, e} (\theta) \revised{=} 
			\begin{cases}
				f^-_e (\theta) \cdot \frac{f^+_{i, e}(\vartheta)}{f^+_e(\vartheta)} & \text{ if }  f^+_e(\vartheta) > 0,\\
				 0 & \text{ else,}
			\end{cases}
	\end{equation}
where $\vartheta \coloneqq \min\set{\vartheta\leq \theta | \vartheta+\tau_e+\frac{q_e(\vartheta)}{\nu_e}=\theta}$ is the earliest point in time a particle can enter edge $e$ in order to leave it at time $\theta$. The quotient $\frac{q_e(\theta)}{\nu_e}$ is hereby the current waiting time to be spent in the queue of edge $e$. 
In other words, \Cref{eq:FIFO} ensures that the share of commodity $i$ of the aggregated outflow rate at some point in time $\theta$ equals the share of commodity $i$ of the aggregated inflow rate at the time the particles entered the edge. 

Note however, that all results within this paper also hold if we relax the FIFO condition \eqref{eq:FIFO} to the following condition
\begin{align}	\label[cons]{eq:Cont-FlowDefProperties-PlayerIdentityInQueues}
	F_{i,e}^-(\theta) \leq F_{i,e}^+(\theta-\tau_e)   \text{ for all } i \in I, e \in E, \theta \in \IR_{\geq 0}. 
\end{align}
This condition allows for overtaking within a queue and only prevents flow from changing its commodity on an edge by requiring the total amount of flow of every commodity that has left an edge to not exceed the total amount of flow of this commodity that has reached the head of this edge up to that point in time.

We assume that, whenever an agent arrives at an intermediate node $v$ at time $\theta$, she is given the information about the current queue lengths and transit times $q_e(\theta),\tau_e, e\in E$, and, based on this information, she computes a shortest $v$-$t$ path and enters the first edge on this path (breaking potential ties arbitrarily). We define the \emph{instantaneous travel time} of an edge $e$ at time $\theta$ as
	\begin{align}
		c_e(\theta) = \tau_e + \frac{q_e(\theta)}{\nu_e}.\label{eq:DefEdgeTravelTime}
	\end{align}
We can now define commodity-specific node labels $\ell_{i,v}(\theta)$ corresponding to current shortest path distances from $v$ to the sink $t_i$.
For $i\in I, v\in V$ and $\theta\in \R_{\geq 0}$, define
\begin{equation}\label{eq:current_shortest_path_distances}
\ell_{i,v}(\theta)\coloneqq\begin{cases} 0, & \text{ for } v=t_i\\
\min\limits_{e=vw\in E} \{\ell_{i,w}(\theta)+c_{e}(\theta)\}, & \text{ else.}\end{cases}
\end{equation}
We say that edge $e=vw$
is \emph{active} for $i\in I$ at time $\theta$,
if   $ \ell_{i,v}(\theta) = \ell_{i,w}(\theta)+c_{e}(\theta)$ and we denote the set of active edges for commodity $i$ by $E_\theta^i\subseteq E$. We call a $v$-$t_i$ path $P$ an \emph{active $v$-$t_i$ path for commodity $i$ at time $\theta$}, if all edges of $P$ are active for $i$ at $\theta$ or, equivalently, $\sum_{e \in P}c_e(\theta) = \ell_{i,v}(\theta)$. For differentiation we call paths that are minimal with respect to the transit times $\tau$ \emph{physical shortest paths}.

Now we are ready to formally define an instantaneous dynamic equilibrium for multi-commodity flows over time:
\begin{defn}
	A feasible flow over time $f$ is an \emph{instantaneous dynamic equilibrium (IDE)}, if
	for all  $i\in I, \theta\in \R_{\geq 0}$ and $e\in E$ it satisfies 
	\begin{align}\label[cons]{eq:Cont-FlowDefProperties-OnlyUseSP}
		f_{i,e}^+(\theta)>0 \Rightarrow e\in E_\theta^i.
	\end{align}
\end{defn}
In other words, a feasible flow over time $f$ is an IDE, if, whenever flow of commodity $i$ enters an edge $e=vw$ at some point $\theta$, this edge is contained in the set of active edges $E_\theta^i$, i.e., $e$ lies on a currently shortest path from $v$ to $t_i$.

Note that, while the set of active edges $E_{\theta}^i$ changes over time, the set of nodes, from which $t_i$ is reachable via $E_{\theta}^i$ (i.e. $\set{v \in V | \ell_{i,v}(\theta)<\infty}$) does not. In particular, this means that, whenever flow of commodity $i$ enters an active edge $uv, v \neq t_i$, at least one of the edges in $\delta_v^+$ will be active by the time the flow reaches $v$ -- although possibly different edges than those active when the flow left $u$.

%!TEX root = ../article.tex
\section{Existence and Computation of IDE Flows in Single-Sink Networks}\label{sec:ExistenceSingleSink}

In this and the following chapter we only consider single-sink networks, i.e., networks where all commodities have one common sink node $t$. In this case all commodities have the same label function, which we will denote by $\ell_v$. Since the origin of a particle in the network is not important for its route to the sink, we do not distinguish the commodities, and instead, only consider the aggregated flow functions $f^+_e$ and $f^-_e$. Furthermore, we restrict the network inflow functions $u_i$ to be right-constant, where a function \mbox{$u: [a,b) \to \IR$} is \emph{right-constant}, if for every $\theta \in [a,b)$ there exists an $\varepsilon > 0$ such that $u$ is constant on~$[\theta,\theta+\varepsilon)$.
%, i.e., for all $y \in [x,x+\varepsilon)$ we have $g(y) = g(x)$. 
For this case, we will now describe an algorithm computing an IDE flow.

Let $f=(f^+,f^-)$ denote a feasible flow over time. We denote by
\begin{equation} 
b^-_v(\theta):= \sum_{e\in\delta^-_v}f^-_e(\theta) + \sum_{i \in I: s_i=v} u_i(\theta)
\end{equation}
the current inflow at node $v$ at time $\theta$. Moreover, let $\delta_v^+(\theta):=\delta^+_v\cap E_\theta$ denote the set of outgoing edges of $v$ that are active at time $\theta$.
The main idea of our algorithm works as follows. Starting from time $\theta=0$ we compute inductively a sequence of intervals $[0,\theta_1),[\theta_1,\theta_2),\dots$ with $0<\theta_i<\theta_{i+1}$ and corresponding \emph{constant} edge inflows  $(f^+_e(\theta))_e$ for $\theta\in [\theta_i,\theta_{i+1})$ that form together with the corresponding edge outflows $(f^-_e(\theta))_e$ an IDE. Suppose we are given an \emph{IDE flow up to time $\theta_k$}, that is,  a tuple $(f^+,f^-)$
of right-constant functions $f_e^+:[0,\theta_k)\rightarrow \R_{\geq 0}$ and $f_e^-:[0,\theta_k+\tau_e)\rightarrow \R_{\geq 0}$ satisfying~\Cref{eq:Cont-FlowDefProperties-FlowCons,eq:Cont-FlowDefProperties-FlowConsSink,eq:Cont-FlowDefProperties-OpAtCap,eq:Cont-FlowDefProperties-OnlyUseSP}. Note that this is enough information to compute $F^+_e(\theta_k)$ and $F^-_e(\theta_k +\tau_e)$, and thus also $q_e(\theta_k), c_e(\theta_k)$ and $\ell_v(\theta_k)$ for all $e \in E$ and $v \in V$. 
We now describe how to extend this feasible flow over time to the interval $[\theta_k,\theta_k+\varepsilon)$ for some $\varepsilon>0$. The idea is that whenever there is positive inflow $b_v^-(\theta_k)>0$ into some node $v\in V$, we assign this inflow to outgoing edges that are currently active. Since the node labels at the heads of these edges depend themselves on queue dynamics at other nodes along a currently shortest path towards $t$, we need to handle \emph{time-varying} labels $\ell_w(\theta)$ when distributing the flow among the edges in $\delta_v^+(\theta)$. In the following, we describe how to define the flow-split in order to maintain the invariant that flow is only assigned to edges that are active for at least some interval even if adjacent labels vary linearly over time.

Assume that $b_v^-(\theta)$ is constant for $\theta\in [\theta_k,\theta_k+\varepsilon)$ for a node $v\in V$ and some $\varepsilon>0$. Moreover, let $\delta_v^+(\theta_k)=\{vw_1,vw_2,\dots,vw_{p_k}\}$ for some $p_k \geq 1$
and $[p_k]:=\{1,\dots,p_k\}$.
Thus, we have
\begin{align}\label{eq:InvariantAtTimeThetaK}
\ell_v(\theta_k)=c_{vw_i}(\theta_k)+\ell_{w_i}(\theta_k) \quad \text{ for all }i\in [p_k].
\end{align}
Suppose that the labels of nodes $w_i$ change linearly
after $\theta_k$, that is, there are constants $a_{w_i}\in \R$ for $i\in [p_k]$ with
\[ \ell_{w_i}(\theta)=\ell_{w_i}(\theta_k)+a_{w_i}(\theta - \theta_k) \quad  \text{ for all }\theta\in [\theta_k,\theta_k+\varepsilon).\]
Our goal is to find constant edge inflow rates $f^+_{vw_i}(\theta)$ during $[\theta_k,\theta_k+\varepsilon)$ satisfying the supply $b_v^-(\theta)$
and,  for some $\varepsilon' > 0$, fulfilling the following invariant for all $i\in [p_k]$ and $\theta\in [\theta_k,\theta_k+\varepsilon')$:
\begin{align}
f^+_{vw_i}(\theta)&>0\quad\Rightarrow\quad \ell_v(\theta)=c_{vw_i}(\theta)+\ell_{w_i}(\theta), \label{eq:invariant1}\\
f^+_{vw_i}(\theta)&=0\quad\Rightarrow\quad \ell_v(\theta)\leq c_{vw_i}(\theta)+\ell_{w_i}(\theta).\label{eq:invariant2} 
\end{align}
Note that a constant inflow rate $f^+_{vw_i}$ implies by \eqref{eq:queue-dynamic} that the queue length $q_{vw_i}$ is piecewise linear. Hence, the instantaneous travel time $c_{vw_i}$ is also piecewise linear on $[\theta_k,\theta_k+\varepsilon')$ for some $\varepsilon' > 0$, with derivative
\[c'_{vw_i}(\theta) = \frac{q' _{vw_i}(\theta)}{\nu_{vw_i}}.\]
%{\color{blue} \textbf{Alt:}
%	Since the invariant is fulfilled at $\theta=\theta_k$ (see \eqref{eq:InvariantAtTimeThetaK}) and the $\ell_{w_i}$ are assumed to be linear on the interval $[\theta_k,\theta_k+\varepsilon)$, a sufficient condition for
%	constant inflows to satisfy \eqref{eq:invariant1} and \eqref{eq:invariant2} for all
%	$\theta\in [\theta_k,\theta_k+\varepsilon')$
%	is the following: For all $i\in [p_k]$ the constant inflows satisfy at time $\theta_k$ 
%	\begin{align}\label[cons]{eq:diff-invariant1}
%	f^+_{vw_i}(\theta_k)&>0\quad\Rightarrow\quad \ell'_v(\theta_k)=c'_{vw_i}(\theta_k)+\ell'_{w_i}(\theta_k) \\\label[cons]{eq:diff-invariant2}
%	f^+_{vw_i}(\theta_k)&=0\quad\Rightarrow\quad \ell'_v(\theta_k)\leq c'_{vw_i}(\theta_k)+\ell'_{w_i}(\theta_k).
%	\end{align}}
	Since all edges $vw_i$ are active at time $\theta_k$, we have $\ell_v(\theta_k) = c_{vw_i}(\theta_k) + \ell_{w_i}(\theta_k)$ and, thus, a flow with constant inflow rates satisfies \eqref{eq:invariant1} and \eqref{eq:invariant2} for all
	$\theta\in [\theta_k,\theta_k+\varepsilon')$, if
	\begin{align}\label[cons]{eq:diff-invariant1}
	f^+_{vw_i}(\theta_k)&>0\quad\Rightarrow\quad \ell'_v(\theta_k)=c'_{vw_i}(\theta_k)+\ell'_{w_i}(\theta_k) \\\label[cons]{eq:diff-invariant2}
	f^+_{vw_i}(\theta_k)&=0\quad\Rightarrow\quad \ell'_v(\theta_k)\leq c'_{vw_i}(\theta_k)+\ell'_{w_i}(\theta_k).
	\end{align}
This condition ensures that whenever an edge $vw_i$ has positive
inflow, the remaining distance towards $t$ grows from $\theta_k$ onwards at the lowest speed.

For a given value $b_v^-(\theta_k)$ and given vector $(a_{w_i})_{i \in [p_k]}$ we consider the following optimization problem in variables $x_{vw_i}$ for $i\in [p_k]$ in order to obtain inflow rates that
satisfy the conditions~\eqref{eq:diff-invariant1} and~\eqref{eq:diff-invariant2}. 
\begin{align}
\min\quad&\sum_{i=1}^{p_k} \int_0^{x_{vw_i}}\frac{g_{vw_i}(z)}{\nu_{vw_i}}+a_{w_i} dz \label{OPT:bvThetaK}\tag{OPT-$b_v^-(\theta_k)$}\\
\text{s.t.}\quad&\label{eq:lagrange}  \sum_{i=1}^{p_k} x_{vw_i}=b_v^-(\theta_k)\\\notag
& x_{vw_i}\geq 0 \text{ for all } i\in [p_k].\notag
\end{align}
Each function $g_{vw_i}$ maps an edge inflow rate $z$ to the change of the queue size if a constant flow with rate $z$ would enter the edge $vw_i$, i.e.,
\begin{equation} g_{vw_i}(z):=\begin{cases}z-\nu_{vw_i}& \text{ if } q_{vw_i}(\theta_k)>0,\\
\max\set{z-\nu_{vw_i},0}& \text{ if } q_{vw_i}(\theta_k)=0.\end{cases}
\end{equation}
Hence, $g_{vw_i}(f_{vw_i}^+(\theta_k))$ is the derivative of $q_{vw_i}$ at $\theta_k$ (cf. \Cref{eq:queue-dynamic}). 
%The following lemma states the existence of an optimal solution to \eqref{OPT:bvThetaK} and its feasibility to \Cref{eq:diff-invariant1,eq:diff-invariant2}.
\begin{lemma}\label{lemma:ExistenceOfOptSolutionBv}
	There exists an optimal solution $(x_{vw_i})_{i\in [p_k]}$  to \eqref{OPT:bvThetaK}
	and for every optimal solution
	$f^+_{vw_i}(\theta_k)\coloneqq x_{vw_i}$ satisfies~\eqref{eq:diff-invariant1} and~\eqref{eq:diff-invariant2} for all $i\in [p_k]$.
\end{lemma}
\begin{proof}
	The objective function is continuous and
	the feasible region is non-empty and compact.
	Hence, by the theorem of Weierstra\ss{}
	at least one optimal solution exists.
	Moreover, the objective is differentiable over the feasible 
	domain, thus, first order optimality conditions
	hold. Assigning a multiplier $\lambda\in \R$ to \eqref{eq:lagrange}
	and taking partial derivatives of the Lagrangian over the positive orthant,
	we obtain
	\begin{align} x_{vw_i}&>0\quad\Rightarrow\quad \frac{g _{vw_i}(x_{vw_i})}{\nu_{vw_i}}+a_{w_i} +\lambda= 0\\
	x_{vw_i}&=0\quad\Rightarrow\quad \frac{g _{vw_i}(x_{vw_i})}{\nu_{vw_i}}+a_{w_i} +\lambda\geq 0.
	\end{align}
	These conditions imply~\eqref{eq:diff-invariant1} and~\eqref{eq:diff-invariant2} with $\ell'_v(\theta_k) := -\lambda$.
\end{proof}

\begin{lemma}\label{lemma:ExtendingContIUEFlows}
	Let $f=(f^+,f^-)$ be an IDE flow up to time $\theta_k\geq 0$ and suppose there are constant inflow rate functions $b_v^-:[\theta_k,\theta_k+\varepsilon)\rightarrow \R_{\geq 0}$ for some $\varepsilon>0$ and all nodes $v\in V$ (in particular, this means $\varepsilon \leq \min\set{\tau_e | e \in E}$). 	
	Then, there exists some $\varepsilon' > 0$ such that $f$ can be extended to an IDE flow up to time $\theta_k+\varepsilon'$ with all functions $f^+_e$ constant on the interval $[\theta_k,\theta_k+\varepsilon')$ and all functions $f^-_e$ right-constant on the intervals $[\theta_k+\tau_e,\theta_k+\tau_e+\varepsilon')$.
\end{lemma}
\begin{proof}
	First, it is possible to determine the queue lengths at time $\theta_k$ using \Cref{eq:Cont-FlowDefProperties-QueueLengthWithF} and from those the labels $\ell_v(\theta_k)$ can be obtained. 
  Applying \Cref{lemma:ExistenceOfOptSolutionBv} on the nodes in order of increasing $\ell_v(\theta_k)$ values, we obtain the outflow rates and, therefore, the slope $a_v$ of label~$\ell_v$ for some interval right after $\theta_k$.
%  Next, we sort the nodes by those labels $\ell_v(\theta_k)$ and will now define the outflows using \Cref{lemma:ExistenceOfOptSolutionBv} for each node, beginning with the one with the smallest label.	
	More precisely, we start with $t$ (since $\ell_t(\theta_k)=0$) for which we can define $f^+_e(\theta)=f^-_e(\theta+\tau_e)=0$ for all outgoing edges $e \in \delta^+_t$ and all times $\theta \in [\theta_k,\theta_k+\varepsilon')$, where $\varepsilon' \coloneqq \varepsilon$. Now we take some node $v$ such that there exists an $\varepsilon' > 0$ and for all nodes $w$ with strictly smaller label at time $\theta_k$ and all edges $e \in \delta^+_w$, we have already defined $f^+_e$ on some interval $[\theta_k,\theta_k+\varepsilon')$ and $f^-_e$ on some interval $[\theta_k+\tau_e,\theta_k+\tau_e+\varepsilon')$ in such a way that on the interval $[\theta_k,\theta_k+\varepsilon')$
	\begin{enumerate}
		\item\label{enum:1} the labels $\ell_w(\theta)$ change linearly with slope $a_w$, 
		\item no additional edges are added to the sets $\delta^+_w(\theta)$ of active edges leaving $w$,
		\item the functions $f^+_e$ and $f^-_e$ for $e\in \delta^+_w$ are constant and right-constant, respectively, and
		\item\label{enum:4} the functions $f^+_e$ and $f^-_e$ for $e\in \delta^+_w$ satisfy \Cref{eq:Cont-FlowDefProperties-FlowCons,eq:Cont-FlowDefProperties-OpAtCap,eq:Cont-FlowDefProperties-OnlyUseSP}.
	\end{enumerate}
	Let $\delta_v^+(\theta_k):=\{vw_1,vw_2,\dots,vw_{p_k}\}$ be the set of active edges at $v$ at time $\theta_k$. Then, at time $\theta_k$, each $w_i$ must have a strictly smaller label than $v$. Hence,
	they satisfy Properties~\ref{enum:1}.-\ref{enum:4}.
	We can now apply \Cref{lemma:ExistenceOfOptSolutionBv} to determine the flows $f^+_{vw_i}(\theta_k)$. Additionally, we set $f^+_e(\theta_k)=0$ for all non-active edges leaving $v$.	
	Assuming that this flow remains constant on the whole interval $[\theta_k,\theta_k+\varepsilon')$, we can determine the first time $\hat{\theta} \geq \theta_k$, where an additional edge $vw \in \delta_v^+$ or $wv \in \delta_w^+$ with $w \in V$ a node with already defined label $\ell_w$ for the coming time interval becomes newly active. This can only happen after some positive amount of time has passed, i.e., for some $\hat{\theta} > \theta_k$, because
	\begin{itemize}
		\item at time $\theta_k$ the edge was non-active, and therefore $\ell_v(\theta_k) > c_{vw}(\theta_k) + \ell_w(\theta_k)$ or $\ell_w(\theta_k) > c_{wv}(\theta_k) + \ell_v(\theta_k)$, respectively,
		\item all labels change linearly (and thus continuously) and
		\item $c_{vw}$ or $c_{wv}$ is changing piecewise linearly, since the length of its queue does so as well (as both $f^+_{vw}$ and $f^-_{wv}$ are piecewise constant).
	\end{itemize}
	If the difference $\hat{\theta}-\theta_k$ is smaller than the current $\varepsilon'$, we take it as our new $\varepsilon'$, otherwise we keep it as it is. In both cases, we extend $f^+_e$ onto the interval $[\theta_k,\theta_k+\varepsilon')$ for all $e \in \delta_v^+$ by setting $f^+_e(\theta) = f^+_e(\theta_k)$ for all $\theta \in [\theta_k,\theta_k+\varepsilon')$. 	
	This guarantees that the label of $v$ changes linearly on this interval, no additional edges become active and the functions $f^+_e$ are constant. Also, $f^+_e$ satisfies \Cref{eq:Cont-FlowDefProperties-FlowCons,eq:Cont-FlowDefProperties-OnlyUseSP} by definition. Finally, we define $f^-_e$ as follows:
	\[f^-_e(\theta+\tau_e) := \begin{cases}
	\nu_e,			& \text{if } q_e(\theta_k) + (\theta-\theta_k)(f^+_e(\theta_k)-\nu_e) > 0, \\
	f^+_e(\theta), 	& \text{else.}
	\end{cases}\] 
	Then, $f^-_e$ is right-constant and together with $f^+_e$ satisfies \Cref{eq:Cont-FlowDefProperties-OpAtCap}.	
	In summary, using this procedure we can extend $f$ node by node to an IDE flow up to $\theta_k+\varepsilon'$ for some $\varepsilon' > 0$.
\end{proof}
\begin{remark}\label{rem:zero}
The above lemma is the key for the following existence result
of IDE flows. For the extension  property, we used that travel times 
are strictly positive, because otherwise, for a given IDE flow
up to time $\theta_k$, the edge outflow rates are not well-defined
for any proper interval after $\theta_k$ assuming zero travel times.
Existence of IDE flows is not ruled out for $\tau_e \in \IR_{\geq 0}$, but this would require the use of a more intricate way of describing the extensions or
a completely different approach.
\end{remark}

\begin{theorem}\label{thm:existence}
	For any multi-source single-sink network with right-constant network inflow rate functions, there exists an IDE flow $f$ with right-constant functions $f^+_e$ and $f^-_e, e\in E$.
\end{theorem}

\begin{proof}
	Let $\mathfrak{F}_0$ be the set of tuples $(f,\theta)$, with $\theta \in \R_{\geq 0}\cup\set{\infty}$ and $f$ an IDE flow up to time $\theta$ with right-constant functions $f^+_e$ and $f^-_e$. Define $\hat\theta_0 := \sup\set{\theta | \exists f \text{ s.t. } (f, \theta) \in \mathfrak{F}_0}$. If $\hat\theta_0 = \infty$ we are done, so suppose $\hat\theta_0 < \infty$. 
	There exists an IDE flow $f_1$ with $\theta_1 := \hat\theta_0/2$ such that $(f_1,\theta_1) \in \mathfrak{F}_0$. Now we define 
	\[\mathfrak{F}_1 := \set{(f, \theta) \in  \mathfrak{F}_0 | f\big|_{[0,\theta_1)} = f_1}.\]
	This set is not empty so we set $\hat\theta_1 := \sup\set{\theta | \exists f \text{ s.t. } (f, \theta) \in \mathfrak{F}_1}$. By \Cref{lemma:ExtendingContIUEFlows} we know that $\hat\theta_1 > \theta_1$, and therefore $\hat\theta_1 \in (\theta_1, \hat\theta_0]$.
	Let $\theta_2 := (\hat\theta_1 - \theta_1)/2$.
	Going on we get a strict monotone increasing sequence $(\theta_i)_{i \in \IN}$ and a non-increasing sequence $(\hat\theta_i)_{i \in \IN}$ with $\theta_i < \hat\theta_i$ for all $i \in \IN$ and $\hat\theta_i - \theta_i \leq \hat\theta_0/2^i \to 0$ for $i \to \infty$. Let $\theta^*$ be the limit of these two sequences. By taking pointwise limits of the sequence $(f_i)_{i \in \IN}$ we can construct a flow $f^*$ such that $(f^*,\theta^*) \in \mathfrak{F}_0$. 
	By \Cref{lemma:ExtendingContIUEFlows} we can extend $f^*$ by some $\varepsilon$ but this is a contradiction to the definition of $\hat\theta_i$ for all $i$ with $\hat\theta_i \in [\theta^*, \theta^* + \varepsilon)$. Hence, $\hat\theta_0 = \infty$, which finishes the proof.
\end{proof}

\begin{remark}
	While there always exists an IDE flow, this flow does not have to be unique. In fact, neither the flow $f$ itself nor the label functions $\ell_v$ and the time of termination need to be unique. This is in contrast to dynamic equilibria, where at least the label functions are uniquely determined (see Cominetti et al.~\cite[Theorem 6]{CominettiCL15}). An example for non-uniqueness is illustrated by the instance in \Cref{fig:NonUniquenessOfIDEFlows}.
	\begin{figure}[h!] 
		\floatbox[{\capbeside\thisfloatsetup{capbesideposition={right,top},capbesidewidth=.55\textwidth}}]{figure}[\FBwidth]
		{\caption{An example of a graph with infinitely many distinct flows in IDE (all with different label functions and termination times). Two such flows are: All flow uses the above path or all flow uses the bottom path. Since the first edge on both paths has rate capacity $2$, no queues will form on those edges. So both edges will lie on a shortest path as long as the flow has not arrived at node $w$.}\label{fig:NonUniquenessOfIDEFlows}}
		{
			\begin{adjustbox}{max width=.38\textwidth}
				\begin{tikzpicture}[scale=1.0]
	\begin{scope}
		\node (s) at (0, 0) {};
		\draw (s) +(0, 0) node[labeledNodeS] (s) {$s$};
		\draw (s) +(1.5, 1.25) node[labeledNodeS] (v) {$v$}
		edge[normalEdgeF, <-] node[above=.3cm,left] {$(1,2)$} (s);
		\draw (s) +(1.5, -1.25) node[labeledNodeS] (w) {$w$}
		edge[normalEdgeF, <-] node[below=.3cm,left] {$(1,2)$} (s);
		\draw (s) +(3, 0) node[labeledNodeS] (t) {$t$}
		edge[normalEdgeF, <-] node[above=.3cm,right] {$(1,2)$} (v)
		edge[normalEdgeF, <-] node[below=.3cm,right] {$(1,1)$} (w);
		
		\node[rectangle, draw, minimum width=.75cm, minimum height=1cm,fill={blue!50}] at (-2,0) {$u_1 = 2 \cdot \mathds{1}_{[0, 1)}$}
		edge[->, line width=2pt] (s); 
	\end{scope}
\end{tikzpicture}
			\end{adjustbox}
		}
	\end{figure}
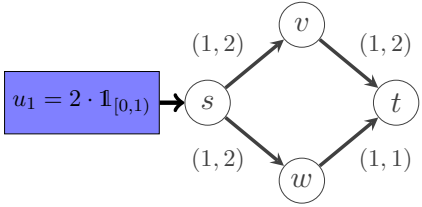
\end{remark}

\begin{remark} \label{remark:IDE_thin_flows_single_sink}
The task to find suitable inflow rates $x_e$ at a given time $\theta_k$ can also be formulated globally as follows:
Find two vectors $(x_e)_{e \in E_{\theta_k}}, (a_v)_{v \in V}$ such that:
  \begin{alignat*}{2}
  \sum_{vw \in E_{\theta_k}} x_e &= b_v^-(\theta_k)&& \text{ for all } v \in V \setminus \set{t},\\
  x_e & \geq 0&& \text{ for all } e \in E_{\theta_k},\\
  a_t &= 0, &&\\ 
  a_v &= \min_{e = vw \in E_{\theta_k}} \frac{g_e(x_e)}{\nu_e} + a_w \quad 
  && \text{ for all } v \in V\setminus\Set{t},\\
  a_v &= \frac{g_e(x_e)}{\nu_e} + a_w
  && \text{ for all } e = vw \in E_{\theta_k} \text{ with } x_e > 0,
  \end{alignat*}
  where \[g_e(x_e) \coloneqq \begin{cases}
  x_e- \nu_e & \text{ if } q_e(\theta_k) > 0\\
  \max\Set{x_e- \nu_e, 0} & \text{ if } q_e(\theta_k) = 0.
  \end{cases}\]
  This formulation is very similar to the \emph{thin flow with resetting} formulation of dynamic equilibria (see \cite{Koch11}, \cite{CominettiCL15}) and in the same way the existence of a solution can be shown via Kakutani's fixed point theorem. We will extend this formulation to a multi-commodity version in order to show the existence of IDE flows in a multiple sink setting in \Cref{sec:ExistenceMultiSink}.
\end{remark}

We have proven now that an IDE flow always exists, but in order to actually compute an IDE flow we need to make \Cref{lemma:ExistenceOfOptSolutionBv} constructive, i.e. provide an algorithm to find an optimal solution to \eqref{OPT:bvThetaK} or, equivalently, a distribution of the inflow $b^-_v(\theta_k)$ satisfying \Cref{eq:diff-invariant1,eq:diff-invariant2}. Since \eqref{OPT:bvThetaK} is a convex program there are various methods available -- but thanks to the particular nice structure of \eqref{OPT:bvThetaK} even a simple water filling procedure suffices:

\ifarxiv
First, note that we have
\[
c'_{vw_i}(\theta_k)+\ell'_{w_i}(\theta_k) = %\frac{g_{vw_i}(f^+_{vw_i}(\theta_k))}{\nu_{vw_i}}+a_{w_i} =
\begin{cases}
\frac{f^+_{vw_i}(\theta_k) - \nu_{vw_i}}{\nu_{vw_i}} + a_{w_i}, &\text{ if } q_{vw_i}(\theta_k) > 0\\
\frac{\max\set{f^+_{vw_i}(\theta_k) - \nu_{vw_i},0}}{\nu_{vw_i}} + a_{w_i}, &\text{ if } q_{vw_i}(\theta_k) = 0
\end{cases}
\]
and thus, setting \[\beta_i \coloneqq \begin{cases}a_{w_i}-1, &q_{vw_i}(\theta_k)>0\\a_{w_i}, &q_{vw_i}(\theta_k)=0,\end{cases} \quad \gamma_i \coloneqq \begin{cases}0, &q_{vw_i}(\theta_k)>0\\\nu_{vw_i}, &q_{vw_i}(\theta_k)=0,\end{cases} \quad \alpha_i \coloneqq \nu_{vw_i} \text{ and }\]
\begin{align}
h_{vw_i}(z) \coloneqq \begin{cases}\beta_i, &z \leq \gamma_i \\\beta_i + \frac{1}{\alpha_i}(z - \gamma_i), &z\geq \gamma_i,\end{cases}\label{eq:FunctionsForWaterFillingAlg}
\end{align}
we obtain $h_{vw_i}(f^+_{vw_i}(\theta_k)) = c'_{vw_i}(\theta_k)+\ell'_{w_i}(\theta_k)$. Without loss of generality we assume that the nodes $w_i$ are sorted such that $\beta_1 \leq \beta_2 \leq \dots \leq \beta_{p_k}$. Then, \Cref{Alg:FindExtension} computes a distribution $b_v^-(\theta_k)=\sum_{i=1}^{p_k}f_{vw_i}(\theta_k)$ satisfying \Cref{eq:diff-invariant1,eq:diff-invariant2}.

\begin{algorithm}[H] 
	\SetKwInOut{Input}{Input}\SetKwInOut{Output}{Output}
	\caption{Water filling procedure for flow distribution}\label{Alg:FindExtension}
	
	\Input{A number $b_v^-(\theta_k) \geq 0$ and functions $h_i: \IR_{\geq 0} \to \IR_{\geq 0}$ of the form \eqref{eq:FunctionsForWaterFillingAlg} with $\alpha_i > 0$ for $i = 1, \dots, {p_k}$ and $\beta_1 \leq \beta_2 \leq \dots \leq \beta_{p_k}$.}
	\Output{Values $z_i \geq 0$ such that $\sum_{i=1}^{p_k} z_i = b_v^-(\theta_k)$ and for some $r' \leq p_k$ satisfying $h_0(z_0)= \dots = h_{r'}(z_{r'}) \leq \beta_{r'+1}$, $z_i>0$ for $i \leq r'$ and $z_i=0$ for $i > r'$.}
	
	\vspace{.3em}\hrule\vspace{.3em}
	
	Find the maximal $r \in \set{0,1, \dots, p_k}$ with $\sum_{i=1}^{r}\max\set{z|h_i(z) \leq \beta_r} \leq b_v^-(\theta_k)$
	
	\eIf{$r<p_k$ and $\sum_{i=1}^{r}\max\set{z|h_i(z) \leq \beta_{r+1}} \leq b_v^-(\theta_k)$}{
		
		Set $z_i \leftarrow \begin{cases}
		\max\set{z|h_i(z) \leq \beta_{r+1}}, 	&i \leq r\\
		b^-_v(\theta_k) - \sum_{i<r}z_i,					&i = r+1\\
		0									&i > r+1
		\end{cases}$ \label{Alg:FindExtension:DefZi1}
	}{
		Set $z_i \leftarrow \begin{cases}
		\max\set{z|h_i(z) \leq \beta_{r}}, 	&i \leq r\\
		0									&i > r
		\end{cases}$ and $b' \leftarrow b^-_v(\theta_k) - \sum_{i=1}^{k}z_i$ \label{Alg:FindExtension:DefZi2a}
		
		Set $z_i \leftarrow z_i + \frac{\alpha_i}{\sum_{j=1}^{r-1}\alpha_j}b'$ for all $i \leq r$. \label{Alg:FindExtension:DefZi2b}
	}
	
	\KwRet{$z_1, \dots, z_{p_k}$}
\end{algorithm}

\begin{lemma}
	\Cref{Alg:FindExtension} computes in $\BigO(\abs{E}^2)$ edge inflow rates $z_i$ such that by setting $f^+_{vw_i}(\theta_k) \coloneqq z_i$ we get a flow distribution satisfying $\sum_{i \in [p_k]}f^+_{vw_i}(\theta_k)=b^-_v(\theta_k)$ as well as \Cref{eq:diff-invariant1,eq:diff-invariant2}. 
\end{lemma}

\begin{proof}
	First, note that the functions $h_i$ are non-decreasing, continuous and unbounded on $\IR_{\geq 0}$. Thus, all maxima in \Cref{Alg:FindExtension} are well defined. Since the functions are even piecewise linear, all these maxima can actually be determined in constant time.
	
	It remains to show that the returned $z_i$ do indeed satisfy $\sum_{i=1}^{p_k} z_i = b_v^-(\theta_k)$ and for some $r' \leq p_k$, we have $h_0(z_0)=h_1(z_1)= \dots = h_{r'}(z_{r'}) \leq \beta_{r'+1}$ and $z_i=0$ for all $i > r'$, which is a straight forward calculation.
\end{proof}

Thus, we can indeed compute an IDE flow for single-sink networks with right constant network inflow rates, provided that 
\begin{enumerate}
	\item the IDE flow eventually terminates and
	\item the $\varepsilon$ in \Cref{lemma:ExtendingContIUEFlows} does not become arbitrarily small (and therefore no limit process is needed in \Cref{thm:existence}).
\end{enumerate}
While the general validity of the second assumption currently remains an open question, we will show in the following \nameCref{sec:SingleSinkTermination} that the first assumption holds for all IDE flows in single-sink networks. We now give an example for the computation of an IDE:

\begin{example}\label{ex:IDEComputation}
	We use the same graph as in the example in \Cref{subsec:Example}, but with different rate capacities and only a single commodity with source node $s$, sink node $t$ and a constant network inflow rate of $16$ over the interval $[0,1]$ (see the top left picture in \Cref{fig:ide2}).
	
	%	\begin{figure}[h!]
	%		\begin{center}
	%			\begin{adjustbox}{max width=.3\textwidth}
	%				\input{tikz/IDE-Computation-Example-Graph}
	%			\end{adjustbox}
	%		\end{center}
	%		\vspace{-0.5cm}
	%		\caption[format=hang]{TODO}
	%		\label{fig:IDEComputation}
	%	\end{figure}
	
	\newcommand{\HIgraph}[5]{\begin{tikzpicture}[anchor=center,scale=.5,solid,black,
		declare function={
			f(\x)= and(\x>=0, \x<=#3) * #2   
			+ (\x>#3) * (#2 + (#1) * (\x-#3));
		}
		]
		\begin{axis}[xmin=0,xmax=15,ymax=6, ymin=-2, samples=500,width=10cm,height=5cm]
		\addplot[blue, ultra thick,domain=0:15] {f(x)};
		\addplot[red,ultra thick] coordinates {(#5, -2) (#5, 6)};
		\end{axis}
		\node[blue]() at (3,1.5) {$h_{#4}$};
		\end{tikzpicture}}
	
	At time $\theta_0=0$, we have $b^-_v(0) = b^-_w(0) = b^-_t(0) = 0$ and $b^-_s(0)=16$, so we only need to distribute outflow at node $s$. Both edges $st$ and $sv$ are active at time $\theta_0$ and we have $a_{t} = a_{v} = 0$, thus, the functions $h_{st}$ and $h_{sv}$ are of the form displayed in \Cref{fig:IDEComputationHIAt0}.
	\begin{figure}[h!]
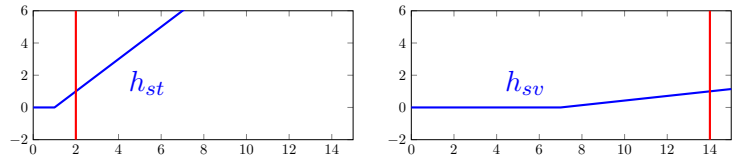

		\begin{center}
			\begin{adjustbox}{max width=.9\textwidth}
				\HIgraph{1}{0}{1}{st}{2}\hspace{1em}\HIgraph{1/7}{0}{7}{sv}{14}
			\end{adjustbox}
		\end{center}
		\vspace{-0.5cm}
		\caption[format=hang]{The functions $h_{st}$ and $h_{sv}$ at time $\theta_0=0$. The vertical red line indicates the $z_i$ values determined by \Cref{Alg:FindExtension}.}
		\label{fig:IDEComputationHIAt0}
	\end{figure}
	\Cref{Alg:FindExtension} then gives us the edge inflow rates $f^+_{st}(0) = 1 + \frac{1}{1+7}\cdot 8 = 1 + 1$ and $f^+_{sv}(0) = 7 + 7$. This ensures that both edges remain active up to time $\theta_1=1$, where the inflow into node $s$ stops and the first flow particles arrive at node $v$. Then, the only node with positive $b_v^-$ is $v$, but since $v$ has only one edge leaving this node, no flow distribution is necessary. At time $\theta_2=2$, the flow arriving at node $w$ only has one active edge, and therefore enters the edge towards $t$ at a rate of $7$. At time $\theta_3=2.5$, the edge $ws$ becomes active and the outflow from node $w$ (at rate $7$) needs to be redistributed. \Cref{Alg:FindExtension} gives us the inflow rates $f^+_{wt}(2.5) = 1$ and $f^+_{ws}(2.5) = 6$. From time $\theta_4=3.5$ onward, flow arrives at node $s$ at rate $6$ and at node $w$ at rate $7$. Since $s$ is currently closer  to $t$ than $w$ ($\ell_s(3.5)=3 < 4 = \ell_{w}(3.5)$), we start by distributing the outflow of $s$. From $s$ only the edge towards $t$ is active, so all flow enters this edge. This means that a queue will start to build up on the edge $st$ at a rate of $5$, and therefore $a_{s} = 5$. Thus, \Cref{Alg:FindExtension} gives us the inflow rates $f^+_{wt}(3.5) = 6$ and $f^+_{ws}(3.5) = 1$. At time $\theta_5=4$, the only active edge for the flow arriving at node $s$ is $st$. At time $\theta_6=4.5$, the edge $sv$ becomes active, too.  Thus, we need to compute a new distribution for the flow arriving at node $s$ (at rate $1$). Since we have no outflow from node $w$, the queue on edge $wt$ decreases at a rate of $1$ leading to $a_w = -1$, and therefore $a_v=-1$. \Cref{Alg:FindExtension} then gives us the inflow rates $f^+_{st}(4.5) = 0$ and $f^+_{sv}(4.5) = 1$. From time $\theta_7=5$ on only the physically shortest paths are active, so all flow particles will stay on their currently chosen path towards $t$. The IDE flow computed by this procedure is displayed in \Cref{fig:ide2}.
	
	\begin{figure}[ht!]
		\begin{center}
			\begin{adjustbox}{max width=1\textwidth}
				\begin{tikzpicture}[scale=1.2, font=\Large]

	\newcommand{\exampleGraph}[2]{		
		\node[labeledNodeS] (s1) at (0, 4) {$s$};
		\draw (s1) ++(4, 0) node[labeledNodeS] (v) {$v$};
		\draw (s1) ++(0, -4) node[labeledNodeS] (t) {$t$};
		\draw (v) ++(0, -4) node[labeledNodeS] (w) {$w$};
		
		\draw[normalEdge,->] (s1) -- node[below] {\ifthenelse{\equal{#2}{labels}}{\Large$(1,7)$}{}} (v);
		\draw[normalEdge,->] (s1) -- node[right] {\ifthenelse{\equal{#2}{labels}}{\Large$(3,1)$}{}} (t);
		\draw[normalEdge,->] (v)  -- node[left] {\ifthenelse{\equal{#2}{labels}}{\Large$(1,7)$}{}} (w);
		\draw[normalEdge,->] (w) -- node[below] {\ifthenelse{\equal{#2}{labels}}{\Large$(\tau_{wt},\nu_{wt})=(1,1)$}{}} (t);
		\draw[normalEdge,->] (w) -- node[below=.5cm] {\ifthenelse{\equal{#2}{labels}}{\Large$(1,6)$}{}} (s1);
		
		\node at (2,5) [rectangle,draw] (theta0) {\Large$\theta=#1$:};
	}
	
	\newcommand{\colComRed}{red!30}
	\newcommand{\colComBlue}{blue!60}
	
	% theta = 0 #########################################################################
	\begin{scope}
	
	\exampleGraph{0}{labels}
	
	\draw (s1) +(0,1.5) node[rectangle, draw, minimum width=1.5cm, minimum height=1cm,fill=\colComRed] (s1u) {\Large $u=16 \cdot \mathds{1}_{[0, 1)}$}; 
	\draw [->, line width=2pt] (s1u) -- (s1);
	
	\end{scope}

	% theta = 1 #########################################################################
	\begin{scope}[xshift=7cm]
	
		\node[rectangle, color=\colComRed,draw, minimum width=1.33cm,minimum height=.2cm, fill=\colComRed,rotate around={90:(0,0)}] at (0,3.1) {};
		\node[rectangle, color=\colComRed,draw, minimum width=.5cm,minimum height=.5cm, inner sep=0, fill=\colComRed,rotate around={90:(0,0)}] at (-.2,3.4) {};
		\node at (-1.5,3.5) [] () {\Large $q_{st}(1)=1$};
		\node[rectangle, color=\colComRed,draw, minimum width=4cm,minimum height=1.4cm, fill=\colComRed] at (2,4) {};
		\node[rectangle, color=\colComRed,draw, minimum width=.5cm,minimum height=3.5cm, fill=\colComRed] at (.55,5.2) {};
		\node at (-1,4.5) [] () {\Large $q_{sv}(1)=7$};
		
		\exampleGraph{1}{}	
	
	\end{scope}

	% theta = 2 #########################################################################
	\begin{scope}[xshift=14cm]
	
		\node[rectangle, color=\colComRed,  draw, minimum width=2.66cm,minimum height=.2cm, fill=\colComRed,rotate around={90:(0,0)}] at (0,2.45) {};
		\node[rectangle, color=\colComRed,  draw, minimum width=4cm,minimum height=1.4cm, fill=\colComRed] at (2,4) {};
		\node[rectangle, color=\colComRed,  draw, minimum width=4cm,minimum height=1.4cm, fill=\colComRed,rotate around={90:(0,0)}] at (4,2) {};
		
		\exampleGraph{2}{}
	
	\end{scope}

	% theta = 2.5 #########################################################################
	\begin{scope}[xshift=21cm]
	
		\node[rectangle, color=\colComRed,draw, minimum width=2.66cm,minimum height=.2cm, fill=\colComRed,rotate around={90:(0,0)}] at (0,2) {};
		\node[rectangle, color=\colComRed,  draw, minimum width=2cm,minimum height=1.4cm, fill=\colComRed] at (2.8,4) {};
		\node[rectangle, color=\colComRed,  draw, minimum width=4cm,minimum height=1.4cm, fill=\colComRed,rotate around={90:(0,0)}] at (4,2) {};
		\node[rectangle, color=\colComRed,draw, minimum width=2cm,minimum height=.2cm, fill=\colComRed] at (2.8,0) {};
		\node[rectangle, color=\colComRed,draw, minimum width=.5cm,minimum height=1.5cm, fill=\colComRed] at (3.5,-.6) {};
		\node at (2,-.8) [] () {\Large $q_{wt}(3)=3$};
		
	\exampleGraph{2.5}{}
	
	\end{scope}		
	
	% theta = 3.5 #########################################################################
	\begin{scope}[yshift=-7cm]
	
		\node[rectangle, color=\colComRed,draw, minimum width=2cm,minimum height=.2cm, fill=\colComRed,rotate around={90:(0,0)}] at (0,1.2) {};
		\node[rectangle, color=\colComRed,draw, minimum width=5.9cm,minimum height=1.2cm, fill=\colComRed,rotate around={-45:(0,0)}] at (2,2) {};
		\node[rectangle, color=\colComRed,  draw, minimum width=2cm,minimum height=1.4cm, fill=\colComRed,rotate around={90:(0,0)}] at (4,1.2) {};
		\node[rectangle, color=\colComRed,draw, minimum width=4cm,minimum height=.2cm, fill=\colComRed] at (2,0) {};
		\node[rectangle, color=\colComRed,draw, minimum width=.5cm,minimum height=1.5cm, fill=\colComRed] at (3.5,-.6) {};
		\node at (2,-.8) [] () {\Large $q_{wt}(3)=3$};
		
		\exampleGraph{3.5}{}
	
	\end{scope}

	% theta = 4 #########################################################################
	\begin{scope}[yshift=-7cm, xshift=7cm]

		\node[rectangle, color=\colComRed,draw, minimum width=1.33cm,minimum height=.2cm, fill=\colComRed,rotate around={90:(0,0)}] at (0,.866) {};
		\node[rectangle, color=\colComRed,draw, minimum width=.5cm,minimum height=1.25cm, inner sep=0, fill=\colComRed,rotate around={90:(0,0)}] at (-.5,3.4) {};
		\node at (-1.5,2.7) [] () {\Large $q_{st}(4)=2.5$};
		\node[rectangle, color=\colComRed,draw, minimum width=4cm,minimum height=.2cm, fill=\colComRed] at (2,0) {};					
		\node[rectangle, color=\colComRed,draw, minimum width=.5cm,minimum height=2.75cm, fill=\colComRed] at (3.5,-1.2) {};
		\node at (2,-.8) [] () {\Large $q_{wt}(4)=5.5$};
		\node[rectangle, color=\colComRed,draw, minimum width=3cm,minimum height=1.2cm, fill=\colComRed,rotate around={-45:(0,0)}] at (1.12,2.88) {};
		\node[rectangle, color=\colComRed,draw, minimum width=3cm,minimum height=.2cm, fill=\colComRed,rotate around={-45:(0,0)}] at (2.88,1.12) {};
		
		\node[rectangle, color=\colComRed,draw, minimum width=.68cm,minimum height=.2cm, fill=\colComRed,rotate around={90:(0,0)}] at (0,3.4) {};

		\exampleGraph{4}{}
	
	\end{scope}

	% theta = 4.5 #########################################################################
	\begin{scope}[yshift=-7cm, xshift=14cm]

		\node[rectangle, color=\colComRed,draw, minimum width=1.33cm,minimum height=.2cm, fill=\colComRed,rotate around={90:(0,0)}] at (0,3.06) {};	
		\node[rectangle, color=\colComRed,draw, minimum width=.5cm,minimum height=2.5cm, inner sep=0, fill=\colComRed,rotate around={90:(0,0)}] at (-1,3.4) {};
		\node at (-1.5,2.7) [] () {\Large $q_{st}(4)=5$};
		\node[rectangle, color=\colComRed,draw, minimum width=4cm,minimum height=.2cm, fill=\colComRed] at (2,0) {};					
		\node[rectangle, color=\colComRed,draw, minimum width=.5cm,minimum height=2.5cm, fill=\colComRed] at (3.5,-1.1) {};
		\node at (2,-.8) [] () {\Large $q_{wt}(4)=5$};
		\node[rectangle, color=\colComRed,draw, minimum width=2.97cm,minimum height=.2cm, fill=\colComRed,rotate around={-45:(0,0)}] at (1.1,2.9) {};
		\node[rectangle, color=\colComRed,draw, minimum width=.66cm,minimum height=.2cm, fill=\colComRed,rotate around={90:(0,0)}] at (0,.533) {};		
		
		\exampleGraph{4.5}{}
	
	\end{scope}

	% theta = 5 #########################################################################
	\begin{scope}[yshift=-7cm, xshift=21cm]
	
		\node[rectangle, color=\colComRed,draw, minimum width=2cm,minimum height=.2cm, fill=\colComRed,rotate around={90:(0,0)}] at (0,2.8) {};	
		\node[rectangle, color=\colComRed,draw, minimum width=.5cm,minimum height=2.25cm, inner sep=0, fill=\colComRed,rotate around={90:(0,0)}] at (-1,3.4) {};
		\node at (-1.5,2.7) [] () {\Large $q_{st}(4)=4.5$};
		\node[rectangle, color=\colComRed,draw, minimum width=4cm,minimum height=.2cm, fill=\colComRed] at (2,0) {};					
		\node[rectangle, color=\colComRed,draw, minimum width=.5cm,minimum height=2.5cm, fill=\colComRed] at (3.5,-1.1) {};
		\node at (2,-.8) [] () {\Large $q_{wt}(4)=4.5$};
		\node[rectangle, color=\colComRed,draw, minimum width=2cm,minimum height=.2cm, fill=\colComRed] at (1.2,4) {};	
		
		\exampleGraph{5}{}
	
	\end{scope}	
\end{tikzpicture}
			\end{adjustbox}
		\end{center}
		\vspace{-0.5cm}
		\caption[format=hang]{The evolution of the computed IDE over the time horizon $[0,5]$.}
		\label{fig:ide2}
	\end{figure}
	
\end{example}

\begin{remark}
	Note that at time $\theta_4=3.5$ it is vital to consider $s$ before $w$. Otherwise we would continue the flow split at node $w$ from time $\theta_3=3$ leading to the edge $vs$ becoming inactive immediately after $\theta=3.5$ (i.e. we could not extend our IDE flow for any $\varepsilon>0$ in that way). At time $\theta_6=4.5$, before distributing the flow arriving at $s$, both $s$-$t$ paths (the direct one and the one over $v$ and $w$) might seem to be completely equivalent as both have a physical path length of $3$ and one queue of current length $5$ decreasing at a rate of $1$. However, we may, in fact, not send any flow into the edge $st$ as this would slow down the decrease of its queues length, making this edge immediately inactive, while sending flow towards $v$ does not change the decrease rate of the queue on edge $wt$. Our algorithm does indeed send all flow into the edge $sv$. After time $\theta_6=4.5$, the flow particles on edge $sv$ are traversing this edge for the second time, i.e., they have completed a cycle.
\end{remark}
\fi

%!TEX root = ../article.tex
\section{Termination of IDE Flows in Single-Sink Networks}\label{sec:SingleSinkTermination}

In this section, we investigate the question, whether an IDE flow vanishes within finite time given finitely lasting network inflow rates. More precisely, given a time $\theta_0$, such that $\supp(u_i) \subseteq [0, \theta_0]$ for every $i \in I$, we ask whether there exists a time $\hat{\theta} \geq \theta_0$, such that all injected flow actually reaches the sink within time $\hat{\theta}$.
%In this section, we investigate the question, whether an IDE flow for the case
%of finitely lasting integrable inflow rates (i.e. there exists a time $\theta_0$, such that the supports of all inflow functions are contained in $[0,\theta_0]$) actually vanishes within finite time,
%that is, if the injected flow actually reaches the sink within finite time.
To answer this question, we first need to introduce some additional notation. For every edge $e \in E$, we define a function $\edgeLoad[e]$ denoting the total amount of flow currently on edge $e$ (either waiting in its queue or traveling along the edge) for any time~$\theta$. As in \cite{Sering2019}, we call these the \emph{edge load functions}:
\[\edgeLoad[e]: \IR_{\geq 0} \to \IR_{\geq 0}, \quad \theta \mapsto F^+_e(\theta)-F^-_e(\theta).\]
The function $\edgeLoad(\theta) := \sum_{e \in E}\edgeLoad[e](\theta)$ specifies the \emph{total amount of flow in the network at time~$\theta$}. Furthermore, we define a function $Z$ indicating the amount of flow that already reached the sink $t$ by time $\theta$:
\[Z: \IR_{\geq 0} \to \IR_{\geq 0}, \quad \theta \mapsto \sum_{e \in \delta^-_t}F^-_e(\theta) - \sum_{e \in \delta^+_t}F^+_e(\theta). \]
Note that for IDE flows the subtrahend is always $0$ since edges leaving $t$ are never active.

For every subset $W \subseteq V$ and any time $\theta$ a direct computation shows that we have
	\begin{equation} \label{eq:GeIsTotalFlowInGraph}
		\sum_{\mathclap{e \in E(W)}}\edgeLoad[e](\theta) = \sum_{\mathclap{e \in \delta^-_W}}F_e^-(\theta) + \sum_{\mathclap{i \in I: s_i \in W}}U_i(\theta) - \sum_{\mathclap{e \in \delta^+_W}}F_e^+(\theta) - \begin{cases}
		Z(\theta), 	& t \in W \\
		0,				& \text{else}
		\end{cases},
	\end{equation}
with $\delta^+_W := \set{wv \in E | w \in W, v\notin W}$ and $\delta^-_W := \set{vw \in E | v \notin W, w\in W}$. In particular taking $W = V$ we get
	\begin{equation}\label{eq:GeIsTotalFlowInGraphAllNodes}
		\edgeLoad(\theta) = \sum_{i \in I} U_i(\theta) - Z(\theta).
	\end{equation}
Since $Z'(\theta) = \sum_{e \in \delta^-_t}f^-_e(\theta) - \sum_{e \in \delta^+_t}f^+_e(\theta)$ is always non-negative by \Cref{eq:Cont-FlowDefProperties-FlowConsSink}, it follows immediately that the total amount of flow in the network is non-increasing after time $\theta_0$. More generally, since all $F_e^+$ are non-decreasing, for all $W \subseteq V$ with $\delta^-_W = \emptyset$ we have
	\begin{equation}\label{eq:FlowVolumeReduction}
		\sum_{e \in E(W)}\edgeLoad[e](\theta_2) \leq \sum_{e \in E(W)}\edgeLoad[e](\theta_1) \quad \text{ for all } \theta_2 \geq \theta_1 \geq \theta_0.
	\end{equation} 
In particular, for $\hat{\theta} \geq \theta_0$ with $\edgeLoad(\hat{\theta}) = 0$, we have $\edgeLoad(\hat{\theta}) = 0$ for all $\theta \geq \hat{\theta}$.
\begin{defn}
   	We say a \revised{feasible flow} $f$ \emph{terminates,} if there exists a time $\hat{\theta} \geq \theta_0$ with $\edgeLoad(\hat{\theta})=0$, i.e., the network is empty at time $\hat \theta$ (and remains empty for all later times).
\end{defn}

Before we turn to the main termination result, we need a technical lemma showing that all flow on an edge eventually leaves the edge (ignoring the identities of the flow particles).
\begin{lemma}\label{lemma:FlowOnEdgesLeavesFastVar}
	Let $f$ be a feasible flow over time, $\theta_1 \in \IR_{\geq 0}$, $e \in E$ and any $\lambda \in [0, \edgeLoad[e](\theta_1)]$. Then there exists a time $\theta_2 \geq \theta_1$ such that a flow volume of at least $\lambda$ leaves $e$ during the interval $[\theta_1, \theta_2]$, i.e.,
	$F^-_e(\theta_2) - F^-_e(\theta_1) \geq \lambda$.
\end{lemma}
\begin{proof}
	Suppose for contradiction that $F^-_e(\theta) - F^-_e(\theta_1) < \lambda \leq \edgeLoad[e](\theta_1)$ for all $\theta \geq \theta_1$. As $F^+_e$ is non-decreasing this implies $F^-_e(\theta+\tau_e) < F^+_e(\theta_1) \leq F^+_e(\theta)$, from which, by \Cref{eq:Cont-FlowDefProperties-QueueLengthWithF}, we get $q_e(\theta) = F^+_e(\theta) - F^-_e(\theta+\tau_e) > 0$ for all $\theta \geq \theta_1$. Hence, \Cref{eq:Cont-FlowDefProperties-OpAtCap} gives us $f_e^-(\theta+\tau_e) = \nu_e$ for all $\theta \geq \theta_0$ implying that $F_e^-$ grows unboundedly, which is a contradiction. 
\end{proof}

We show next that for acyclic networks every feasible flow over time terminates. This intuitive result will serve as a building block for the more general result that in a single-sink network all IDE flows terminate.

\begin{lemma}\label{lemma:TerminationAcyclicNetworks}
	Let $G$ be an acyclic graph and assume finitely lasting bounded inflow functions. Then, every feasible flow over time terminates.
\end{lemma}
\begin{proof}
	Since the graph is acyclic, we can consider a topological order $\beforeN$ on $V$, i.e., $uv \in E \implies u \beforeN v$.
	Without loss of generality, let $t$ be reachable from every node, since the part of the graph that cannot reach $t$ will never be used by any feasible flow over time and can therefore be removed. Thus, $t$ must be the last (smallest) element in this order. We will also use the notation $e \beforeNq w$ to indicate that an edge $e=uv$ lies before $w$, i.e.,~$u,v \beforeNq w$. Now, given a feasible flow over time $f$ we can define a function $\edgeLoad[\beforeNq w]$ for every node $w$ that gives us the total edge load before $w$ at any time $\theta$, i.e., $\edgeLoad[\beforeNq w](\theta) = \sum_{e \beforeNq w}\edgeLoad[e](\theta)$.
  
	As a first step, we show that after $\theta_0$, if there is no flow before some node $w$ (i.e., on edges $e \beforeNq w$), there will be no flow on any edge leaving $w$ some time later, or, more formally:
	\begin{claim}\label{lemma:TerminationAcyciclicNetworks:Claim}
		Let $e = wv \in E$ and $\theta_w \geq \theta_0$ such that we have $\edgeLoad[\beforeNq w](\theta_w) = 0$. Then, there exists a time $\theta_{wv} \geq \theta_w$ such that for all $\theta \geq \theta_{wv}$, we have $\edgeLoad[e](\theta)=0$.
	\end{claim}

	\begin{proofClaim}[Proof of \Cref*{lemma:TerminationAcyciclicNetworks:Claim}]
		We first show that $F^+_e(\theta) = F^+_e(\theta_w)$ for all $\theta \geq \theta_w$. We define the set $W := \set{u \in V | u \beforeNq w}$, so that we have $wv \in \delta^+_W$, as well as $\delta^-_W = \emptyset$ and 
			\[0 \leq \edgeLoad[\beforeNq w](\theta) = \sum_{e \in E(W)} \edgeLoad[e](\theta) \overset{\text{\eqref{eq:FlowVolumeReduction}}}{\leq} \sum_{e \in E(W)} \edgeLoad[e](\theta_w) = \edgeLoad[\beforeNq w](\theta_w) = 0\] 
		for every $\theta \geq \theta_w$. For all $\theta \geq \theta_w$ we obtain
		\begin{align*}
		0 &\leq F^+_e(\theta) - F^+_e(\theta_w) \leq \sum_{e' \in \delta^+_W}\left(F^+_{e'}(\theta) - F^+_{e'}(\theta_w)\right) \\
		&\overset{\mathclap{\text{\eqref{eq:GeIsTotalFlowInGraph}}}}{=} \sum_{i \in I: s_i \in W}U_i(\theta) - \edgeLoad[\beforeNq w](\theta) - \sum_{i \in I: s_i \in W}U_i(\theta_w) + \edgeLoad[\beforeNq w](\theta_w) = 0
		\end{align*}
    since $F_{e'}^+$ is non-decreasing, $t \notin W$ and $U_i(\theta) = U_i(\theta_w)$ for all $i \in I$.
   Hence, we have $F^+_e(\theta) = F^+_e(\theta_w)$.
		
	By \Cref{lemma:FlowOnEdgesLeavesFastVar} there exists a time $\theta_{wv} \geq \theta_w$ with $F^-_e(\theta_{wv})-F^-_e(\theta_w) \geq \edgeLoad[e](\theta_w)$. Since $F^-_e$ is non-decreasing, we also get $F^-_e(\theta)-F^-_e(\theta_w) \geq \edgeLoad[e](\theta_w)$ for any $\theta \geq \theta_{wv}$, and hence
		\ifarxiv
		\[
		0 \leq \edgeLoad[e](\theta) = F^+_e(\theta) - F^-_e(\theta) = F^+_e(\theta_w) - F^-_e(\theta) = \edgeLoad[e](\theta_w) + F^-_e(\theta_w) - F^-_e(\theta) \leq 0.\qedhere
		\]
		\fi\ifspringer
		\begin{align*}
				0 	&\leq \edgeLoad[e](\theta) = F^+_e(\theta) - F^-_e(\theta) = F^+_e(\theta_w) - F^-_e(\theta) \\
					&= \edgeLoad[e](\theta_w) + F^-_e(\theta_w) - F^-_e(\theta) \leq 0.\qedhere
		\end{align*}
		\fi
	\end{proofClaim}
	\begin{claim}\label{lemma:TerminationAcyciclicNetworks:Claim2}
		For every $v \in V$, there exists a time $\theta_v \geq \theta_0$ such that $\edgeLoad[\beforeNq v](\theta) = 0$ for all $\theta \geq \theta_v$.
	\end{claim}
	\begin{proofClaim}[Proof of \Cref*{lemma:TerminationAcyciclicNetworks:Claim2}]
		We show this by induction on the number of nodes greater than $v$ in the given topological order on $V$. Our base case is, that there are no nodes $w \beforeN v$. Then there are also no edges $e \beforeNq v$, and therefore $\edgeLoad[\beforeNq v](\theta) = 0$ holds for all $\theta \geq \theta_0$. So, we can assume that for all $w \beforeN v$ there are already times $\theta_w \geq \theta_0$ with $\edgeLoad[\beforeNq w](\theta)=0$ for all $\theta \geq \theta_w$. Then for every edge $wv \in E$, \Cref{lemma:TerminationAcyciclicNetworks:Claim} gives us a time $\theta_{wv} \geq \theta_w$ with $\edgeLoad[e](\theta) = 0$ for all $\theta\geq \theta_{wv}$. Setting $\theta_v \coloneqq \max\set{\theta_{wv} | wv \in \delta^-_v}$ then guarantees for all $\theta \geq \theta_v$ that
			\[\edgeLoad[\beforeNq v](\theta) = \sum_{e \beforeNq v}\edgeLoad[e](\theta)  \leq \sum_{wv \in E}(\edgeLoad[wv](\theta) + \edgeLoad[\beforeNq w](\theta)) = 0.\qedhere\]	
   \end{proofClaim}
	Finally, the lemma follows directly from \Cref{lemma:TerminationAcyciclicNetworks:Claim2} by setting $v = t$, as then we have $\edgeLoad(\theta_t) = \edgeLoad[\beforeNq t](\theta_t) = 0$ for some $\theta_t \geq \theta_0$.
\end{proof}
In the next step we show that if the sum of all edge loads between a node $v$ and the sink~$t$ are small enough (and, thus, in particular all queues on edges between $v$ and $t$ are small), then an IDE flow can not be diverted away from the physically shortest paths towards $t$. Since those physically shortest paths form a time independent acyclic subgraph, we will be able to apply \Cref{lemma:TerminationAcyclicNetworks} to the flow inside this subgraph. For the next lemma, we need the minimal non-zero difference between two path lengths
$\tauMinDiff \coloneqq \min\set{\tau(P)-\tau(P')>0 | u \in V, P,P' \text{ two }u\text{-}t\text{ paths}}$ and the minimal rate capacity $\nu_{\min} \coloneqq \min\set{\nu_e | e \in E}$.
\begin{lemma}\label{lemma:ShortestPathsWithLowVolumeAreActive}
	If, for some node $v \in V$ and	some time $\theta \in \IR_{\geq 0}$, every physical shortest $v$-$t$ path (i.e., w.r.t. $\tau$) has total flow volume of less than $\tauMinDiff \cdot \nu_{\min}$, then all active $v$-$t$ paths at time $\theta$ are also physical shortest $v$-$t$ paths, i.e., if $\sum_{e \in P}\edgeLoad[e](\theta) < \tauMinDiff\cdot \nu_{\min}$ for all physical shortest $v$-$t$ paths~$P$,
	then the following holds:
		\[P' \text{ is an active $v$-$t$ path at time $\theta$ } \implies P' \text{ is a physical shortest $v$-$t$ path}.\]
\end{lemma}

\begin{proof}
	Let $P$ be a physically shortest $v$-$t$ path and $P'$ an active $v$-$t$ path. Then we have 
		\ifarxiv
		\begin{align*} 
			\sum_{e \in P'}\tau_e \leq \sum_{e \in P'}c_e(\theta) &\leq \sum_{e \in P}c_e(\theta) = \sum_{e \in P}\tau_e + \sum_{e \in P}\frac{q_e(\theta)}{\nu_e} \leq \sum_{e \in P}\tau_e + \sum_{e \in P}\frac{\edgeLoad[e](\theta)}{\nu_{\min}} < \sum_{e \in P}\tau_e + \tauMinDiff.
		\end{align*}
		\fi\ifspringer
		\begin{align*} 
			\sum_{e \in P'}\tau_e \leq \sum_{e \in P'}c_e(\theta) &\leq \sum_{e \in P}c_e(\theta) = \sum_{e \in P}\tau_e + \sum_{e \in P}\frac{q_e(\theta)}{\nu_e} \\
					&\leq \sum_{e \in P}\tau_e + \sum_{e \in P}\frac{\edgeLoad[e](\theta)}{\nu_{\min}} < \sum_{e \in P}\tau_e + \tauMinDiff.
		\end{align*}
		\fi
	This implies
		$0 \leq  \sum_{e \in P'}\tau_e - \sum_{e \in P}\tau_e < \tauMinDiff$,
	and therefore $\sum_{e \in P'}\tau_e = \sum_{e \in P}\tau_e$ as $\tauMinDiff$ is the smallest nonzero distance between two physical path lengths. Thus, $P'$ is also a shortest path w.r.t. the physical transit times $\tau$.
\end{proof}

\begin{cor}\label{cor:TerminationIfTotalVolumeLess1}
	Let $f$ be an IDE flow with $\edgeLoad(\hat{\theta}) < \tauMinDiff \cdot \nu_{\min}$ for some $\hat{\theta} \geq \theta_0$. Then, $f$ terminates.
\end{cor}

\begin{proof}
	By \eqref{eq:FlowVolumeReduction} we also have $\edgeLoad(\theta) \leq \edgeLoad(\hat{\theta}) < \tauMinDiff\cdot\nu_{\min}$ for all $\theta \geq \hat{\theta}$. So from \Cref{lemma:ShortestPathsWithLowVolumeAreActive} we know that after $\hat{\theta}$ only shortest paths can be active, and therefore the flow only uses a time independent acyclic subgraph of $G$. Thus, by \Cref{lemma:TerminationAcyclicNetworks} the flow terminates.
\end{proof}

\begin{theorem}\label{thm:Termination_SingleSink}
	For multi-source single-sink networks, any IDE flow $f$ with finitely lasting bounded network inflow rates $u_i$ terminates.
\end{theorem}
\begin{proof}
	Let $W \subseteq V$ be a subset of nodes with the following properties:
	\begin{enumerate}
		\item For every $w \in W$, all physical shortest $w$-$t$ paths only use edges in $E(W)$.
		\item There is a $\theta_W$ such that for all $\theta \geq \theta_W$ and $e \in E(W)$, we have $\edgeLoad[e](\theta) < \frac{\tauMinDiff\cdot \nu_{\min}}{\abs{E}}$.
	\end{enumerate}
	We show that for every such $W \neq V$, there exists a node $v \in V\setminus W$ such that $W \cup \set{v}$ also has the two properties.
	Since $W = \set{t}$ satisfies the two properties this shows that $W = V$ exhibits those as well and, in particular, there exists some time $\theta_V$ with $\edgeLoad[e](\theta_V) < \frac{\tauMinDiff\cdot\nu_{\min}}{\abs{E}}$ for all edges $e \in E$, and therefore $\sum_{e \in E}\edgeLoad[e](\theta_V) < \tauMinDiff\cdot\nu_{\min}$. Hence, by \Cref{cor:TerminationIfTotalVolumeLess1}, $f$ terminates.
  
	Let $W \subsetneq V$ be a set of nodes fulfilling both properties and $v \in V\setminus W$ be the node with the shortest distance to $t$ with respect to $\tau$ of all nodes in $V \setminus W$. Then, all physically shortest $v$-$t$ paths only use edges from $E(W \cup \set{v})$, so the first property holds for $W \cup \set{v}$. 
	Since the second property holds for $W$ we know that for all $\theta \geq \theta_W$, we have 
		\[\sum_{e \in E(W)}q_e(\theta) \leq \sum_{e \in E(W)}\edgeLoad[e](\theta) < \frac{\abs{E(W)}\cdot\tauMinDiff\cdot\nu_{\min}}{\abs{E}} \leq \tauMinDiff\cdot\nu_{\min}.\]
	\Cref{lemma:ShortestPathsWithLowVolumeAreActive} implies that for every node $w \in W$, all active edges leaving $w$ have to be in $E(W)$, i.e., $\delta^+_w \cap E_\theta\subseteq E(W)$. Since $f$ is an IDE flow this implies $f^+_e(\theta) = 0$ for all $e \in \delta^+_W$, and thus $F^+_e(\theta) = F^+_e(\theta_W)$ for all those edges.
	We now assume by contradiction that the second property does not hold for $W \cup \set{v}$, so there is an edge $e \in \delta^+_v \cap \delta^-_W$ and a sequence of times $\theta_1 < \theta_2 < \dots$ with $\edgeLoad[e](\theta_k) \geq \frac{\tauMinDiff\cdot\nu_{\min}}{\abs{E}}$ for all $k \in \IN$ and $\theta_k \to \infty$ for $k \to \infty$. From \Cref{lemma:FlowOnEdgesLeavesFastVar}, we get times $\theta_k' \geq \theta_k$ with $F^-_e(\theta_k') - F^-_e(\theta_k) \geq \edgeLoad[e](\theta_k) \geq \frac{\tauMinDiff\cdot\nu_{\min}}{\abs{E}}$. By possibly taking a subsequence, we can assume $\theta_{k-1}' \leq \theta_{k}$ for all $k$, and thus,
		\ifarxiv
		\begin{align*}
			F^-_e(\theta_k') \geq \edgeLoad[e](\theta_k) + F^-_e(\theta_k) &\geq \edgeLoad[e](\theta_k) + F^-_e(\theta'_{k-1}) \geq \dots \geq \sum_{j=1}^{k}\edgeLoad[e](\theta_j) \geq k\cdot \frac{\tauMinDiff\cdot\nu_{\min}}{\abs{E}}.
		\end{align*}
		\fi\ifspringer
		\begin{align*}
			F^-_e(\theta_k') \geq \edgeLoad[e](\theta_k) + F^-_e(\theta_k) &\geq \edgeLoad[e](\theta_k) + F^-_e(\theta'_{k-1}) \geq \dots \\
					&\geq \sum_{j=1}^{k}\edgeLoad[e](\theta_j) \geq k\cdot \frac{\tauMinDiff\cdot\nu_{\min}}{\abs{E}}.
		\end{align*}
		\fi
	Hence $F^-_e(\theta_k)$ tends to infinity as $k$ grows larger. On the other hand \eqref{eq:GeIsTotalFlowInGraph} states that
	\[F^-_e(\theta_k) = \sum_{\mathclap{e' \in E(W)}}\edgeLoad[e'](\theta_k) + \sum_{\mathclap{e' \in \delta^+_W}}F_{e'}^+(\theta_k) + Z(\theta_k) - \sum_{\mathclap{i \in I: s_i \in W}} U_i(\theta_k) - \sum_{\mathclap{e' \in \delta^-_W\setminus\set{e}}}F_{e'}^-(\theta_k)\]
	which in turn is bounded from above as all positive summands are bounded as well:
	\begin{itemize}
		\item $\sum_{e' \in E(W)}\edgeLoad[e'](\theta_k) \leq \frac{\abs{E(W)}}{\abs{E}}$ since $\theta_k \geq \hat{\theta}$ and $W$ has the second property.
		\item $\sum_{e' \in \delta^+_W}F_{e'}^+(\theta_k) = \sum_{e' \in \delta^+_W}F_{e'}^+(\hat{\theta})$ as shown above.
		\item $Z(\theta_k) \overset{\text{\eqref{eq:GeIsTotalFlowInGraphAllNodes}}}{\leq} \sum_{i \in I}U_i(\theta_k) = \sum_{i \in I} U_i(\theta_0) < \infty$.  
	\end{itemize}
	This is a contradiction. So the second property must also hold for $W \cup\set{v}$, which concludes the proof.
\end{proof}

%!TEX root = ../article.tex
\section{Existence of IDE Flow\revised{s} in Multi-Sink Networks}\label{sec:ExistenceMultiSink}

We now turn to the general model with multiple sinks. In this section, we will give two different proofs for the existence of IDE flows in such networks. The first one is similar to the proof for single-sink networks in \Cref{sec:ExistenceSingleSink} using the concept of thin flows, and thus, constructive in the same sense. 
%(i.e., under the assumption that the limit argument for extending the flow over the whole of $\IR_{\geq 0}$ is not actually needed). 
The second proof uses an infinite dimensional variational inequality. While not constructive, it allows for much more general network inflow rate functions $u_i$ (and interestingly, avoids the need for a limit argument to extend the flow for all times altogether).

\subsection{Right-Constant Network Inflow Rates}
First, we want to show that IDE flows exist in networks with multiple sinks, still under the assumption that the network inflow rates $u_i$ are right-constant. As before, we show that for a given IDE flow up to some point in time $\theta_k$ it is possible to extend it by some $\varepsilon > 0$. In contrast to the case of a single sink, we need to determine all inflow rates $f_{i,e}^+$ and the slopes $a_{i, k}$ of the current shortest path distances at the same time.

First of all, we denote the current inflow at a node $v$ for commodity $i$ by
\[b_{i, v}^-(\theta_k) \coloneqq \sum_{e \in \delta_v^-} f_{i, e}^-(\theta_k) + \mathds{1}_{v = s_i} \cdot u_i(\theta_k),\]
where $\mathds{1}_{v = s_i}=1$ if $v=s_i$ and $\mathds{1}_{v = s_i}=0$ otherwise. 
With this we can extend the idea of IDE thin flows (inspired by the \emph{thin flows with resetting} for dynamic equilibria; see \cite{Koch11}, \cite{CominettiCL15}) introduced in \Cref{remark:IDE_thin_flows_single_sink} to the multi-commodity setting.

\begin{defn}[IDE thin flows]
For a given IDE flow up to time $\theta_k$ we call a pair of vectors $(x, a) \in \IR_{\geq 0}^{I \times E} \times \IR^{I \times V}$ an \emph{IDE thin flow} if it satisfies:
  \ifarxiv
  \begin{alignat*}{2}
  \sum_{e \in \delta_v^+} x_{i, e} &= b_{i, v}^-(\theta_k)&& \text{ for all } i \in I \text{ and } v \in V \setminus \set{t_i},\label{eq:thin_flow_conservation}\tag{TF1}\\
  x_{i, e} & = 0&& \text{ for all } i \in I \text{ and } e \in E \setminus E^i_{\theta_k},\label{eq:thin_flow_zero_on_inactive}\tag{TF2}\\
  a_{i,t_i} &= 0 && \text{ for all } i \in I,\label{eq:thin_flow_t}\tag{TF3}\\ 
  a_{i, v} &= \min_{e = vw \in E^i_{\theta_k}} \frac{g_e(\sum_{j \in I} x_{j, e})}{\nu_e} + a_{i, w} \quad 
  && \text{ for all } v \in V\setminus\Set{t_i},\label{eq:thin_flow_min}\tag{TF4}\\
  a_{i, v} &= \frac{g_e(\sum_{j \in I} x_{j, e})}{\nu_e} + a_{i, w}
  && \text{ for all } e = vw \in E^i_{\theta_k} \text{ with } x_{i, e} > 0,\label{eq:thin_flow_tight}\tag{TF5}
  \end{alignat*}
  \fi\ifspringer
    \begin{alignat*}{2}
  \sum_{\revised{e \in \delta_v^+}} x_{i, e} &= b_{i, v}^-(\theta_k)&& \text{ f.a. } i \in I \text{ and } v \in V \setminus \set{t_i},\label{eq:thin_flow_conservation}\tag{TF1}\\
  x_{i, e} & = 0&& \text{ f.a. } i \in I \text{ and } e \in E \setminus E^i_{\theta_k},\label{eq:thin_flow_zero_on_inactive}\tag{TF2}\\
  a_{i,t_i} &= 0 && \text{ f.a. } i \in I,\label{eq:thin_flow_t}\tag{TF3}\\ 
  a_{i, v} &= \min_{\substack{e \in E^i_{\theta_k} \\ e = vw}} \frac{g_e(\sum_{j \in I} x_{j, e})}{\nu_e} + a_{i, w} \quad 
  && \text{ f.a. } v \in V\setminus\Set{t_i},\label{eq:thin_flow_min}\tag{TF4}\\
  a_{i, v} &= \frac{g_e(\sum_{j \in I} x_{j, e})}{\nu_e} + a_{i, w}
  && \text{ f.a. } e = vw \in E^i_{\theta_k} \text{ with } x_{i, e} > 0,\label{eq:thin_flow_tight}\tag{TF5}
  \end{alignat*}
  \fi
  where \[g_e(x_e) \coloneqq \begin{cases}
  x_e- \nu_e & \text{ if } q_e(\theta_k) > 0,\\
  \max\Set{x_e- \nu_e, 0} & \text{ if } q_e(\theta_k) = 0.
  \end{cases}\]
\end{defn}

As a first step, we show that such an IDE thin flow always exists for every reasonable network (i.e., every node $v$ with $b^-_{i, v}(\theta_k) > 0$ can reach $t_i$ within $E^i_{\theta_k}$) and all current inflow vectors $b^-(\theta_k)$. We do this by utilizing the following fixed point theorem by Kakutani \cite{Kakutani1941}:
\begin{theorem}[Kakutani's Fixed Point Theorem] \label{thm:kakutani}
Let $K$ be a compact, convex and non-empty subset of $\IR^n$, $n \in \IN$, and $\Gamma\colon K \to 2^K$, such that for every
$x \in K$ the image $\Gamma(x)$ is non-empty and convex and the set $\Set{(x, y) | x \in K, y \in \Gamma(x)}$ is
closed. Then there exists a fixed point $x^*$ of $\Gamma$, i.e.,~$x^* \in \Gamma(x^*)$.
\end{theorem}

With this theorem we can prove the following lemma:
\begin{lemma}
For all possible queues $q(\theta_k) \in \IR_{\geq 0}^E$, acyclic edge sets $E^i_{\theta_k} \subseteq E$ and all current inflow rates $b^-(\theta_k) \in \IR_{\geq 0}^{I \times V}$, such that every node $v$ with $b^-_{i, v}(\theta_k) > 0$ can reach $t_i$ within $E^i_{\theta_k}$, there exists an IDE thin flow $(x, a)$.
\end{lemma}
\begin{proof}
For every vector $x \in \IR_{\geq 0}^{I \times E}$ satisfying \eqref{eq:thin_flow_conservation} and \eqref{eq:thin_flow_zero_on_inactive} \revised{there exist} uniquely defined node labels $a \in \IR_{\geq 0}^{I \times V}$ that fulfil \eqref{eq:thin_flow_t} and \eqref{eq:thin_flow_min}. Existence follows since $E^i_{\theta_k}$ is acyclic and the uniqueness follows from the fact that for every $v$ there is a $v$-$t_i$-path within $E^i_{\theta_k}$. This mapping $x \mapsto a$ is continuous. So the only difficult part is to satisfy \eqref{eq:thin_flow_tight}. 

Let $K$ be the set of $x$ vectors satisfying \eqref{eq:thin_flow_conservation} and \eqref{eq:thin_flow_zero_on_inactive}, i.e.,
\[K \coloneqq \Set{x \in \IR_{\geq 0}^{I \times E}| \begin{array}{ll}
\sum_{e \in \delta_v^+} x_{i, e} = b^-_{i, v}(\theta_k) &\text{for all } i \in I \text{ and } v \in V \setminus \set{t_i}\\[4pt]
\;\;\; \text{ and } x_{i, e} = 0 &\text{for all } i \in I \text{ and } e \in E \setminus E^i_{\theta_k}
\end{array}}.\]
Clearly, $K$ is compact, convex and non-empty.

We define a set-valued function $\Gamma: K \to 2^K$ as follows:
\[\Gamma(x) = \Set{y \in K: y_{i, e} = 0 \text{ for all } e \in E^i_{\theta_k} \text{ with } a_{i, v} < \frac{g_e\left(\sum_{j \in I} x_{j, e}\right)}{\nu_e} + a_{i, w}}\]
where $a$ are the label corresponding to $x$. Then $\Gamma(x)$ is non-empty and convex. For non-emptiness note that every node $v$ with $b^-_{i, v} > 0$, except $t_i$, has at least one outgoing edge with $a_{i, v} = \nicefrac{g_e\big(\sum_{j \in I} x_{j, e}\big)}{\nu_e} + a_{i, w}$, so $y$ can send everything into this edge. 
Convexity is clear as well since $x$ determines which edges can be used and which not and these are fixed within $\Gamma(x)$.

Finally, we show that $\set{(x, y)| x \in K, y \in \Gamma(x)}$ is a closed set. \revised{Therefore,} let $(x^n, y^n)_{n \in \IN}$ be a sequence in this set that converges in $\IR^{I \times E}\times \R^{I \times E}$. 
Since $K$ is compact, both sequences separately converge to some points $x$ and $y$ in $K$. 
Let $(a^n)_{n \in \IN}$ be the sequence of associated node labels of $x^n$ and $a$ the node label of $x$. Since $x \mapsto a$ is continuous we have $a = \lim_{n \to \infty} a^n$.
We need to show that $y \in \Gamma(x)$. Suppose for contradiction that there is a commodity $i \in I$ and an $e = vw \in E^i_{\theta_k}$ with $y_{i, e} > 0$ and $a_{i, v} <  \nicefrac{g_e\big(\sum_{j \in I} x_{j, e}\big)}{\nu_e} + a_{i, w}$. 
But since $g_e$ is continuous, there has to be an $n_0 \in \IN$ such that $y_{i, e}^n > 0$ and $a^n_{i, v} <  \nicefrac{g_e\big(\sum_{j \in I} x^n_{j, e}\big)}{\nu_e} + a^n_{i, w}$ for all $n \geq n_0$. 
This is a contradiction to $y^n \in \Gamma(x^n)$.

Hence, by Kakutani's fixed point theorem (\Cref{thm:kakutani}) there exists an $x^* \in K$ with $x^* \in \Gamma(x^*)$, which forms together with the associated node label $a^*$ an IDE thin flow.
\end{proof}

Consider an IDE flow $f$ up to time $\theta_k$ where the inflow rates $f_{i, e}^+$ are right-constant.
Due to the continuity of $q_e$ and $a_{i, v}$ we can determine the active arcs $E^i_{\theta_k}$, as well as, the current node inflows $b^-_{i, v}(\theta_k)$ since the feasibility conditions \eqref{eq:Cont-FlowDefProperties-OpAtCap} and \eqref{eq:FIFO} determine unique outflow rates $f_{i, e}^-(\theta_k)$ for given inflow rates $f_{i, e}^+$ from the past.

In order to extend $f$, we consider an IDE thin flow $(x, a)$ and extend the inflow rates and current shortest path distances for all $i \in I$, $e \in E$ and $v \in V$ by
\[f_{i, e}^+(\theta_k + \xi) \coloneqq x_{i,e} \quad \text{ and } \quad \ell_{i, v}(\theta_k + \xi) \coloneqq \ell_{i, v}(\theta_k) + \xi \cdot a_{i, v} \qquad \text{ for all } \xi \in [0, \alpha).\]
We call this extended flow over time an \emph{$\alpha$-extension}.
 
To ensure that we end up with an IDE up to time $\theta + \alpha$, the following requirements on the size $\alpha$ of the \emph{extension phase} $[\theta_k,\theta_k+\alpha)$ must be satisfied:

First of all, the queues can never be negative, and therefore, the phase ends as soon as a queue depletes:
\begin{equation}
 q_e(\theta_k) + \alpha \cdot \left(\sum_{j \in I} x_{j, e} -\nu_e\right) \geq 0  \qquad \text{ for all } e \in E \text{ with } q_e(\theta_k) > 0.
  \label{eq:alpha_resetting:ide}\end{equation}
 Furthermore, the phase ends as soon as an inactive edge gets active. Since queues can build up on inactive edges as well (due to flow of other commodities), we need to take into account the changing rate of a queue, as well. Hence, for all $i \in I$ and $e = vw \in E \setminus E^i_{\theta_k}$ we have:
  \begin{equation}
  \l_{i, v}(\theta_k)  + \alpha \cdot a_{i, v}  \leq  \tau_e + \frac{q_e(\theta_k)}{\nu_e} + \alpha \cdot \frac{g_e\left(\sum_{j \in I} x_{j, e}\right)}{\nu_e} + \l_{i, w}(\theta_k) + \alpha \cdot a_{i, w}.
    \label{eq:alpha_others:ide}
 \end{equation}
 Finally, the current node inflow should stay constant during a phase:
  \begin{equation}
  b^-_{i, v}(\theta_k + \xi) = b^-_{i, v}(\theta_k) \qquad \text{ for all } i \in I \text{ and } v \in V \setminus \set{t_i} \text{ and all } \xi \in [0, \alpha) \label{eq:alpha_b_constant:ide}
  \end{equation}
  We call $\alpha > 0$ \emph{feasible} if it satisfies \eqref{eq:alpha_resetting:ide}, \eqref{eq:alpha_others:ide} and \eqref{eq:alpha_b_constant:ide}.
  
  It is easy to see that such a feasible $\alpha > 0$ always exists since $\l_{i, v}(\theta_k) < \tau_e + \frac{q_e(\theta_k)}{\nu_e} + \l_{i, w}(\theta_k)$ for all $i \in I$ and $e = vw \in E \setminus E^i_{\theta_k}$. Furthermore, the functions $b^-_{i, v}$ are right-constant, since $f_{i, e}^+$ as well as $u_i$ are right-constant. Since $\tau_e > 0$ for all $e \in E$ we have that $b^-_{i, v}(\theta_k)$ is well-defined and constant on some small interval $[\theta_k, \theta_k + \varepsilon)$.

\begin{lemma}\label{lem:extending_multi_sink_IDEs}
Given an IDE flow up to time $\theta_k$, an IDE thin flow $(x, a)$ at time $\theta_k$ and a feasible $\alpha > 0$. Then the $\alpha$-extension is an IDE flow up to time $\theta_{k + 1} \coloneqq \theta_k + \alpha$ and the extended $\l$-functions denote the current shortest path distances. 
\end{lemma}
    
\begin{proof}
First note, that the feasibility conditions are satisfied, since the outflow rates $f_{i, e}^-$ are exactly defined that way. Furthermore, flow conservation holds since
\ifarxiv
\[
	\sum_{e \in \delta_v^+} f_{i, e}^+(\theta_k + \xi) = \sum_{e \in \delta_v^+} x_{i, e} = b^-_{i, v}(\theta_k) = b^-_{i, v}(\theta_k + \xi) = \sum_{e \in \delta_v^-} f_{i, e}^-(\theta_k + \xi) +  \mathds{1}_{v = s_i} \cdot u_i(\theta_k + \xi)
\]
\fi\ifspringer
\begin{align*}
	\sum_{e \in \delta_v^+} f_{i, e}^+(\theta_k + \xi) &= \sum_{e \in \delta_v^+} x_{i, e} = b^-_{i, v}(\theta_k) = b^-_{i, v}(\theta_k + \xi) \\
						&= \sum_{e \in \delta_v^-} f_{i, e}^-(\theta_k + \xi) +  \mathds{1}_{v = s_i} \cdot u_i(\theta_k + \xi)
\end{align*}
\fi
for all $v \in V \setminus \set{t_i}, i \in I$ and all $\xi \in [0, \alpha)$.

Next we show that the $\l$ labels satisfy \Cref{eq:current_shortest_path_distances}. Given a point in time $\theta_k + \xi$ with $\xi \in [0, \alpha)$ we have by \eqref{eq:queue-dynamic} applied on the total inflow rate $f_e^+(\theta_k + \xi)$ that \[q'_e(\theta_k + \xi) = \left\{\begin{array}{ll}
       f_e^+(\theta_k + \xi) - \nu_e & \text{ if } q_e(\theta_k + \xi) > 0,\\
       \max\Set{f_e^+(\theta + \xi) - \nu_e, 0} & \text{ else,}
       \end{array}\right\} = g_e\left(\textstyle{\sum_{j \in I} x_{j, e}}\right).\]
Note that we have $q_e(\theta_k + \xi) = q_e(\theta_k) + \xi \cdot g_e\left(\textstyle{\sum_{j \in I} x_{j, e}}\right)$ since $q'_e(\theta_k + \xi)$ is constant for $\xi \in [0, \alpha)$.
 
For non-active arcs $e = vw \notin E^i_{\theta_k}$ we have by \eqref{eq:alpha_others:ide} that
\ifarxiv
\begin{align*}\l_{i,v}(\theta_k + \xi) = \l_{i, v}(\theta_k) + \xi \cdot a_{i, v} &\leq \tau_e + \frac{q_e(\theta_k)}{\nu_e} + \xi \cdot \frac{g_e\left(\textstyle{\sum_{j \in I} x_{j, e}}\right)}{\nu_e} + \l_{i, w}(\theta_k) + \xi \cdot a_{i, w}\\
&= \tau_e + \frac{q_e(\theta_k + \xi)}{\nu_e} + \l_{i, w}(\theta_k + \xi).\end{align*}
\fi\ifspringer
\begin{align*}\l_{i,v}(\theta_k + \xi) &= \l_{i, v}(\theta_k) + \xi \cdot a_{i, v} \\
	&\leq \tau_e + \frac{q_e(\theta_k)}{\nu_e} + \xi \cdot \frac{g_e\left(\textstyle{\sum_{j \in I} x_{j, e}}\right)}{\nu_e} + \l_{i, w}(\theta_k) + \xi \cdot a_{i, w}\\
	&= \tau_e + \frac{q_e(\theta_k + \xi)}{\nu_e} + \l_{i, w}(\theta_k + \xi).\end{align*}
\fi

For active arcs $e = vw \in E^i_{\theta_k}$ we have by \eqref{eq:thin_flow_min} that
\ifarxiv
\begin{align*}
\l_{i,v}(\theta_k + \xi) = \l_{i, v}(\theta_k) + \xi \cdot a_{i, v} &\leq \tau_e + \frac{q_e(\theta_k)}{\nu_e} + \l_{i, w}(\theta_k) + \xi \cdot \left(\frac{g_e\left(\sum_{j \in I} x_{j, e}\right)}{\nu_e} + a_{i, w} \right) \\
&= \tau_e + \frac{q_e(\theta_k + \xi)}{\nu_e} + \l_{i, w}(\theta_k + \xi).
\end{align*}
\fi\ifspringer
\begin{align*}
\l_{i,v}(\theta_k + \xi) &= \l_{i, v}(\theta_k) + \xi \cdot a_{i, v} \\
&\leq \tau_e + \frac{q_e(\theta_k)}{\nu_e} + \l_{i, w}(\theta_k) + \xi \cdot \left(\frac{g_e\left(\sum_{j \in I} x_{j, e}\right)}{\nu_e} + a_{i, w} \right) \\
&= \tau_e + \frac{q_e(\theta_k + \xi)}{\nu_e} + \l_{i, w}(\theta_k + \xi).
\end{align*}
\fi

Since there has to be one active arc that satisfies \eqref{eq:thin_flow_min} with equality, the same arc satisfies the inequality above with equality, which shows that \eqref{eq:current_shortest_path_distances} holds.
In other words, the extended $\l$ labels denote the current shortest path distances in the $\alpha$-extension.

Finally, we show that the $\alpha$-extension satisfies the IDE condition \eqref{eq:Cont-FlowDefProperties-OnlyUseSP}: For all $\xi \in [0, \alpha)$ and all arcs $e = vw \in E$ we have that $f_{i,e}^+(\theta_k + \xi) > 0$ implies that $x_{i, e} > 0$, and therefore,
\[a_{i,v}(\theta_k + \xi) = a_{i,v} \stackrel{\text{\eqref{eq:thin_flow_tight}}}{=} \frac{g_e\left(\textstyle{\sum_{j \in I} x_{j, e}}\right)}{\nu_e} + a_{i, w}  = \frac{q'_e(\theta_k +\xi)}{\nu_e} + a_{i,w}(\theta_k)\]
for all $\xi \in [0, \alpha)$.
Hence, 
\[\l_{i, v}(\theta_k + \xi) = \tau_e + \frac{q_e(\theta_k + \xi)}{\nu_e} + \l_{i, w}(\theta_k + \xi)\]
which shows $e \in E^i_{\theta_k + \xi}$.
Hence, the $\alpha$-extension is indeed an IDE flow up to time $\theta_{k + 1} = \theta_k + \alpha$.
\end{proof}

\begin{theorem}
\label{theorem:existence_multiThinFlow}
Consider a multi-source multi-sink network with a finite set of commodities~$I$ and right-constant network inflow functions $u_i$. Then, there exists an IDE flow $f$ with right-constant inflow rate functions $f_{i, e}^+$.
\end{theorem}

The proof is exactly the same as the proof for \Cref{thm:existence} but this time we extend the IDE flow by \Cref{lem:extending_multi_sink_IDEs}.

Since Kakutani's fixed point theorem doesn't give much insight into how to construct such multi-commodity thin flows, we want to give a brief idea how to do this with a mixed integer program.
For this we introduce two different types of boolean decision variables. For every $i \in I$ and $e \in E^i_{\theta_k}$ we have $y_{i,e} \in \set{0, 1}$ and for every $e \in E$ with $q_e(\theta_k) = 0$ we have $z_e \in \set{0, 1}$.

\begin{align*}
 y_{i,e} &= 1 \quad \Leftrightarrow \quad  x_{i,e} = 0, &&\text{ thus, \eqref{eq:thin_flow_tight} does not apply,}\\
z_{e} &= 1 \quad \Leftrightarrow \quad \sum_{j \in I} x_{j, e} - \nu_e \leq 0, &&\text{ thus, } g_e\left(\textstyle{\sum_{j \in I} x_{j, e}}\right) = 0.
\end{align*}

When these decision variables are guessed correctly, the task to find an IDE thin flow is a simple linear program. It remains open if the complete thin flow can be computed efficiently.

\subsection{Arbitrary Network Inflow Rates}
Next, we want to show that multi-sink IDE flows exist even with very general network inflow functions. Recall that the $L^2$-space is defined by 
	\[L^2([a, b)) := \Set{x: [a, b) \to \IR | \int_a^b x(\xi)^2\diff \xi < \infty}\]
for every time interval $[a, b) \subseteq \IR_{\geq 0}$, where two function are equal if they are equal almost everywhere. Together with the scalar product
$\scalar{x}{y} =  \int_a^b x(\xi)\cdot y(\xi) \diff \xi$,
it forms a Hilbert space. If we have vectors of functions $f, g \in L^2([a, b))^d$, for $d \in \IN$, the scalar product is defined as
$\scalar{f}{g} =  \sum_{i=1}^d \int_a^b f_i(\xi)\cdot g_i(\xi) \diff \xi$.
A sequence $f^k \in L^2([a, b))^d$ \emph{converges weakly} to $f \in L^2([a, b))^d$, if $\scalar{f^k}{g} \to \scalar{f}{g}$ for all $g \in L^2([a, b))^d$.
For a subset $K \subseteq L^2([a, b))^d$ we call a mapping $\A\colon K \to L^2([a, b))^d$ \emph{weak-strong-continuous} at $f \in K$, if for every $f^k \in K$ that converges weakly to $f$, we have that $\A(f^k)$ converges to $\A(f)$ with respect to the $L^2$-norm.

\begin{lemma}\label{lemma:ExtendingMultiFlows}
	Given a network with a finite set of commodities $I$ with arbitrary sources and sinks and bounded network inflow functions $u_i$ such that \mbox{$u_i\big|_{[a,b)} \in L^2([a, b))$} for all $a < b$ and $i \in I$.  Let $f$ be an IDE flow up to time $\phi\geq 0$.
	Then, for any $0 < \varepsilon < \min\set{\tau_e | e \in E}$, we can extend $f$ to an IDE flow up to time $\phi+\varepsilon$.
\end{lemma}
In order to prove this, we utilize the following variational inequality:

\paragraph{Variational Inequality.}
Given an interval $[a, b) \subseteq \IR_{\geq 0}$, a number $d \in \IN$, a subset $K \subseteq L^2([a,b))^d$ and a mapping $\A: K \to L^2([a,b))^d$, then the variational inequality $\VI(K,\A)$ is the following:
\begin{equation}\label{eqn:VI}\tag{VI}\text{Find }g \in K \text{ such that } \scalar{\A(g)}{g'-g} \geq 0 \text{ for all } g' \in K.\end{equation}
Conditions to guarantee the existence of such an element $g$ are given by Br\'ezis \cite[Theorem 24]{brezis1968} (see also \cite{Tang2017}):
\begin{theorem} \label{brezis}
Let $K$ be a nonempty, closed, convex and bounded subset of $L^2([a,b))^d$. Let $\A : K \rightarrow L^2([a,b))^d$ be a weak-strong-continuous mapping. Then, the variational inequality $\VI(K,\A)$ has a solution $g^* \in K$.
\end{theorem}

\begin{proof}[Proof of \Cref*{lemma:ExtendingMultiFlows}]
For $\theta \in [\phi, \phi + \varepsilon)$, we define 
	\[
		b^-_{i,v}(\theta) \coloneqq 
			\begin{cases}
				u_i(\theta) + \sum_{e \in \delta^-_v} f_{i,e}^-(\theta), &\text{ if } v=s_i\\
				\sum_{e \in \delta^-_v} f_{i,e}^-(\theta)					&\text{ else,}
			\end{cases}
	\]
where $f_{i,e}^-(\theta)$ is uniquely defined according to \Cref{eq:Cont-FlowDefProperties-OpAtCap,eq:FIFO}. Note that by the choice of $\varepsilon$ these functions $f_{i,e}^-$, and hence $b^-_{i,v}$, do not depend on any inflow function during $[\phi, \phi + \varepsilon)$. 
We set $d = \abs{I} \cdot \abs{E}$ and define $K$ to be the set of all feasible flows over time with FIFO, described by the edge inflow functions only, on the given network during the interval $[\phi, \phi + \varepsilon)$ that satisfy the inflow functions, i.e.,
\[K := \Set{ (g_{i, e})_{i \in I, e \in E} \in L^2([\phi, \phi + \varepsilon))^d | 
\begin{array}{l}
g_{i,e}(\theta) \geq 0,  \sum_{e \in \delta^+(v)}g_{i,e}(\theta) = b^-_{i,v}(\theta) \\
\text{for all } v \neq t_i \text{ and almost all } \theta \in \IR_{\geq 0}
\end{array} }.\]
Note that $K$ is indeed a nonempty, closed, convex and bounded subset of $L^2([\phi, \phi + \varepsilon))^d$.
We define the following mapping $\A : K \rightarrow L^2([\phi, \phi + \varepsilon))^d$ that maps
\[g = (g_{i, e})_{i \in I, e \in E} \mapsto (h_{i, e})_{i \in I, e \in E} \hspace{.5em} \text{ with }\hspace{.5em} h_{i, e}(\theta) := \tau_e + \frac{q_e(\theta)}{\nu_e} + \ell_{i,v}(\theta) - \ell_{i,u}(\theta).\]
Here $e = uv$, $\ell_{i,v}(\theta)$ and $\ell_{i,u}(\theta)$ are the current shortest paths distances from $v$ ($u$, respectively) to $t_i$ at time $\theta$ and $q_e(\theta)$ is the queue length at edge $e$ at time $\theta$ all considering the feasible flow over time $f$ extended by $g$.

This mapping $\A$ is indeed weak-strong-continuous. As shown by Cominetti et\ al.\ \cite[Lemma 4]{CominettiCL15} the mapping $(g_{i, e})_{i \in I, e \in E} \mapsto (q_e)_{e \in E}$ is weak-strong-con\-tinuous and, since $(q_e)_{e \in E} \mapsto (\ell_{i,v})_{i\in I, v \in V}$ is (strong-strong)-continuous, it follows immediately that $\A$ is weak-strong-continuous. 
Applying \Cref{brezis} provides a solution $g^*$ for $\VI(K,\A)$.

We have to show that $f$ extended by $g^*$ is a multi-commodity IDE flow. By the definition of $b_{i,v}^-$ the flow conservation is satisfied. 
Suppose that \cref{eq:Cont-FlowDefProperties-OnlyUseSP} does not hold for almost all $\theta \in [\phi, \phi + \varepsilon)$. Then there is an edge $e$, a commodity $i$, and a set of times $\Theta \subseteq [\phi, \phi + \varepsilon)$ of positive measure, such that $g^*_{e, i}(\theta) > 0$ and $e \not\in E_\theta^i$ for all~$\theta \in \Theta$. It follows that $h^*_{e, i}(\theta) > 0$ for all $\theta \in \Theta$.
Since all functions of $g^*$ and $h^*$ are non-negative we have:
\[\scalar{\A(g^*)}{g^*} \geq \int_{\Theta} h^*_{e, i}(\theta) \cdot g^*_{e, i}(\theta) \diff{\theta} > 0.\]
We define a new flow $g'$ that fulfils $\scalar{\A(g^*)}{g'} = 0$. 
Note that for any flow, and especially for $g^*$ we have the following property:
For every node $v$, every commodity $i$ and every time $\theta$ there exists an outgoing edge $e \in \delta^+_v$ that is active, i.e., $e \in E^i_\theta$. This follows immediately from the fact, that $E^i_\theta$ connects every node $v$ with $t_i$. Furthermore, the sets $\Theta_{i,e} := \set{\theta \in [\phi, \phi + \varepsilon) | e \in E^i_\theta}$ are, by their definition and the continuity of the label functions $\ell_{i,v}$, a union of closed intervals, and therefore measurable. 

We now define $g' \in K$ as follows. At every node $v$, for every commodity $i$ and at every point in time $\theta$, we send all arriving flow at $v$ of commodity $i$ into an edge $e \in E^i_\theta$, where $E^i_\theta$ are the active edges according to $g^*$. It is easy to check that we have $\scalar{\A(g^*)}{g'} = 0$.
Combining these we get
	\[\scalar{h^*}{g'-g^*} = \scalar{h^*}{g'} - \scalar{h^*}{g^*} < 0,\]
which is a contradiction to \eqref{eqn:VI}. To show that  \cref{eq:Cont-FlowDefProperties-OnlyUseSP} is fulfilled for every $\theta$, recall that set $\Theta_0$ of the points in time where this is not satisfied has measure zero. It is possible to modify the edge inflow rates at every $\theta \in \Theta_0$, such that flow conservation and \eqref{eq:Cont-FlowDefProperties-OnlyUseSP} is fulfilled by sending all flow into edges in $E^i_{\theta}$. This has no impact on the queues or the shortest path distances.
\end{proof}

\begin{theorem}
\label{theorem:existence_multi}\label{theorem:existence_arbitrary_inrates}
Consider a multi-source multi-sink network with a finite set of commodities~$I$ and bounded network inflow functions $u_i$ with \mbox{$u_i\big|_{[a,b)} \in L^2([a, b))$} for all~$a < b$ and~$i \in I$. Then there exists an IDE flow $f$ with bounded inflow rate functions $f_{i, e}^+$.
\end{theorem}

\begin{proof}
Starting with the empty flow which is a feasible IDE flow for the empty set $[0, 0)$, we can repeatedly extend it for $\varepsilon := \frac{1}{2}\min\set{\tau_e | e \in E}$ with \Cref{lemma:ExtendingMultiFlows} to obtain a sequence $(f_k)_{k \in \IN}$, where $f_k$ is an IDE flow for $[0, k \cdot
\varepsilon)$. Taking the pointwise limit for every inflow rate function gives us an IDE flow for all times.
Note that $u_i(\theta)$ and $f_{i, e}^-(\theta) \leq \nu_e$ are bounded at every point in time. Hence, the flow conservation constraint implies that the inflow rate functions $f_{i,e}^+$ are bounded as well.
\end{proof}

%!TEX root = ../article.tex
\section{Termination of IDE Flows in Multi-Sink Networks}\label{sec:TerminationMultiSink}

We show that there are instances in which all IDE flows do not terminate. We first observe that while the proofs of \Cref{lemma:TerminationAcyclicNetworks} and \Cref{cor:TerminationIfTotalVolumeLess1} can easily be adapted to the multi-sink case (so it is still true that all flows in an acyclic network and all IDE flows with total volume less than $\tauMinDiff\nu_{\min}$ eventually terminate), this is not true for the proof of \Cref{thm:Termination_SingleSink}.

\begin{theorem}\label{thm:NonTermination}
	There is a multi-source multi-sink network with two sinks and all edge transit times and rate capacities equal to $1$, where any IDE flow does not terminate.
\end{theorem}

To construct such an instance we make use of several gadgets. The first one, gadget $A$, will serve as the main building block and is depicted in \Cref{fig:ProofNonTerm:A}. It consists of two cycles with one common edge $v_1v_2$ and one player $i$ with sink node $t$ (outside the gadget and reachable from the nodes $v_2, v_5$ and $v_7$ via some paths $P_2, P_5$ and $P_7$, respectively) and a constant network inflow rate of $2$ on the interval $[0,1)$ at node $v_1$. Our goal will be to embed this gadget into a larger instance in such a way, that for any IDE flow, the flow associated with player $i$ will exhibit the following flow pattern for all $h \in \IN$ (see \Cref{fig:ProofNonTerm:AFlow}):

\begin{description}
	\item[\textbf{1. On the interval $\bm{[5h,5h+1)}$:}] All flow generated at $v_1$ (for $h=0$) or arriving at $v_1$ (for $h>0$) enters the edge to $v_2$ at a rate of $2$, half of it directly starting to travel along the edge, half of it building up a queue of length $1$ at time $5h+1$.
	\item[\textbf{2. On the interval $\bm{[5h+1,5h+2)}$:}] The flow arriving at node $v_2$ enters the edge to $v_3$ because $v_2,v_3,v_4,v_5,P_5$ is currently the shortest path to $t$. The length of the queue of edge $v_1v_2$ decreases until it reaches $0$ at time $5h+2$.
	\item[\textbf{3. On the interval $\bm{[5h+2,5h+3)}$:}] The flow arriving at node $v_2$ enters the edge to $v_6$ because $v_2,v_6,v_7,P_7$ is currently the shortest path to $t$.
	\item[\textbf{4. On the interval $\bm{[5h+4,5h+5)}$:}] The flows arriving at nodes $v_5$ and $v_7$ enter the respective edges towards node $v_1$ because $v_5,v_1,v_2,P_2$ as well as $v_7,v_1,v_2,P_2$ are currently the shortest paths to get to $t$.
	\item[\textbf{5. On the interval $\bm{[5h+5,5h+6)}$:}] There is a total inflow of $2$ at node $v_1$, which enters the edge to $v_2$. Thus, the pattern repeats.
\end{description}

\begin{figure}[h!]
	\begin{center}
		\begin{adjustbox}{max width=1\textwidth}
			\begin{tikzpicture}
	\newcommand{\colComGreen}{green!90!blue!65}
	
	\newcommand{\graphA}[1]{
		\node[namedVertex] (1) at (4,0) {$v_1$};
		\node[namedVertex] (2) at (4,3) {$v_2$};
		\node[namedVertex] (3) at (2,3) {$v_3$};		
		\node[namedVertex] (4) at (0,1.5) {$v_4$};		
		\node[namedVertex] (5) at (2,0) {$v_5$};		
		\node[namedVertex] (6) at (6,3) {$v_6$};		
		\node[namedVertex] (7) at (6,0) {$v_7$};	
		
		\path[edge] (1) -- (2);
		\path[edge] (2) -- (3);
		\path[edge] (3) -- (4);
		\path[edge] (4) -- (5);
		\path[edge] (5) -- (1);
		\path[edge] (2) -- (6);
		\path[edge] (6) -- (7);
		\path[edge] (7) -- (1);
		
		\path[edge,dashed,<-] (1) -- ++(0,-.8);
		\path[edge,dashed] (2) -- ++(0,.8);
		
		\path[draw,->,thick,dashed] (2) -- ++(.6,-.6);
		\path[draw,->,thick,dashed] (5) -- ++(.6,-.6);
		\path[draw,->,thick,dashed] (7) -- ++(.6,.6);
		
		\node at (0.2,3.4) [rectangle,draw] () {\Large$\theta=#1$:};		
	}
	
	\begin{scope}
		\graphA{0}
		
		\draw (1) +(0,-1) node[rectangle, draw, minimum width=1cm, minimum height=.7cm,fill=\colComGreen] (u) {$u=2$}; 
		\draw [->, line width=2pt] (u) -- (1);
	\end{scope}
	
	\begin{scope}[xshift=8cm]
		\node[rectangle, color=\colComGreen, draw, minimum width=2.1cm,minimum height=.2cm, fill=\colComGreen, rotate around={90:(0,0)}] at (4,1.5) {};
		\node[rectangle, color=\colComGreen, draw, minimum width=.5cm,minimum height=1.5cm, fill=\colComGreen, rotate around={90:(0,0)}] at (3.2,0.7) {};
		\node at (2.5,1.5) [] () {$q_{v_1v_2}(1)=1$};
		
		\graphA{5h+1}
	\end{scope}
	
	\begin{scope}[xshift=16cm]
		\node[rectangle, color=\colComGreen, draw, minimum width=2.1cm,minimum height=.2cm, fill=\colComGreen, rotate around={90:(0,0)}] at (4,1.5) {};
		\node[rectangle, color=\colComGreen, draw, minimum width=1.2cm,minimum height=.2cm, fill=\colComGreen] at (3,3) {};
		
		\graphA{5h+2}
	\end{scope}
	
	\begin{scope}[yshift=-5.5cm]
		\node[rectangle, color=\colComGreen, draw, minimum width=1.6cm,minimum height=.2cm, fill=\colComGreen, rotate around={37:(0,0)}] at (1,2.25) {};
		\node[rectangle, color=\colComGreen, draw, minimum width=1.2cm,minimum height=.2cm, fill=\colComGreen] at (5,3) {};
		
		\graphA{5h+3}
	\end{scope}
	
	\begin{scope}[yshift=-5.5cm, xshift=8cm]
		\node[rectangle, color=\colComGreen, draw, minimum width=1.6cm,minimum height=.2cm, fill=\colComGreen, rotate around={-37:(0,0)}] at (1,0.75) {};
		\node[rectangle, color=\colComGreen, draw, minimum width=2.1cm,minimum height=.2cm, fill=\colComGreen, rotate around={90:(0,0)}] at (6,1.5) {};
		
		\graphA{5h+4}
	\end{scope}
	
	\begin{scope}[yshift=-5.5cm, xshift=16cm]
		\node[rectangle, color=\colComGreen, draw, minimum width=1.2cm,minimum height=.2cm, fill=\colComGreen] at (5,0) {};
		\node[rectangle, color=\colComGreen, draw, minimum width=1.2cm,minimum height=.2cm, fill=\colComGreen] at (3,0) {};
		
		\graphA{5h+5}
	\end{scope}

\end{tikzpicture}
		\end{adjustbox}
	\end{center}
	\vspace{-0.5cm}
	\caption[format=hang]{The desired flow pattern in gadget $A$ at times $\theta=0,1,2,3,4,5,\dots$.}\label{fig:ProofNonTerm:AFlow}
\end{figure}
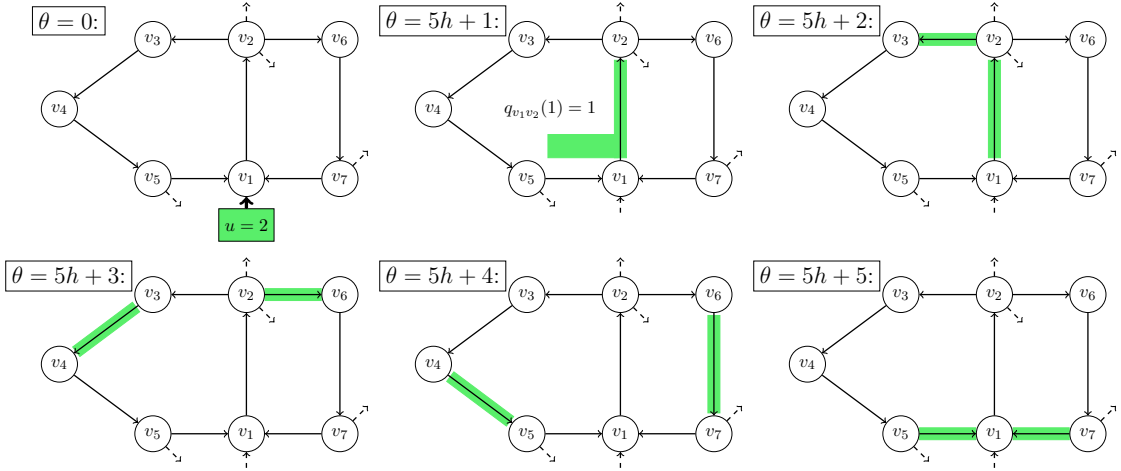

The effect of this behavior is, that other particles outside the gadget, who want to travel through this gadget along the central vertical path, will estimate an additional waiting time as indicated by the diagram displayed inside gadget $A$ in \Cref{fig:ProofNonTerm:A} (next to the vertical red path). Now, in order to actually guarantee the described behavior, we need to embed gadget $A$ into a larger instance in such a way, that for any IDE flow the following assumptions hold:
\begin{enumerate}
	\item The only edges leaving $A$ are the start edges of the four dashed paths indicated in \Cref{fig:ProofNonTerm:A}.
	\item The three (blue) paths $P_2, P_5$ and $P_7$ are of the same length $L$ (w.r.t. $\tau_e$).
	\item For all $h \in \IN$ 
	\begin{itemize}
		\item the unique shortest $v_2$-$t$ path for all $\theta \in [5h+1,5h+2)$ is $v_2,v_3,v_4,v_5,P_5$,
		\item the unique shortest $v_2$-$t$ path for all $\theta \in (5h+2,5h+3]$ is $v_2,v_6,v_7,P_7$,
		\item the unique shortest $v_5$-$t$ path for all $\theta \in [5h+4,5h+5]$ is $v_5,v_1,v_2,P_2$ and
		\item the unique shortest $v_7$-$t$ path for all $\theta \in [5h+4,5h+5]$ is $v_7,v_1,v_2,P_2$.
	\end{itemize}
\end{enumerate} 

Note that at time $\theta = 5h+2$ we do not require that there is only one unique $v_2$-$t$ path. This is due to the fact that waiting times always change continuously and therefore when the shortest $v_2$-$t$ path changes from one path to another, there needs to be a time where both paths are equally long. Thus, there cannot always be a unique shortest path. However, this does not influence the overall flow pattern, since those discrete points in time form a set of measure zero and thus only allow for flow of volume zero to escape the overall flow pattern.

In order to satisfy the assumptions 1.-3., we will now construct three types of gadgets $B_2, B_5$ and $B_7$ for the three paths $P_2, P_5$ and $P_7$, each of equal length and on which any IDE flow induces waiting times as shown by the respective diagrams on the right side in \Cref{fig:ProofNonTerm:A}.

\begin{figure}[h!]\centering
	\begin{adjustbox}{max width=.8\textwidth}
		\begin{tikzpicture}
	\node () at (0,4) {\LARGE $A$};
	
	\draw (-0.75,-1.55) rectangle (6.75,5);
	
	\node[namedVertex] (1) at (4,0) {$v_1$};
	\node[namedVertex] (2) at (4,4) {$v_2$};
	\node[namedVertex] (3) at (2,4) {$v_3$};		
	\node[namedVertex] (4) at (0,2) {$v_4$};		
	\node[namedVertex] (5) at (2,0) {$v_5$};		
	\node[namedVertex] (6) at (6,4) {$v_6$};		
	\node[namedVertex] (7) at (6,0) {$v_7$};	
	
	\path[edge,color=red] (1) -- (2);
	\path[edge] (2) -- (3);
	\path[edge] (3) -- (4);
	\path[edge] (4) -- (5);
	\path[edge] (5) -- (1);
	\path[edge] (2) -- (6);
	\path[edge] (6) -- (7);
	\path[edge] (7) -- (1);
	
	\path[edge,dashed,color=red] (4,-2) -- (1);
	\path[edge,dashed,color=red] (2) -- ++(0,1.5);
	
	\node[namedVertex,dashed,color=blue,fill=white] (t) at (14,1) {t};
	
	\path[draw,-,thick,dashed,color=blue] (2) -- ++(1,-1) -- node[pos=1,right](P2) {$P_2$} ++(2.5,0);
	\path[draw,-,thick,dashed,color=blue] (5) -- ++(0,-1) -- node[pos=1,right](P5) {$P_5$} ++(5.5,0);
	\path[draw,-,thick,dashed,color=blue] (7) -- node[pos=1,right](P7) {$P_7$} ++(1.5,1);		
	
	\path[draw,-,thick,dashed,color=blue] (P2) -- node[pos=1,right](P2plot) {
		\begin{tikzpicture}[anchor=center,scale=.6,solid,black,
		declare function={
			f(\x)= and(\x>0, \x<=1) * (4*\x)   
			+ and(\x>1, \x<=3) * (4)
			+ and(\x>3, \x<=4) * (16-4*\x);
		}
		]
		\begin{axis}[xmin=0,xmax=5,xtick={0,1,2,3,4,5}, xticklabels={$0$,$1$,$2$,$3$,$4$,$5$}, ymax=4.1,samples=500,width=6cm,height=4cm,font=\Large]
		\addplot[blue, ultra thick,domain=0:5] {f(x)};
		\end{axis}
		\end{tikzpicture}
	} ++(.43,0);
	\path[edge,dashed,color=blue] (P2plot) -- ++(3,0) -- (t);
	\path[draw,-,thick,,dashed,color=blue] (P5) -- node[pos=1,right](P5plot) {
		\begin{tikzpicture}[anchor=center,scale=.6,solid,black,
		declare function={
			f(\x)= and(\x>=0, \x<=1) * (3-3*\x)   
			+ and(\x>3, \x<=4) * (-9+3*\x)
			+ and(\x>4, \x<=5) * (3)
			+ and(\x>5, \x<=6) * (18-3*\x);
		}
		]
		\begin{axis}[xmin=1,xmax=6,xtick={1,2,3,4,5,6}, xticklabels={$1$,$2$,$3$,$4$,$5$,$6$},ymax=4.1,samples=500,width=6cm,height=4cm,font=\Large]
		\addplot[blue, ultra thick,domain=1:6] {f(x)};
		\end{axis}
		\end{tikzpicture}
	} ++(.96,0);
	\path[edge,dashed,color=blue] (P5plot) -- ++(2.5,0) -- (t);
	\path[draw,-,thick,,dashed,color=blue] (P7) -- node[pos=1,right](P7plot) {
		\begin{tikzpicture}[anchor=center,scale=.6,solid,black,
		declare function={
			f(\x)= and(\x>=0, \x<=3) * (3-\x)   
			+ and(\x>3, \x<=4) * (-9+3*\x)
			+ and(\x>4, \x<=5) * (3)
			+ and(\x>5, \x<=8) * (8-\x);
		}
		]
		\begin{axis}[xmin=1,xmax=6,xtick={1,2,3,4,5,6}, xticklabels={$1$,$2$,$3$,$4$,$5$,$6$},ymax=4.1,samples=500,width=6cm,height=4cm,font=\Large]
		\addplot[blue, ultra thick,domain=1:6] {f(x)};
		\end{axis}
		\end{tikzpicture}
	} ++(.96,0);
	\path[edge,dashed,color=blue] (P7plot) -- ++(2.5,0) -- (t);
	
	\node () at (3,2) {
		\begin{tikzpicture}[anchor=center,scale=.6,rotate=90,
		declare function={
			f(\x)= and(\x>0, \x<=1) * (\x)   
			+ and(\x>1, \x<=2) * (2-\x)
			+ and(\x>5, \x<=6) * (-5+\x) 
			+ and(\x>6, \x<=7) * (7-\x);
		}
		]
		\begin{axis}[xmin=0,xmax=5,ymax=1,samples=500,width=6cm,height=4cm,font=\Large]
		\addplot[red, ultra thick,domain=0:5] {f(x)};
		\end{axis}
		\end{tikzpicture}
	};
\end{tikzpicture}
	\end{adjustbox}
	\caption{Gadget $A$ (the dashed paths and nodes are not part of the gadget). The (red) diagram inside the box $A$ indicates the waiting time on edge $v_1v_2$ (and therefore on the (red) vertical path through the gadget), provided that the flow originating inside this gadget follows the flow pattern indicated in \Cref{fig:ProofNonTerm:AFlow}. The (blue) diagrams on the right indicate the desired waiting times on the paths $P_2, P_5$ and $P_7$, respectively, which in turn ensure that the flow inside the gadget does indeed follow the desired flow pattern.}\label{fig:ProofNonTerm:A}
\end{figure}
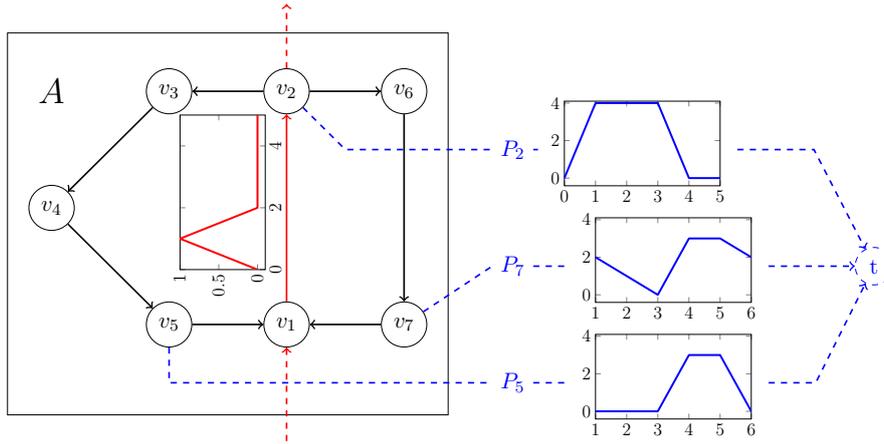

To build these gadgets we need time shifted versions of gadget $A$, which we denote by $A^{+k}$. Such a gadget is constructed the same way as gadget $A$ above, with the only difference that the support of the network inflow rate function $u_i$ is shifted to the interval $[k \mod 5, k\mod 5 +1)$. Gadget $B_2$ now consists of the concatenation of four gadgets of type $A^{+0}$, four gadgets of type $A^{+1}$ and four gadgets of type $A^{+2}$ in series along their vertical paths through them with three edges between each two gadgets (see \Cref{fig:ProofNonTerm:B2}). 

\begin{figure}[h!]\centering
	\floatbox[{\capbeside\thisfloatsetup{capbesideposition={right,top},capbesidewidth=.30\textwidth}}]{figure}[\FBwidth]
	{\caption{Gadget $B_2$ consisting of four copies of each of the types $A^{+0},A^{+1},A^{+2}$. The diagram inside the box of gadget $B_2$ indicates the waiting time on the vertical path through gadget $B_2$, provided that within all of the used gadgets $A$, the flow follows the flow pattern from \Cref{fig:ProofNonTerm:AFlow}. The dashed parts are not part of the gadgets and only sketch how this gadget needs to be embedded in a larger instance.}\label{fig:ProofNonTerm:B2}}
	{
		\begin{adjustbox}{max width=.6\textwidth}
			\begin{tikzpicture}
	\draw (-4.5,-.5) rectangle (3.5,11.5);
	
	\node () at (-3.5,10.5) {\LARGE $B_2$};
	
	\node (1) at (-2.5,5) {
		\begin{tikzpicture}[anchor=center,scale=1.2,rotate=90,
		declare function={
			f(\x)= and(\x>0, \x<=1) * (4*\x)   
			+ and(\x>1, \x<=3) * (4)
			+ and(\x>3, \x<=4) * (16-4*\x);
		}
		]
		\begin{axis}[xmin=0,xmax=5,ymax=4.1,samples=500,width=8cm,height=4cm]
		\addplot[red, ultra thick,domain=0:5] {f(x)};
		\end{axis}
		\end{tikzpicture}
	};
	
	\node () at (1,1) {\LARGE $A^{+0}$};
	
	\draw (0,0) rectangle ++(3,2);

	\node (1) at (2,1) {
		\begin{tikzpicture}[anchor=center,scale=.3,rotate=90,
		declare function={
			f(\x)= and(\x>0, \x<=1) * (\x)   
			+ and(\x>1, \x<=2) * (2-\x)
			+ and(\x>5, \x<=6) * (-5+\x) 
			+ and(\x>6, \x<=7) * (7-\x);
		}
		]
		\begin{axis}[xmin=0,xmax=5,ymax=1,samples=500,width=6cm,height=4cm,xtick=\empty,ytick=\empty]
		\addplot[red, ultra thick,domain=0:5] {f(x)};
		\end{axis}
		\end{tikzpicture}
	};
	
	\node[vertex,red](1a) at (2,2.5){};
	
	\node[vertex,red](1b) at (2,3.5){};
	
	\node () at (1,5) {\LARGE $A^{+0}$};
	
	\draw (0,4) rectangle ++(3,2);
	
	\node (2) at (2,5) {
		\begin{tikzpicture}[anchor=center,scale=.3,rotate=90,
		declare function={
			f(\x)= and(\x>0, \x<=1) * (\x)   
			+ and(\x>1, \x<=2) * (2-\x)
			+ and(\x>5, \x<=6) * (-5+\x) 
			+ and(\x>6, \x<=7) * (7-\x);
		}
		]
		\begin{axis}[xmin=0,xmax=5,ymax=1,samples=500,width=6cm,height=4cm,xtick=\empty,ytick=\empty]
		\addplot[red, ultra thick,domain=0:5] {f(x)};
		\end{axis}
		\end{tikzpicture}
	};
	
	\node[vertex,red](2a) at (2,6.5){};
	
	\node (3) at (1.5,7.5) {\LARGE $\vdots$};
	
	\node[vertex,red](3a) at (2,8.5){};
	
	\node () at (1,10) {\LARGE $A^{+2}$};
	
	\draw (0,9) rectangle ++(3,2);

	\node (4) at (2,10) {
		\begin{tikzpicture}[anchor=center,scale=.3,rotate=90,
		declare function={
			f(\x)= and(\x>2, \x<=3) * (-2+\x)   
			+ and(\x>3, \x<=4) * (4-\x);
		}
		]
		\begin{axis}[xmin=0,xmax=5,ymax=1,samples=500,width=6cm,height=4cm,xtick=\empty,ytick=\empty]
		\addplot[red, ultra thick,domain=0:5] {f(x)};
		\end{axis}
		\end{tikzpicture}
	};
	
	\draw[draw,<-,thick,red,dashed] (1) -- ++(0,-2);
	
	\draw[edge,red] (1) -- (1a);
	\draw[edge,red] (1a) -- (1b);
	\draw[edge,red] (1b) -- (2);
	
	\draw[edge,red] (2) -- (2a);
	\draw[draw,-,thick,red,dotted] (2a) -- ++(0,.5);
	
	\draw[draw,-,thick,red,dotted] (3a) -- ++(0,-.5);
	\draw[edge,red] (3a) -- (4);
	
	\draw[edge,thick,red,dashed] (4) -- ++(0,2);
	
	\node[namedVertex,dashed,color=blue,fill=white] (t) at (10,5) {t};
	
	\node[vertex,dashed,blue](P2) at (6,4) {}; 
	\node[vertex,dashed,blue](P7) at (6,3) {}; 
	\node[vertex,dashed,blue](P5) at (6,2) {}; 
	
	\draw[edge,blue,dashed] (1.north east) -- ++(2,0) -- (P2);
	\draw[edge,blue,dashed] (2.north east) -- ++(2,0) -- (P2);
	\draw[edge,blue,dashed] (P2) -- ++(2,0) node[pos=.5,above]() {\LARGE$P_2$} -- (t);
	
	\draw[edge,blue,dashed] (1.east) -- ++(2,0) -- (P7);
	\draw[edge,blue,dashed] (2.east) -- ++(2,0) -- (P7);
	\draw[edge,blue,dashed] (P7) -- ++(2,0) node[pos=.5,above]() {\LARGE$P_7$} -- (t);
	
	\draw[edge,blue,dashed] (1.south east) -- ++(2,0) -- (P5);
	\draw[edge,blue,dashed] (2.south east) -- ++(2,0) -- (P5);
	\draw[edge,blue,dashed] (P5) -- ++(2,0) node[pos=.5,above]() {\LARGE$P_5$} -- (t);

	\node[vertex,blue,dashed](P2b) at (6,10) {}; 
	\node[vertex,blue,dashed](P7b) at (6,9) {}; 
	\node[vertex,blue,dashed](P5b) at (6,8) {}; 
	
	\draw[edge,blue,dashed] (4.north east) -- ++(2,0) -- (P2b);
	\draw[edge,blue,dashed] (P2b) -- ++(2,0) node[pos=.5,above]() {\LARGE$P_2^{+2}$} -- (t);
	
	\draw[edge,blue,dashed] (4.east) -- ++(2,0) -- (P7b);
	\draw[edge,blue,dashed] (P7b) -- ++(2,0) node[pos=.5,above]() {\LARGE$P_7^{+2}$} -- (t);
	
	\draw[edge,blue,dashed] (4.south east) -- ++(2,0) -- (P5b);
	\draw[edge,blue,dashed] (P5b) -- ++(2,0) node[pos=.5,above]() {\LARGE$P_5^{+2}$} -- (t);
	
\end{tikzpicture}	
		\end{adjustbox}
	}
\end{figure}
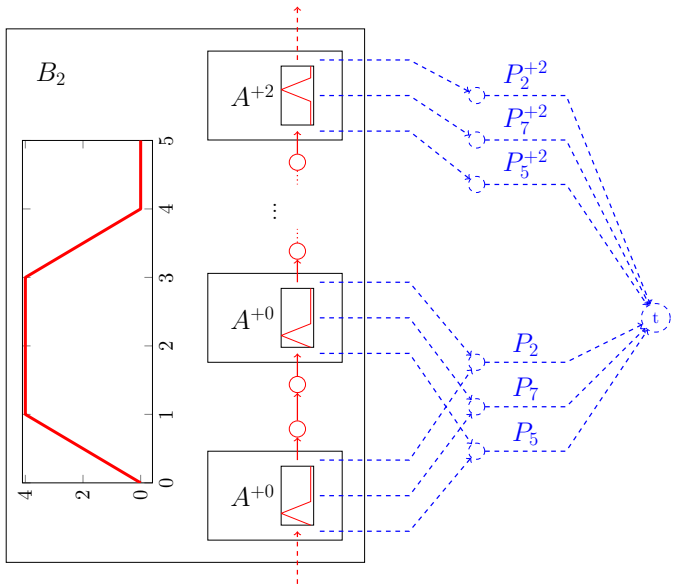	

Similarly, gadget $B_5$ consists of three copies of $A^{+3}$-type gadgets, three copies of $A^{+4}$-type gadgets and additional $6\cdot 4$ edges to ensure that the vertical path has the same length as the one of gadget $B_2$. Finally, gadget $B_7$ consists of three copies of $A^{+3}$-type gadgets, three copies of $A^{+4}$-type gadgets, two copies of $A^{+5}$-type gadgets, one copy of $A^{+6}$-type gadgets and additional $3\cdot 4$ edges to ensure that the vertical path has the same length as the one of gadget $B_2$. 

We again use the notation $B_j^{+k}$ to refer to a time shifted version of gadget $B_j$ -- i.e. with all used gadgets $A$ shifted by additional $k$ time steps. Next, we build a gadget $C$ by just taking one copy of each $B_j^{+k}$ for all $j \in \set{2,5,7}$ and $k=0,1,2,3,4$ (see \Cref{fig:ProofNonTerm:C}).

\begin{figure}[h]\centering
	\begin{adjustbox}{max width=.9\textwidth}
		\begin{tikzpicture}
	\draw (-1,0) rectangle (15,6.5);
	
	\node () at (-0.3,5.5) {\LARGE $C$};
	
	\draw (.5,1) rectangle ++(2,5);
	
	\node () at (1,5.5) {\large$B_2$};
	
	\node (B2) at (1.5,3.3) {
		\begin{tikzpicture}[anchor=center,scale=.5,rotate=90,
		declare function={
			f(\x)= and(\x>0, \x<=1) * (4*\x)   
			+ and(\x>1, \x<=3) * (4)
			+ and(\x>3, \x<=4) * (16-4*\x);
		}
		]
		\begin{axis}[xmin=0,xmax=5,ymax=4.1,samples=500,width=9cm,height=4cm,xtick=\empty,ytick=\empty]
		\addplot[red, ultra thick,domain=0:5] {f(x)};
		\end{axis}
		\end{tikzpicture}
	};
	
	\draw[edge,red,dashed] (B2) -- ++(0,3.8);
	\draw[edge,<-,red,dashed] (B2) -- ++(0,-3.8);

	\draw (4,1) rectangle ++(2,5);
	
	\node () at (4.5,5.5) {\large$B_2^{+1}$};
	
	\node (B21) at (5,3.3) {
		\begin{tikzpicture}[anchor=center,scale=.5,rotate=90,
		declare function={
			f(\x)= and(\x>=1, \x<=2) * (-4+4*\x)   
			+ and(\x>2, \x<=4) * (4)
			+ and(\x>4, \x<=5) * (20-4*\x);
		}
		]
		\begin{axis}[xmin=0,xmax=5,ymax=4.1,samples=500,width=9cm,height=4cm,xtick=\empty,ytick=\empty]
		\addplot[red, ultra thick,domain=0:5] {f(x)};
		\end{axis}
		\end{tikzpicture}
	};
	
	\draw[edge,red,dashed] (B21) -- ++(0,3.8);
	\draw[edge,<-,red,dashed] (B21) -- ++(0,-3.8);
	
	\node () at (8,3.3) {\LARGE $\dots$};
	
	\draw (10,1) rectangle ++(2,5);
	
	\node () at (10.5,5.5) {\large$B_7^{+4}$};
	
	\node (B74) at (11,3.3) {
		\begin{tikzpicture}[anchor=center,scale=.5,rotate=90,
		declare function={
			f(\x) = and(\x>=0, \x<=2) * (2-\x)   
			+ and(\x>2, \x<=3) * (-6+3*\x)
			+ and(\x>3, \x<=4) * (3)
			+ and(\x>4, \x<=7) * (7-\x);
		}
		]
		\begin{axis}[xmin=0,xmax=5,ymax=4.1,samples=500,width=9cm,height=4cm,xtick=\empty,ytick=\empty]
		\addplot[red, ultra thick,domain=0:5] {f(x)};
		\end{axis}
		\end{tikzpicture}
	};
	
	\draw[edge,red,dashed] (B74) -- ++(0,3.8);
	\draw[edge,<-,red,dashed] (B74) -- ++(0,-3.8);

	\node[namedVertex,blue,inner sep=2pt,fill=white](v20) at (14,1) {$v_{2,0}$};
	\node[namedVertex,blue,inner sep=2pt,fill=white](v50) at (14,2.2) {$v_{5,0}$};
	\node[namedVertex,blue,inner sep=2pt,fill=white](v70) at (14,3.4) {$v_{7,0}$};
	\node[namedVertex,blue,inner sep=2pt,fill=white](v21) at (14,4.6) {$v_{2,1}$};
	\node[blue]() at (14,5.8) {$\bm\vdots$};
	
	\node(B2-P20) at (3.3,.25) {};
	\node(B2-P21) at (3.7,.5) {};
	\node(B21-P21) at (6.3,.75) {};
	
	\draw[draw,-,blue,thick] (2.5,1.5) -- ++(0.3,0);
	\draw[draw,-,blue,thick] (2.5,1.9) -- ++(0.3,0);
	\draw[draw,-,blue,thick] (2.5,2.3) -| (B2-P20.center);
	\draw[edge,blue,thick] (B2-P20.center) -- (13,.25) |- (v20);
	
	\draw[draw,-,blue,thick] (2.5,3.1) -- ++(0.3,0);
	\draw[draw,-,blue,thick] (2.5,3.5) -- ++(0.3,0);
	\draw[draw,-,blue,thick] (2.5,3.9) -| (B2-P21.center);
	\draw[edge,blue,thick] (B2-P21.center) -- (12.75,.5) |- (v21);
	
	\draw[draw,-,blue,thick] (2.5,4.7) -- ++(0.3,0);
	\draw[draw,-,blue,thick] (2.5,5.1) -- ++(0.3,0);
	\draw[draw,-,blue,thick] (2.5,5.5) -- ++(0.3,0);
	
	\draw[draw,-,blue,thick] (6,1.5) -| (B21-P21.center);
	\draw[edge,blue,thick] (B21-P21.center) -- (12.5,.75) |- (v21);

	\node[namedVertex,dashed,color=blue,fill=white] (t) at (20,3.4) {t};
	\draw[edge,dashed,blue] (v20) -- node[above,pos=.8]() {\Large$P_2$} ++(4,0) -- (t);
	\draw[edge,dashed,blue] (v50) -- node[above,pos=.8]() {\Large$P_5$} ++(4,0) -- (t);
	\draw[edge,dashed,blue] (v70) -- node[above,pos=.8]() {\Large$P_7$} ++(4,0) -- (t);
	\draw[edge,dashed,blue] (v21) -- node[above,pos=.8]() {\Large$P_2^{+1}$} ++(4,0) -- (t);
\end{tikzpicture}
	\end{adjustbox}
	\caption{Gadget $C$}\label{fig:ProofNonTerm:C}
\end{figure}
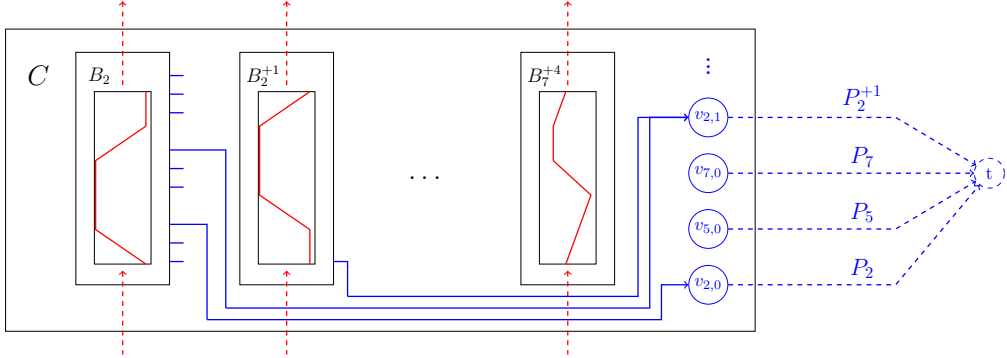	

Finally, taking two copies of this gadget, $C$ and $C'$, and two additional nodes, $t$ and $t'$, where $t$ will be the sink node for all players in $C$ and $t'$ the sink node for all players in $C'$, we can build our entire graph as indicated by \Cref{fig:ProofNonTerm:Graph}. We connect the top edges of the gadgets $B_j^{+k}$ in gadget $C'$ with the sink $t$ and use those gadgets' respective vertical paths as the $P_j^{+k}$ paths for gadget $C$ and vice versa. 

\begin{figure}[h]\centering
	\begin{adjustbox}{max width=.9\textwidth}
		\begin{tikzpicture}
	\draw (0,0) rectangle ++(9,3);
	\node () at (0.5,2.5) {\LARGE $C$};
	
	\draw (1,.5) rectangle ++(1,2);
	\node () at (1.5,1.5) {\Large$B_2$};
	
	\node () at (3,1.5) {\Large$\cdots$};
	
	\draw (4,.5) rectangle ++(1,2);
	\node () at (4.5,1.5) {\Large$B_5$};
	
	\node () at (6,1.5) {\Large$\cdots$};
	
	\node[namedVertex,red,fill=white] (ts) at (4.5,4.5) {$t'$};	
	
	\draw[edge,red] (1.5,2.5) -- ++(0,1) -- (ts);
	\draw[edge,red] (4.5,2.5) -- ++(0,1) -- (ts);
	
	\node[vertex,blue] (v20) at (8.5,.5) {};
	\node[vertex,blue] (v50) at (8.5,1) {};
	\node[blue] () at (8.5,2) {\large$\bm\vdots$};
	
	\draw[edge,blue] (v20) -| ++(1,-.8) -| ++(3,.8);
	\draw[edge,blue] (v50) -| ++(2,-1.6) -| ++(5,1.1);
	
	\draw (11,0) rectangle ++(9,3);
	\node () at (11.5,2.5) {\LARGE $C'$};
	
	\draw (12,.5) rectangle ++(1,2);
	\node () at (12.5,1.5) {\Large$B_2$};
	
	\node () at (14,1.5) {\Large$\cdots$};
	
	\draw (15,.5) rectangle ++(1,2);
	\node () at (15.5,1.5) {\Large$B_5$};
	
	\node () at (17,1.5) {\Large$\cdots$};
	
	\node[namedVertex,blue,fill=white] (t) at (15.5,4.5) {$t$};	
	
	\draw[edge,blue] (12.5,2.5) -- ++(0,1) -- (t);
	\draw[edge,blue] (15.5,2.5) -- ++(0,1) -- (t);
	
	\node[vertex,red] (v20s) at (19.5,.5) {};
	\node[vertex,red] (v50s) at (19.5,1) {};
	\node[red] () at (19.5,2) {\large$\bm\vdots$};
	
	\draw[edge,red] (v20s) -| ++(1,-1.4) -| ++(-19,1.4);
	\draw[edge,red] (v50s) -| ++(1.5,-2.2) -| ++(-16.5,1.7);
	
\end{tikzpicture}		
	\end{adjustbox}
	\caption{The whole graph}\label{fig:ProofNonTerm:Graph}
\end{figure}
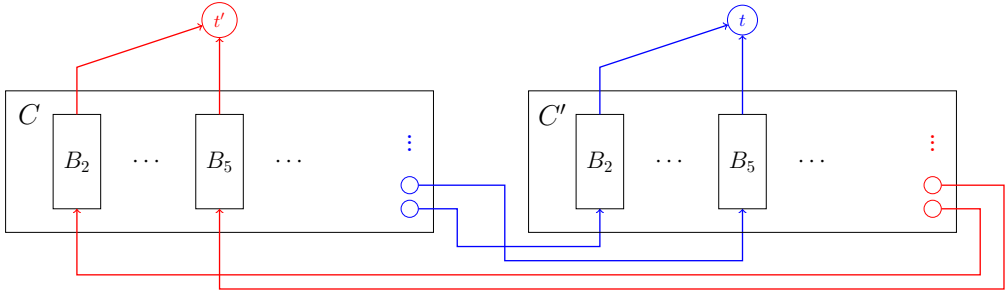

In order to prove the correctness of our construction (i.e. that any IDE flow on this instance does not terminate) we need the following important observation:

\begin{obs}\label{obs:HenOrEggGadgetA}
	If a flow in some $A^{+k}$-type gadget (with $k \in \set{0,1,2,3,4}$) follows the desired flow pattern for all unit time intervals between $k$ and some $\theta \in \IN_0, \theta \geq k$, the induced waiting time on edge $v_1v_2$ of this gadget (and therefore on the vertical path through this gadget) will follow the waiting time function indicated by the diagram in \Cref{fig:ProofNonTerm:A} (shifted by $k$) for the next unit time interval $[\theta,\theta+1)$, independent of the evolution of the flow in this interval.
\end{obs}

\begin{proof}[Proof of \Cref*{obs:HenOrEggGadgetA}]
	If $(\theta-k) \equiv 0 \mod 5$, then over the following interval $[\theta,\theta+1)$ we will have an inflow of 2 at node $v_1$. Either because it originates here (if $\theta=k$) or because (by assumption) the flow pattern already holds for the previous unit time interval and, thus, flow has entered edges $v_5v_1$ and $v_7v_1$, both at rate $1$, in this interval. Since the edge $v_1v_2$ is the only one leaving $v_1$, all flow will enter this edge and thereby starting to build up a queue of length $1$ at time $\theta+1$.
	
	If $(\theta-k) \equiv 1 \mod 5$, we start with a queue of length $1$ at edge $v_1v_2$, which linearly decreases to $0$ over the course of the interval as no new flow arrives at $v_1$ and flow leaves the edge at $v_2$ at rate $1$. As all edges leaving $v_2$ are currently empty, no new queues can form, regardless of which edge the flow actually uses.
	
	In all other cases, we start with empty queues at all edges. As no flow arrives at node $v_1$ and at all other nodes flow arrives with at most rate $1$ (since all other nodes only have one incoming edge with rate capacity $1$), no new queues can form.
\end{proof}

\begin{cor}\label{cor:HenOrEggGadgetsB}
	If a flow in some $B_j^{+k}$-type gadget follows the desired flow pattern for all unit time intervals up to some $\theta \in \IN_0$, the induced waiting time on the vertical path through this gadget will follow the waiting time function indicated by the respective diagram on the right in \Cref{fig:ProofNonTerm:A} (shifted by $k$) for the next unit time interval $[\theta,\theta+1)$, independent of the evolution of the flow in this interval. \qedhere
\end{cor}

We now want to prove that any IDE flow in the constructed instance does not terminate. To do that we will take a generic $A^{+k}$-type gadget from this instance and show by induction that the flow originating in $v_1$ of this gadget will follow the flow pattern described at the beginning of the construction and indicated in \Cref{fig:ProofNonTerm:AFlow}.

\begin{proof}[Proof of \Cref*{thm:NonTermination}]
	Let $\tilde{A}$ be a copy of gadget $A$, time shifted by some $k \in \set{0,1,2,3,4}$ and w.l.o.g. in $C$. We then need to show that all flow in this gadget will obey the pattern described above (time shifted by $k$) for all unit time intervals between $k$ and $\theta$ for all $\theta \in \IN_0, \theta \geq k$. As our induction basis we take $\theta=k$, for which the claim trivially holds.
	
	For our induction step we assume that the flow in this gadget (and therefore all other gadgets $A$ in the instance) follows the desired pattern for all unit time intervals up to some $\theta \in \IN_0, \theta \geq k$ and want to show that it also does so for the next unit time interval $[\theta,\theta+1)$. 
	
	\begin{description}
		\item[\textbf{Case 1: $\bm{(\theta-k) \equiv 0 \mod 5}$}] See the respective case in the proof of \Cref{obs:HenOrEggGadgetA}.
		
		\item[\textbf{Case 2: $\bm{(\theta-k) \equiv 1 \mod 5}$}] Over the following unit time interval, flow will arrive at node $v_2$ at a rate of $1$ while the queue on edge $v_1v_2$ decreases. We need to show that all arriving flow will enter edge $v_2v_3$, as $v_2,v_3,v_4,v_5,P_5$ is currently the shortest $v_2$-$t$ path (w.r.t. instantaneous travel time). By induction and \Cref{cor:HenOrEggGadgetsB}, we already know that all paths $P_j^{+k'}$ will exhibit the waiting time pattern indicated by the respective diagram on the right in \Cref{fig:ProofNonTerm:A} (shifted by $k$). As, by construction, all those paths have a common length $L$ (w.r.t. $\tau_e$), we can calculate the length of all possible $v_2$-$t$ paths:
		\begin{itemize}
			\item The path $v_2,v_3,v_4,v_5,P_5$ has length $3+L$ and an additional waiting time of $0$.
			\item The path $v_2,P_2$ has length $L$ and an additional waiting time of $4$.
			\item The path $v_2,v_6,v_7,P_7$ has length $2+L$ and an additional waiting time between $2$ and $1$.
			\item All paths leaving gadget $\tilde{A}$ through the vertical path have a length of at least $3+1+L$ (three edges between the current gadget and the next one, one edge through this gadget and $L$ edges for whatever gadget $B_j^{+k'}$ is finally used to get to $t$) and possibly additional waiting times.
		\end{itemize}
		So the path beginning with edge $v_2v_3$ has the shortest total instantaneous travel time and even uniquely so for all times except $\theta+1$. Therefore all flow arriving at $v_2$ must enter edge $v_2v_3$ for the whole interval $[\theta,\theta+1)$.
		
		\item[\textbf{Case 3: $\bm{(\theta-k) \equiv 2 \mod 5}$}] Over the following unit time interval, flow will arrive at rate $1$ at node $v_3$, which has to enter edge $v_3v_4$ as this is the only one leaving $v_3$, and at node $v_2$. We now need to show that all this flow enters edge $v_2v_6$ (except possibly at time $\theta$). We will do this in the same way as in case 2, i.e. by calculating the instantaneous travel times for all relevant paths with the help of \Cref{cor:HenOrEggGadgetsB}:
		\begin{itemize}
			\item The path $v_2,v_3,v_4,v_5,P_5$ has length $3+L$ and no additional waiting time.
			\item The path $v_2,P_2$ has length $L$ and an additional waiting time of $4$.
			\item The path $v_2,v_6,v_7,P_7$ has length $2+L$ and an additional waiting time between $1$ and $0$.
			\item All paths leaving gadget $\tilde{A}$ immediately again have a length of at least $4+L$.
		\end{itemize}
		So all flow must enter edge $v_2v_6$ as the path beginning with this edge is the unique shortest $v_2-t$ path.
		
		\item[\textbf{Case 4: $\bm{(\theta-k) \equiv 3 \mod 5}$}] Over the following unit time interval, flow arrives at rate $1$ at the nodes $v_4$ and $v_6$. Since those only have one edge leaving them, the flow will just follow the only possible path.
		
		\item[\textbf{Case 5: $\bm{(\theta-k) \equiv 4 \mod 5}$}] Over the following unit time interval, flow arrives at rate $1$ at the nodes $v_5$ and $v_7$. We need to show that all this flow enters the edges $v_5v_1$ and $v_7v_1$, respectively. So as in case 2 and 3 we need to calculate the instantaneous travel times on all relevant $v_5$-$t$ and $v_7$-$t$ paths - again using \Cref{cor:HenOrEggGadgetsB}:
		\begin{itemize}
			\item The path $v_5,v_2,P_2$ has length $2+L$ and an additional waiting time of $0$.
			\item The path $v_5,P_5$ has length $L$ and an additional waiting time of $3$.
			\item The path $v_7,v_2,P_2$ has length $2+L$ and an additional waiting time of $0$.
			\item The path $v_7,P_7$ has length $L$ and an additional waiting time of $3$.
		\end{itemize}
	\end{description}
	
	With the induction completed we have shown that for any IDE flow and any copy of gadget $A$ within the given instance all flow generated at $v_1$ will arrive back at its start node after five unit time steps at which point the whole network is in exactly the same state as before. Thus, every IDE flow cycles and does not terminate.
\end{proof}

Although we made use of several distinct source nodes, this is in fact not necessary in order to get a network with non-terminating IDEs. I.e. \Cref{thm:NonTermination} can be strengthened as follows.

\begin{theorem}\label{thm:NonTerminationSingleSource}
	There exists a single-source multi-sink network with two sinks, where any IDE flow does not terminate.
\end{theorem}

\ifarxiv
\begin{proof}
	To prove this theorem we extend the network from \Cref{thm:NonTermination} in such a way that after some initial warm-up time any IDE flow behaves exactly as in the initial network. First, we make sure that all network inflow rate functions $u_i$ have the interval $[0,1)$ as their support. This can be accomplished by introducing a new source node $\tilde{s}_i$ for every commodity $i$ and connecting that new source node with an edge of length $r_i$ and rate capacity $2$ to the original source node $s_i$ (see \Cref{fig:ProofNonTermSingleSink:SameReleaseTime}). Also note that the two sink nodes $t$ and $t'$ have no outgoing edges.
	
	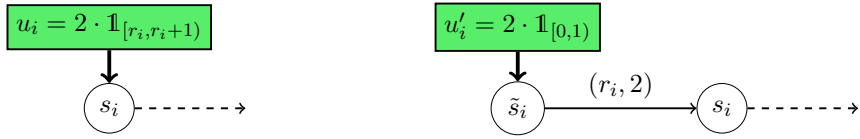
\begin{figure}[h]\centering
		\begin{adjustbox}{max width=.8\textwidth}
			\begin{tikzpicture}
	\newcommand{\colComGreen}{green!90!blue!65}
	
	\useasboundingbox (-2,-.5) rectangle (11,3);
	
	\begin{scope}
		\node[namedVertex](si) at (0,0) {$s_i$};
		\path[edge,dashed] (si) -- +(2,0);
		
		\node[rectangle, color=black, thick, draw, minimum width=.8cm,minimum height=.6cm, fill=\colComGreen] (inflow) at ($(si)+(0,1.2)$) {$u_i = 2 \cdot \mathds{1}_{[r_i, r_i+1)}$};
		
		\path[edge,ultra thick] (inflow) -- (si);
	\end{scope}

	\begin{scope}[xshift=9cm]
		\node[namedVertex](si) at (0,0) {$s_i$};
		\path[edge,dashed] (si) -- +(2,0);
		\node[namedVertex](ssi) at (-3,0) {$\tilde{s}_i$};
		\path[edge] (ssi) --node[above]{$(r_i,2)$} (si);
		
		\node[rectangle, color=black, thick, draw, minimum width=.8cm,minimum height=.6cm, fill=\colComGreen] (inflow) at ($(ssi)+(0,1.2)$) {$u'_i = 2 \cdot \mathds{1}_{[0, 1)}$};
		
		\path[edge,ultra thick] (inflow) -- (ssi);
	\end{scope}			
\end{tikzpicture}			
		\end{adjustbox}
		\caption{Changing the graph from the proof of \Cref{thm:NonTermination} (left) in such a way that all release times are within the interval $[0,1)$ (right).}\label{fig:ProofNonTermSingleSink:SameReleaseTime}
	\end{figure}
	
	Next, we add one node $\hat{s}$, which will be our super source. Then for every commodity  $i$ we add two distinct nodes $v_i, w_i$ and edges as indicated in \Cref{fig:ProofNonTermSingleSink:NewEdges} (where $w_i$ is always connected to $t_i$, the sink node of commodity $i$). Finally, we replace the commodities $I$ by two new commodities $0$ and $0'$ with common source node $\hat{s}$. Commodity $0$ has $t$ as its sink node and the following network inflow rate function
		
		\[\hat{u}_0(\theta) = \begin{cases}\sum_{\substack{i \in I\\t_i = t}}(\tau(P_i)+5), &\theta \in [0,1)\\0, &\text{else}\end{cases},\]
		
	where for every commodity $i \in I$ with sink node $t$, $P_i$ is a shortest $\tilde{s}_i$-$t$ path in $G$. Commodity $0'$ has $t'$ as its sink node and an analogous inflow rate function $\hat{u}_{0'}$.
	
	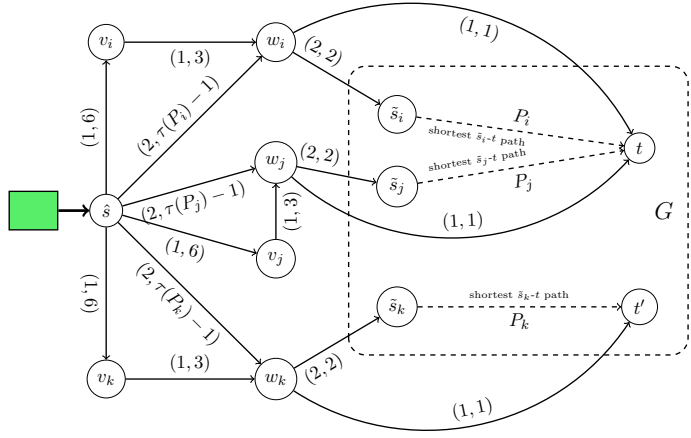
\begin{figure}[h]\centering
		\begin{adjustbox}{max width=.7\textwidth}
			\begin{tikzpicture}
	\newcommand{\colComGreen}{green!90!blue!65}
		
	\useasboundingbox (-10.5,-6) rectangle (5.5,4);
	
	\draw[rounded corners=10, thick, dashed] (-3,-4) rectangle (4, 2);
	\node() at (3.5,-1) {\Large$G$};
	
	\node[namedVertex](t) at (3,0.3) {$t$};
	\node[namedVertex](ts) at (3,-3) {$t'$};
	
	\node[namedVertex](si) at (-2,1) {$\tilde{s}_i$};
	\path[edge,dashed] (si) --node[above,sloped]{$P_i$}node[below,sloped,pos=.3]{\tiny shortest $\tilde{s}_i$-$t$ path} (t);
	\node[namedVertex](sj) at (-2,-0.5) {$\tilde{s}_j$};
	\path[edge,dashed] (sj) --node[below,sloped]{$P_j$}node[above,sloped,pos=.3]{\tiny shortest $\tilde{s}_j$-$t$ path} (t);
	\node[namedVertex](sk) at (-2,-3) {$\tilde{s}_k$};
	\path[edge,dashed] (sk) --node[below,sloped]{$P_k$}node[above,sloped]{\tiny shortest $\tilde{s}_k$-$t$ path} (ts);
	
	\node[namedVertex](shat) at (-8,-1) {$\hat{s}$};

	\node[namedVertex](vi) at (-8,2.5) {$v_i$};
	\node[namedVertex](wi) at (-4.5,2.5) {$w_i$};

	\node[namedVertex](vj) at (-4.5,-2) {$v_j$};
	\node[namedVertex](wj) at (-4.5,0) {$w_j$};	

	\node[namedVertex](vk) at (-8,-4.5) {$v_k$};
	\node[namedVertex](wk) at (-4.5,-4.5) {$w_k$};	
	
	\path[edge] (shat) --node[above,sloped]{$(1,6)$} (vi);
	\path[edge] (vi) --node[below,sloped]{$(1,3)$} (wi);
	\path[edge] (shat) --node[above,sloped]{$(2,\tau(P_i)-1)$} (wi);
	\path[edge] (wi) --node[above,sloped,near start]{$(2,2)$} (si);
	\path[edge] (wi) to [bend left=45] node[below,sloped]{$(1,1)$} (t);
	\path[edge] (shat) --node[below,sloped]{$(1,6)$} (vj);
	\path[edge] (vj) --node[below,sloped]{$(1,3)$} (wj);
	\path[edge] (shat) --node[below,sloped]{$(2,\tau(P_j)-1)$} (wj);
	\path[edge] (wj) --node[above,sloped,near start]{$(2,2)$} (sj);
	\path[edge] (wj) to [bend right=45] node[above,sloped]{$(1,1)$} (t);
	\path[edge] (shat) --node[below,sloped]{$(1,6)$} (vk);
	\path[edge] (vk) --node[above,sloped]{$(1,3)$} (wk);
	\path[edge] (shat) --node[below,sloped]{$(2,\tau(P_k)-1)$} (wk);
	\path[edge] (wk) --node[below,sloped,near start]{$(2,2)$} (sk);
	\path[edge] (wk) to [bend right=45] node[above,sloped]{$(1,1)$} (ts);
	
	\node[rectangle, color=black, thick, draw, minimum width=1cm,minimum height=.8cm, fill=\colComGreen] (inflow) at ($(shat)+(-1.5,0)$) {};
	
	\path[edge,ultra thick] (inflow) -- (shat);
	
\end{tikzpicture}			
		\end{adjustbox}
		\caption{The modified network with only one single source and two sinks. After an initial warm-up phase (see \Cref{fig:ProofNonTermSingleSink:FlowEvolution}) the flow inside $G$ behaves just as in the original network from the proof of \Cref{thm:NonTermination}.}\label{fig:ProofNonTermSingleSink:NewEdges}
	\end{figure}
	
	Now at time $\theta=0$, the shortest $\hat{s}$-$t$ paths are $\hat{s},w_i,t$ and $\hat{s},v_i,w_i,t$ (for $i$ with $t_i=t$) with length $3$ since every other path has to go through $G$ and therefore has a length of at least $4$. As the inflow rate of commodity $0$ exactly matches the total rate capacity of the first edges of all those paths, the flow of commodity splits between these edges $\hat{s},w_i$ and $\hat{s},v_i$ using all of them with full capacity.

	At time $\theta=1$, the flow arriving at node $v_i$ with rate $6$ enters the edge $v_iw_i$ and starts forming a queue on that edge. At time $\theta=2$, flow arrives with rate $\tau(P_i)+2$ at node $w_i$ and starts entering the edge $w_i,\hat{t}_i$, building up a queue there. At time $\theta=3$, this queue has reached a length of $\tau(P_i)+\revised{1}$, at which point the paths $w_i,t$ and $w_i,\tilde{s}_i,P_i,t$ have the same instantaneous travel time (note that by this time no flow has yet entered $G$ and, thus, no waiting times occur within $G$). Thus, from now on the flow arriving at node $w_i$ at rate $3$ splits between the edges $w_i,t$ and $w_i,\tilde{s}_i$ proportional to the respective rate capacities. Thus, between time $\theta=5$ and $\theta=6$, flow arrives at node $\tilde{s}_i$ at a rate of $2$. 
	
		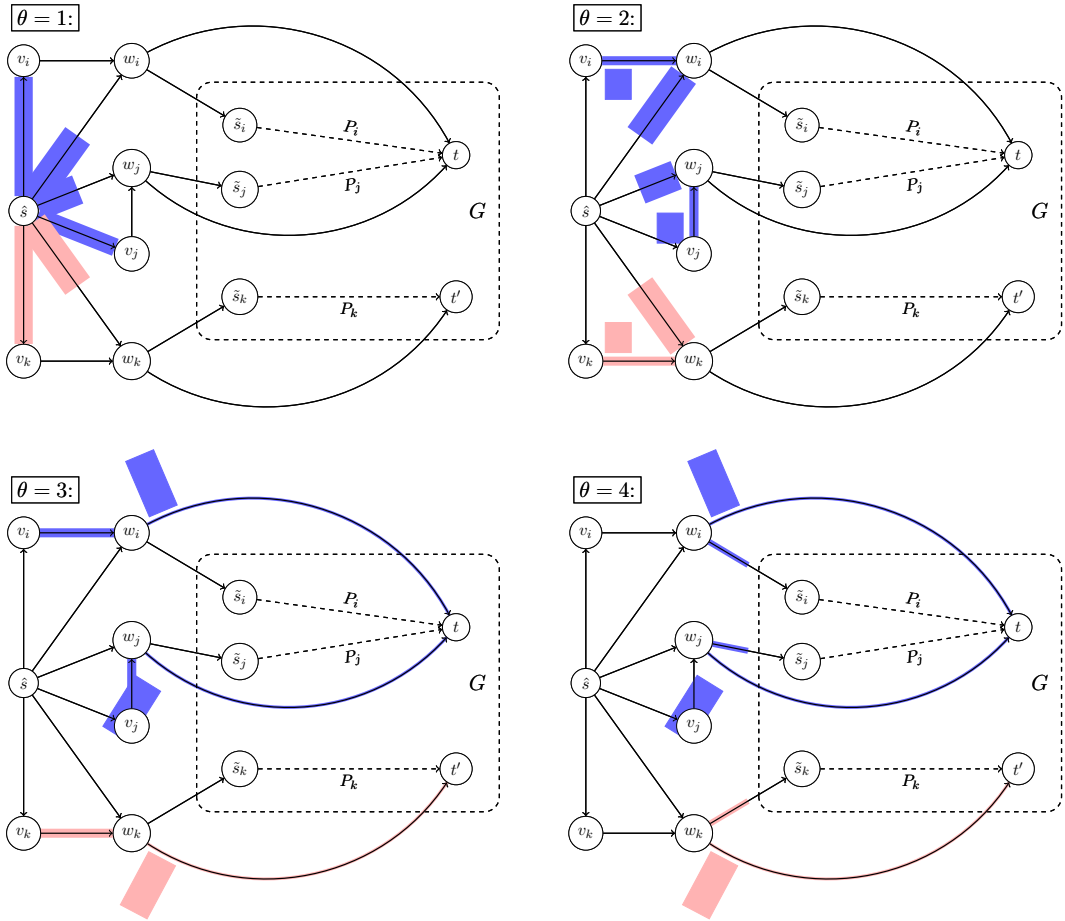
\begin{figure}[h]\centering
			\begin{adjustbox}{max width=.95\textwidth}
				\begin{tikzpicture}
	\newcommand{\colComBlue}{blue!60}
	\newcommand{\colComRed}{red!30}	
	\newcommand{\colComGreen}{green!90!blue!65}

	\newcommand{\network}[1]{
		\draw[rounded corners=10, thick, dashed] (-3,-4) rectangle (4, 2);
		\node() at (3.5,-1) {\Large$G$};
		
		\node[namedVertex](t) at (3,0.3) {$t$};
		\node[namedVertex](ts) at (3,-3) {$t'$};
		
		\node[namedVertex](si) at (-2,1) {$\tilde{s}_i$};
		\path[edge,dashed] (si) --node[above,sloped]{$P_i$} (t);
		\node[namedVertex](sj) at (-2,-0.5) {$\tilde{s}_j$};
		\path[edge,dashed] (sj) --node[below,sloped]{$P_j$} (t);
		\node[namedVertex](sk) at (-2,-3) {$\tilde{s}_k$};
		\path[edge,dashed] (sk) --node[below,sloped]{$P_k$} (ts);
		
		\node[namedVertex](shat) at (-7,-1) {$\hat{s}$};
		
		\node[namedVertex](vi) at (-7,2.5) {$v_i$};
		\node[namedVertex](wi) at (-4.5,2.5) {$w_i$};
		
		\node[namedVertex](vj) at (-4.5,-2) {$v_j$};
		\node[namedVertex](wj) at (-4.5,0) {$w_j$};	
		
		\node[namedVertex](vk) at (-7,-4.5) {$v_k$};
		\node[namedVertex](wk) at (-4.5,-4.5) {$w_k$};	
		
		\path[edge] (shat) -- (vi);
		\path[edge] (vi) -- (wi);
		\path[edge] (shat) -- (wi);
		\path[edge] (wi) -- (si);
		\path[edge] (wi) to [bend left=45] (t);
		\path[edge] (shat) -- (vj);
		\path[edge] (vj) -- (wj);
		\path[edge] (shat) -- (wj);
		\path[edge] (wj) -- (sj);
		\path[edge] (wj) to [bend right=45] (t);
		\path[edge] (shat) -- (vk);
		\path[edge] (vk) -- (wk);
		\path[edge] (shat) -- (wk);
		\path[edge] (wk) -- (sk);
		\path[edge] (wk) to [bend right=45] (ts);
		
		\node at (-6.5,3.5) [rectangle,draw] () {\Large$\theta=#1$:};		
	}

	%\useasboundingbox (-9.5,-6) rectangle (5.5,4);

	\begin{scope}
		\network{1}
		
		\path[draw,line width=12pt,color=\colComBlue] (shat) -- (vi);
		\path[draw,line width=20pt,color=\colComBlue] (shat) -- ($(shat)!.5!(wi)$);
		\path[draw,line width=12pt,color=\colComBlue] (shat) -- (vj);
		\path[draw,line width=20pt,color=\colComBlue] (shat) -- ($(shat)!.5!(wj)$);
		\path[draw,line width=12pt,color=\colComRed] (shat) -- (vk);
		\path[draw,line width=20pt,color=\colComRed] (shat) -- ($(shat)!.5!(wk)$);
		
		\network{1}
	\end{scope}
	
	\begin{scope}[xshift=13cm]
		\network{2}
		
		\path[draw,line width=6pt,color=\colComBlue] (vi) -- (wi);
		\node[rectangle, color=\colComBlue, draw, minimum width=.7cm,minimum height=.6cm, fill=\colComBlue, rotate around={90:(0,0)}] at ($(vi)!.3!(wi)!-.55cm!90:(wi)$) {};
		\path[draw,line width=20pt,color=\colComBlue] ($(shat)!.5!(wi)$) -- (wi);

		\path[draw,line width=6pt,color=\colComBlue] (vj) -- (wj);		
		\node[rectangle, color=\colComBlue, draw, minimum width=.7cm,minimum height=.6cm, fill=\colComBlue, rotate around={90:(0,0)}] at ($(vj)!.3!(wj)!-.55cm!270:(wj)$) {};
		\path[draw,line width=20pt,color=\colComBlue] ($(shat)!.5!(wj)$) -- (wj);

		\path[draw,line width=6pt,color=\colComRed] (vk) -- (wk);		
		\node[rectangle, color=\colComRed, draw, minimum width=.7cm,minimum height=.6cm, fill=\colComRed, rotate around={90:(0,0)}] at ($(vk)!.3!(wk)!-.55cm!270:(wk)$) {};
		\path[draw,line width=20pt,color=\colComRed] ($(shat)!.5!(wk)$) -- (wk);
		
		\network{2}
	\end{scope}
	
	\begin{scope}[yshift=-11cm]
		\network{3}
		
		\path[draw,line width=6pt,color=\colComBlue] (vi) -- (wi);
		\path[draw,line width=2pt,color=\colComBlue] (wi) to [bend left=45] (t);
		\node[rectangle, color=\colComBlue, draw, minimum width=.7cm,minimum height=1.4cm, fill=\colComBlue, rotate around={23:(0,0)}] at ($(shat)!1.28!(wi)!-.3cm!90:(wi)$) {};
		
		\path[draw,line width=6pt,color=\colComBlue] (vj) -- (wj);
		\path[draw,line width=2pt,color=\colComBlue] (wj) to [bend right=45] (t);
		\node[rectangle, color=\colComBlue, draw, minimum width=.7cm,minimum height=1.4cm, fill=\colComBlue, rotate around={-32:(0,0)}] at ($(shat)!1.00001!(wj)!-1.6cm!-90:(wj)$) {};

		\path[draw,line width=6pt,color=\colComRed] (vk) -- (wk);
		\path[draw,line width=2pt,color=\colComRed] (wk) to [bend right=45] (ts);
		\node[rectangle, color=\colComRed, draw, minimum width=.7cm,minimum height=1.4cm, fill=\colComRed, rotate around={-28:(0,0)}] at ($(shat)!1.28!(wk)!-.4cm!-90:(wk)$) {};
		
		\network{3}
	\end{scope}
	
	\begin{scope}[yshift=-11cm, xshift=13cm]
		\network{4}
		
		\path[draw,line width=4pt,color=\colComBlue] (wi) -- ($(wi)!.5!(si)$);
		\path[draw,line width=2pt,color=\colComBlue] (wi) to [bend left=45] (t);
		\node[rectangle, color=\colComBlue, draw, minimum width=.7cm,minimum height=1.4cm, fill=\colComBlue, rotate around={23:(0,0)}] at ($(shat)!1.28!(wi)!-.3cm!90:(wi)$) {};
		
		\path[draw,line width=4pt,color=\colComBlue] (wj) -- ($(wj)!.5!(sj)$);
		\path[draw,line width=2pt,color=\colComBlue] (wj) to [bend right=45] (t);
		\node[rectangle, color=\colComBlue, draw, minimum width=.7cm,minimum height=1.4cm, fill=\colComBlue, rotate around={-32:(0,0)}] at ($(shat)!1.00001!(wj)!-1.6cm!-90:(wj)$) {};

		\path[draw,line width=4pt,color=\colComRed] (wk) -- ($(wk)!.5!(sk)$);
		\path[draw,line width=2pt,color=\colComRed] (wk) to [bend right=45] (ts);
		\node[rectangle, color=\colComRed, draw, minimum width=.7cm,minimum height=1.4cm, fill=\colComRed, rotate around={-28:(0,0)}] at ($(shat)!1.28!(wk)!-.4cm!-90:(wk)$) {};
		
		\network{4}
	\end{scope}
\end{tikzpicture}			
			\end{adjustbox}
			\caption{The evolution of any IDE flow in the modified network. Between times $5$ and $6$, flow will arrive at a rate of $2$ at all nodes $\tilde{s_i}$, while the flow on the edges $w_it$ or $w_it'$ does not interfere with the flow inside $G$ from now on.}\label{fig:ProofNonTermSingleSink:FlowEvolution}
		\end{figure}
	
	As the same flow evolution happens for commodity $0'$, at time $\theta=5$, the new network is in the same state as the original network at time $\theta=0$. In particular, all the flow entering $G$ at one of the nodes $\tilde{s}_i$ will stay inside the network forever and so the flow will never terminate.
\end{proof}
\fi

%!TEX root = ../article.tex

\section{Summary and Open Problems}
We introduced in this paper the concept of IDE flows and
investigated two key questions: existence and termination of
IDE flows.
Regarding the former, we gave in \Cref{sec:ExistenceSingleSink} an extension-algorithm leading to the existence of IDE flows
for single-sink instances.  While the extension-algorithm is constructive,
it is not clear if finitely many calls of the algorithm suffice to
compute an IDE -- at the moment existence relies on a limit argument.
Especially for restricted graph classes (series-parallel graphs or acyclic
graphs) we expect finiteness.  
For multi-source multi-sink instances,
we gave in \Cref{sec:ExistenceMultiSink} a general existence 
theorem -- also based on an extension property. 
For piece-wise constant network inflow rate functions this extension can be achieved by solving a set of equations (called IDE thin flows), which we can obtain algorithmically by solving a linear mixed-integer program. 
The extension property for general network inflow rates, however, relies on a
solution to a variational inequality rendering it non-constructive.

Regarding termination of IDE flows, we showed in \Cref{thm:Termination_SingleSink} that for single-sink networks, all IDE flows terminate. In a forthcoming paper \cite{GrafHarks19}, we give quantitative upper bounds of $\BigO(U\tau(G))$ for the time such an IDE flow needs to terminate, where $U \coloneqq \sum_{i \in I}\int_{0}^{\infty}u_i(\theta)d\theta$ is the total network inflow and $\tau(G) \coloneqq \sum_{e \in E}\tau_e$ the sum of all edge transit times. On the other hand, we can give lower bounds of order $\Theta(U\log\tau(G))$, see \cite{GrafHarks19} for details. 
In \Cref{sec:TerminationMultiSink}, we gave an example for a multi-sink network, where no IDE flow terminates. By \Cref{thm:NonTerminationSingleSource}, we know that only a single-source and two sinks are needed for this effect to appear (and by \Cref{thm:Termination_SingleSink} we also know that this is the minimal number of sinks necessary). However, the underlying graph is quite complex and it would be interesting to see, whether there are certain graph classes  beside acyclic ones, where termination is guaranteed even in the multi-sink case (e.g., planar graphs, series parallel graphs, \dots).
%\end{itemize}

\paragraph*{Acknowledgments:}\ We thank the anonymous reviewers for their careful proofreading and suggestions on how to improve this paper. 
	We also thank Kathrin Gimmi for proofreading the first drafts as well as Marcel Gabor, Johannes Hagenmaier and Georg Kraus for their helpful comments on previous versions of this paper. Finally, we thank the Deutsche Forschungsgemeinschaft (DFG) for their financial support.

\clearpage
\bibliographystyle{plain}
\bibliography{literatur}
%\bibliography{master-bib}
%
\appendix
\revised{\section{List of Symbols}

\begin{tabular}{llp{8cm}}
	\textbf{Symbol}			& \textbf{Name}
		& \textbf{Description} \\\hline
	$u,v,w \in V$				& vertices 
		& \\
	$e \in E$				& edges 
		& \\
	$\delta_v^+ \subseteq E$			& outgoing edges
		& the set of edges leaving $v$, \newline i.e. $\delta_v^+ \coloneqq \set{e \in E | e = vw}$ \\ 
	$\delta_v^- \subseteq E$			& incoming edges 
		& the set of edges entering $v$, \newline i.e. $\delta_v^- \coloneqq \set{e \in E | e = uv}$ \\
	$\tau_e \in \IRnn$ 				& physical travel time 
		& the physical travel time on edge $e$\\
	$\nu_e \in \IRnn$ 				& edge capacity
		& the capacity of edge $e$\\
	$i \in I$				& commodities 
		& the set of all commodities \\
	$s_i \in V$					& source (node)		
		& the node at which particles of commodity $i$ enter the network \\
	$t_i \in V$					& sink (node)
		& the destination of particles of commodity $i$ \\
	$\theta \in \IR_{\geq 0}$ 				& time
		& the time is usually denoted by $\theta$ \\
	$u_i(\theta)\in \IR_{\geq 0}$			& network inflow rate
		& inflow rate of particles of commodity $i$ at node $s_i$ at time $\theta$ \\
	$f^+_{e,i}(\theta)\in \IR_{\geq 0}$ 	& (edge) inflow rate
		& the rate at which particles of commodity $i$ enter edge $e$ at time $\theta$, $f_e^+$ is the total inflow rate, aggregated over all commodities, i.e. $f^+_{e}(\theta) \coloneqq \sum_{i \in I}f^+_{e,i}(\theta)$ \\
	$f^-_{e,i}(\theta)\in \IR_{\geq 0}$ 	& (edge) outflow rate
		& the rate at which particles of commodity $i$ leave edge $e$ at time $\theta$, $f_e^-$ is the total outflow rate, aggregated over all commodities, i.e. $f^-_{e}(\theta) \coloneqq \sum_{i \in I}f^-_{e,i}(\theta)$\\
	$F^+_{e,i}(\theta) \in \IR_{\geq 0}$ & cumulative (edge) inflow 
		& the total volume of flow having entered edge $e$ up to time $\theta$, i.e. $F_{e,i}^+(\theta) \coloneqq \int_{0}^{\theta}f_{e,i}^+(\zeta)d\zeta$. For the aggregated variant we use  $F_{e}^+(\theta) \coloneqq \int_{0}^{\theta}f_{e}^+(\zeta)d\zeta$ \\ 
	$F^-_{e,i}(\theta) \in \IR_{\geq 0}$ & cumulative (edge) outflow 
		& the total volume of flow having left edge $e$ up to time $\theta$, i.e. $F_{e,i}^-(\theta) \coloneqq \int_{0}^{\theta}f_{e,i}^-(\zeta)d\zeta$. For the aggregated variant we use $F_{e}^-(\theta) \coloneqq \int_{0}^{\theta}f_{e}^-(\zeta)d\zeta$\\
	$q_e(\theta) \in \IR_{\geq 0}$ 			& queue length 
		& the length of the queue on edge $e$ at time $\theta$, defined as $q_e(\theta) \coloneqq F^+_e(\theta)-F^-_e(\theta+\tau_e)$ \\
	$c_e(\theta) \in \IRnn$			& instantaneous travel time 
		& the current or instantaneous travel time at time $\theta$ over edge $e$, defined as $c_e(\theta) \coloneqq \tau_e + q_e(\theta)/\nu_e$ \\
	$\ell_{i,v}(\theta) \in \IR_{\geq 0}$	& node labels
		& the current distance from $v$ to $t_i$ at time $\theta$, i.e. the length of a shortest $v$-$t_i$ path w.r.t. $c_e(\theta)$. If all commodities share a common sink, we write $\ell_{v}(\theta) \coloneqq \ell_{i,v}(\theta)$ \\
	$E_{\theta}^i \subseteq E$ 			& active edges 
		& the set of all edges active for commodity $i$ at time $\theta$, i.e. \newline $E_{\theta}^i \coloneqq \set{vw \in E | \ell_{i, v}(\theta) = \ell_{i,w}+c_{vw}(\theta)}$ \\ 
\end{tabular}}

\end{document}